\def\draft{0}\def\sigconf{0}\def\big{0}\def\anon{0}\def\masterthesis{0}\def\cryptology{0}
\ifnum\big=1
    \documentclass[final,17pt]{extarticle} \usepackage{fullpage}
\else
    \ifnum\sigconf=1
        \ifnum\draft=1
            \documentclass[sigconf,draft,capitalise]{acmart}
        \else
            \documentclass[sigconf,final,capitalise]{acmart}
        \fi
    \else
        \ifnum\draft=1
           \documentclass[11pt,draft,a4paper]{article}
        \else
            \ifnum\masterthesis=1
              \documentclass[11pt,final,a4paper,titlepage]{article}
            \else
                \documentclass[11pt,final,a4paper]{article}
            \fi
        \fi
        \usepackage{lmodern}
        \pdfoutput=1
    \fi
\fi

\input{macros}

\usepackage{cuted}

\newcommand{\keygen}{\ensuremath{\mathsf{key\textit{-}gen}}}

\renewcommand{\sign}{\ensuremath{\mathsf{sign}}}
\newcommand{\mint}{\ensuremath{\mathsf{mint}}}

\tikzset{
    pil/.style={
        ->,
       thick,
       shorten <=2pt,
       shorten >=2pt
    }
}
 
\begin{document}

\title{Semi-Quantum Money
%\ifdraft{\\(working draft)}
}

\ifnum\anon=0
    \ifnum\cryptology=1
        \author{Roy Radian}
        \affil{Computer Science Department, Ben-Gurion University, Israel\\
                royrad@post.bgu.ac.il}
        \author{Or Sattath}
        \affil{Computer Science Department, Ben-Gurion University, Israel\\
                sattath@post.bgu.ac.il}
    \else
        \ifnum\sigconf=0
            \author[1]{Roy Radian}
            \author[1]{Or Sattath}
            \affil[1]{Computer Science Department, Ben-Gurion University of the Negev}
        \else
            \author{Roy Radian}
            \author{Or Sattath}
    
            \affiliation{%
            \institution{Computer Science Department, Ben-Gurion University of the Negev}
            \country{Israel}}
        \fi
    \fi
\else
    \ifnum\sigconf=0
        \author{}
    \fi
\fi

\ifnum\sigconf=0
    \ifnum\masterthesis=0
        \maketitle
    \fi
\fi
\renewcommand\Authands{ and }

\ifnum\masterthesis=1
    \begin{titlepage}
        \centering
        { Ben-Gurion University of the Negev}
        
        {The Faculty of Natural Sciences}
        
        {\small The Department of Computer Science}
        
        \vspace{2cm}
        
        {\Large \bfseries Semi-Quantum Money}
        
        \vspace{2cm}
        
        {\small Thesis submitted in partial fulfillment of the requirements for the Master of Sciences degree}
        
        \vspace{1cm}
        
        {\bfseries Roy Radian}
        
        {Under the supervision of Dr. Or Sattath}
        
        \vspace{2cm}
        
        \today
    \end{titlepage}
    
    \begin{titlepage}
        \centering
        { Ben-Gurion University of the Negev}
        
        {The Faculty of Natural Sciences}
        
        {\small The Department of Computer Science}
        
        \vspace{2cm}
        
        {\Large \bfseries Semi-Quantum Money}
        
        \vspace{2cm}
        
        {\small Thesis submitted in partial fulfillment of the requirements for the Master of Sciences degree}
        
        \vspace{1cm}
        
        {\bfseries Roy Radian}
        
        {Under the supervision of Dr. Or Sattath}
        
        \vspace{1cm}
        
        {\small Signature of student: \longunderline Date: \longunderline}
        
        \vspace{0.5cm}
        
        {\small Signature of supervisor: \longunderline Date: \longunderline}
        
        \vspace{0.5cm}
        
        \begin{changemargin}{-1cm}{-1cm}
        \centering
            {\small Signature of the committee for graduate studies: \longunderline Date: \longunderline}
        \end{changemargin}
        \vspace{2cm}
        
        \today
    \end{titlepage}

    \pagenumbering{roman}
    \begin{center}
        {\large \bfseries Semi-Quantum Money}
        
        \vspace{0.5cm}
        
        {\bfseries Roy Radian}
        
        \vspace{0.5cm}
        
        Thesis submitted in partial fulfillment of the requirements for the Master of Sciences degree
        
        \vspace{0.25cm}
        
        {Ben-Gurion University of the Negev}
        
        \vspace{0.25cm}
        
        \today
        
        \vspace{2cm}
        
        {\bfseries Abstract}
        
        \vspace{0.5cm}
    \end{center}
\else
    \begin{abstract}
\fi

Quantum money allows a bank to mint quantum money states that can later be verified and cannot be forged. Usually, this requires a quantum communication infrastructure to transfer quantum states between the user and the bank. Gavinsky~\ifnum\cryptology=1 (CCC 2012) \else \cite{Gav12} \fi introduced the notion of classically verifiable quantum money, which allows verification through classical communication. In this work we introduce the notion of classical minting, and combine it with classical verification to introduce semi-quantum money. Semi-quantum money is the first type of quantum money to allow transactions with completely classical communication and an entirely classical bank. This work features constructions for both a public memory-dependent semi-quantum money scheme and a private memoryless semi-quantum money scheme. The public construction is based on the works of Zhandry~\ifnum\cryptology=0 \cite{Zha19} \fi and \ifnum\cryptology=0 Coladangelo~\cite{Col19}, \else Coladangelo, \fi and the private construction is based on the notion of Noisy Trapdoor Claw Free Functions (NTCF) introduced by Brakerski et al.~\ifnum\cryptology=1 (FOCS 2018). \else \cite{BCM+18}. \fi

In terms of technique, our main contribution is a perfect parallel repetition theorem for NTCF.

\ifnum\masterthesis=0
    \end{abstract}
\fi

\ifnum\sigconf=1
    \settopmatter{printfolios=true}
    \acmYear{2019}\copyrightyear{2019}
    \setcopyright{acmlicensed}
    \acmConference[AFT '19]{AFT '19: Conference on Advances in Financial Technologies}{October 21--23, 2019}{Zurich, Switzerland}
    %\acmBookTitle{AFT '19: Conference on Advances in Financial Technologies, October 21--23, 2019, Zurich, Switzerland}
    \acmPrice{15.00}
    \acmDOI{10.1145/3318041.3355462}
    \acmISBN{978-1-4503-6732-5/19/10}
    \maketitle
\fi

\ifnum\masterthesis=1
    \pagebreak
    \pagenumbering{arabic}
    This thesis is based on and significantly extends \cite{rs19b}.
    \subsection*{Acknowledgments}
    I would like to thank Dr. Or Sattath for his patient guidance, and for going the extra mile to see that I learn.

    \setcounter{tocdepth}{3}
    \tableofcontents
    
    \listoffigures
\fi

\ifnum\cryptology=1
    \paragraph{Keywords:} Quantum cryptography, Quantum Money, Trapdoor Claw Free Functions, Semi-Quantum Money
\fi

\section{Introduction} % (fold)
\label{sec:introduction}
Introduced by Wiesner circa 1969, quantum money was the precursor to what is now known as quantum cryptography~\cite{Wie83}. The motivation behind quantum money is to design money that is physically impossible to counterfeit, by using a variant of the (quantum) no-cloning theorem~\cite{WZ82,Par70,Die82}. This notion of quantum money is in sharp contrast to our current notions of bills and coins that, at least in principle, can be counterfeited.

All quantum money schemes consist of three parts: $\keygen$, which generates a key, $\mint$ which uses the key to issue a new quantum money state, and $\verify$ which tests whether an alleged money state is legitimate. There are two main categories of quantum money: private and public. In a private setting, the key is required to run the verification. On the other hand, in a public quantum money scheme, $\keygen$ generates a secret/public key-pair, where the secret key is used in $\mint$ and the public key is used in $\verify$. In this work we deal both with private and public schemes.

A variant of quantum money called classically verifiable quantum money was introduced in~\cite{Gav12} (see also~\cite{PYJ+12,MVW13,GK15,BS16,AA17,AGKZ20}): the money is verified via an interactive protocol between the user and the bank. This protocol requires a quantum computer for the user, a classical computer for the bank, and classical communication between them. In a classical verification, the banknote is always measured and destroyed. This destructive measurement makes sure it cannot be reused\footnote{This is not problematic; the bank would simply mint a new banknote and send it, through a quantum channel, to the receiver. The concept of non-reusable money is not a new one; in fact, an otherwise secure quantum money scheme could sometimes be broken if banknotes are reused (see \cite{BNSU14, Lut10, Aar09}).}.

In this work, we introduce a new variant of classically verifiable quantum money: semi-quantum money. In this setting, the minting also shares this property, i.e., it is a protocol that involves \emph{both} the bank \emph{and} the user, and requires only classical resources from the bank. In standard quantum money, in contrast, minting is a quantum algorithm run by the bank, which sends the output --- the quantum money state --- to the user, via a quantum channel.

In semi-quantum money, the money state is generated by the \emph{user}. This concept seems somewhat counterintuitive; if banknotes are generated by the user, could the user not create as many notes as he or she pleases? The key point of the minting process is the protocol between the user and the bank: the user is supposed to generate a superposition over two registers using information provided by the bank, measure one of the registers, and report the result back to the bank. If the user will try to repeat the same procedure, the measurement outcome --- as well as the post-measured state --- will be different with overwhelming probability. As far as the authors are aware, no prior work considered classical minting.

The fact that semi-quantum money is also classically verifiable means that instead of sending the quantum state to the bank for verification, the user and the bank run a classical interactive verification protocol that tests the validity of the money. Semi-quantum money got its name from the fact that the minting and verification protocols require only classical resources (communication and computation) from the bank.

This introduction of a quantum money scheme where the banks are classical perhaps raises the question whether the concept could be improved, such that the bank would be quantum and the user classical. However, such a setting is inherently flawed; if the user is classical, they could not hold their own money, meaning the bank would have to hold the state of every note of every user\footnote{We refer to such a scheme as "memory-dependent", and explore its consequences in \ifnum\sigconf=0 \cref{sec:advantage_of_statelessness}\else the full version~\cite{RS19}\fi.}. This makes the "quantumness" of the money redundant, since it would be permanently kept within the bank in any case. Thus, it would seem that the setting where the bank and communication is classical and the user is quantum is the "least quantum" a quantum money scheme could be.

In this work we introduce both a public construction and a private construction for semi-quantum money. The public construction is based on an existing public quantum money scheme which we combine with an existing tool that allows classical verification, so our public construction requires little technical work. Our private scheme, on the other hand, is based on NTCF --- a tool which was designed for a different purpose entirely. Its construction, therefore, entails a much greater technical challenge. For that reason we address the public result first

\paragraph{Assumptions.} Our results assume authenticated and noiseless classical channels (which could be realized using standard classical error-correction and authentication techniques), along with perfect quantum devices (quantum memory, quantum computer and quantum communication channels). Of course, such quantum devices are not currently available, and are not predicted to be available in the short term (especially because of the long term quantum memory inherently required for quantum money).

\paragraph{Prior Knowledge.} Before we go any further, we discuss the accessibility of this work. The reader is assumed to have a basic understanding of classical cryptography, and we follow the definitions and conventions of \cite{Gol04} and \cite{KL14}. This work is aimed at readers who are familiar with quantum computing, but is also accessible to other readers. For further reading, consult \cite{NC11} for general quantum computing, and \cite{BS16b} for quantum cryptography. The two major "quantum" facts that are crucial to understand for this paper are the following: (i) A qubit is the quantum analog of a bit. Unlike bits, qubits cannot be copied due to the no-cloning theorem. (ii) To extract classical information from qubits, a measurement has to be preformed. The measurement changes the quantum state. Crucially, this process is not reversible. This is in contrast to classical systems, where rewinding is possible.

\paragraph{Public semi-quantum money.}
In a public quantum money scheme, unlike in a private scheme, any user can verify a banknote using the bank's public key without aid from the bank. There are several advantages for a public scheme: it does not require three-party quantum communication between the bank, the sender, and the receiver. The only requirement is a quantum channel between the sender and the receiver\footnote{In our public semi-quantum money construction, the classical verification can only be done with the bank; users still require quantum communication to transfer banknotes without the bank. We leave it as an open question whether public semi-quantum money with \textit{entirely} classical communication can be made. A recent work of Amos et al. \cite{AGKZ20} exhibits completely classical communication between users, but in the oracle model.}. Public schemes have a major advantage over private schemes also in terms of privacy: since the bank is not involved in the transactions, the bank cannot track all transactions of the note. However, it is much harder to construct a secure public scheme --- see the related works paragraph below.

We construct a public semi-quantum money scheme based on Zhandry's quantum lightning (\cite{Zha19}), and the notion of bolt-to-certificate introduced in Coladangelo's follow-up work (\cite{Col19}). Our classical verification based on \cite{Col19} is memory-dependent, meaning the bank has to keep a database of spent notes. We leave it as an open question whether a \emph{memoryless} public semi-quantum money exists (we compare memory-dependent vs. memoryless quantum money in \ifnum\sigconf=0 \cref{sec:advantage_of_statelessness}\else the full version~\cite{RS19}\fi).

Our main public result is:

\begin{restatable}[Public Semi-Quantum Money]{theorem}{publicresult}
    \label{thm:main_public_theorem}
    Assuming the existence of a secure Quantum Lightning scheme (\cref{def:quantum_lightning,def:ql_security}) with bolt-to-certificate capability (\cref{def:bolt_to_certificate}), and the existence of a PQ-EU-CMA digital signature scheme (\cref{def:digital_signature,def:unforgeability_of_digital_signature}), then a secure memory-dependent (\cref{def:memory_dependent}) public semi-quantum money scheme exists (\cref{def:public_semi_quantum}).
\end{restatable}

Quantum lightning (\cite{Zha19}) is a type of public quantum money such that each quantum banknote (called a \emph{bolt}) is unique: a "lightning bolt" is a quantum state, and has a serial number that is a classical string. It is hard for everyone, \emph{including the bank}, to construct two valid bolts with the same serial number.
This can be thought of as if someone would "freeze" and "capture" lightning bolts in a thunderstorm that have the same fingerprint (in this case, the serial number).

The notion of bolt-to-certificate was introduced in \cite{Col19}, and it describes a process of turning a quantum lightning bolt into a classical \emph{certificate}. The certificate proves a lightning bolt with a certain serial number has been destroyed (the serial number is classical and thus survives the destruction of the bolt itself).

In this work, we use quantum lightning with bolt-to-certificate capability to construct a public semi-quantum scheme. In public semi-quantum money we want to allow any user to verify banknotes \emph{without} destroying them, and classical verification with the bank (communication between two users is still quantum). Therefore, we facilitate a quantum verification algorithm to be used by the quantum users (that would preserve the banknote), and a classical verification protocol to be used with the classical bank (that would destroy the banknote). The quantum verification is derived directly from quantum lightning verification (since quantum lightning is already a public money scheme), and the classical verification is derived from the bolt-to-certificate capability --- the bolt is exchanged for a classical certificate that is shown to the bank to prove the note has been spent. Moreover, we introduce a slight alteration to the quantum lightning scheme such that the banknotes become user-generated; instead of the bank producing the bolt and sending it to the user, the user would generate the bolt and the bank would sign its (classical) serial number. The resulting scheme requires quantum communication between the quantum users and classical communication with the classical bank (banknotes could still be passed between users with only classical communication by going through a bank, in the same manner of a private scheme).

The security notion is as follows: no $w>\ell$ verifications can be made using $\ell$ bolts. In our case this means that, besides the fact that the same bolt cannot be used for two quantum or two classical verifications, the same bolt cannot be used for both a quantum and a classical verification (while the quantum verification does not destroy the note, it passes it to the other user). Moreover, there is an additional security concern to be considered; the idea of \emph{sabotage} (introduced in the context of quantum lightning in~\cite{CS20}, and in the context of quantum money in~\cite{BS16}). This notion captures the possibility of paying a user with a "sabotaged" note such that it is accepted by the receiving user (i.e., it passes the quantum verification) but will not be later accepted by another user, or by the bank (i.e., will not pass the classical verification with the bank). A quantum lightning scheme that is secure against sabotage is enough to ensure such security for our construction, and such a security proof was made in \cite{CS20}.

It should be stated that the current candidate constructions for quantum lightning are problematic. There are currently three candidate constructions: two suggested in the original work (\cite{Zha19}) and one that was published roughly 7 years earlier in \cite{FGHLS12} but seems to be compatible with quantum lightning (the connection is made in \cite{CS20}). One of the constructions in \cite{Zha19} is based on "collision resistant non-collapsing hash functions", which currently do not have a candidate construction. The other has been successfully attacked by Roberts \cite{Rob19}, though he ended on a positive note with some ideas for modification to thwart the attack. The construction in \cite{FGHLS12} seems to be a valid quantum lightning construction (it was introduced 7 years earlier and is yet unbroken), and can be used to construct public quantum money with classical minting, but does not feature bolt-to-certificate capability which is necessary for classical verifiability, thus being unusable for a semi-quantum scheme. This means that our scheme can be based on either of Zhandry's original constructions. Moreover, only one of which (the one that is broken in some respects) was proven to be secure against sabotage in~\cite{CS20} under the same assumption. Therefore, our construction for public semi-quantum money is on shaky ground. Nevertheless, even if the existing constructions of quantum lightning do not provide strong security guarantees, we believe the notion of a quantum lightning scheme, as well as that of bolt-to-certificate, to be plausible. In fact, a new candidate, with hardness based on lattices, was recently announced in an unpublished work by Peter Shor (see \url{https://youtu.be/8fzLByTn8Xk}).

\paragraph{Private semi-quantum money.}
In a private scheme, a banknote can be verified only with a bank. A private semi-quantum money scheme requires only classical communication with a classical bank.

Our main private result is: 

\begin{restatable}[Private Semi-Quantum Money]{theorem}{privateresult}
    \label{thm:main_private_theorem}
    Assuming that the Learning With Errors (LWE) problem with certain sets of parameters is hard for $\BQP$, then a secure private semi-quantum money scheme exists (\cref{def:semi_classical}). 
\end{restatable}

Our assumptions are stated in each theorem separately, and they all boil down to our LWE assumptions. There are a number of variants of the LWE problem, each one with different security parameters. We rely on different constructions from LWE which use different variants --- but both assumptions are LWE assumptions. Interestingly, while private quantum money schemes can be secure without computational assumptions, they are required for semi-quantum money. In fact, there can be no information-theoretic secure quantum money scheme with classical minting (see \cref{sec:computational_assumptions_are_necessary}).

The main technical tool through which to implement this scheme is the quantum secure trapdoor claw-free function \ifnum\sigconf=0 (TCF --- \cref{sec:definition_NTCF}) \fi recently introduced in~\cite{BCM+18} (see also ~\cite{Mah18a, Mah18b, GV19}). Informally, a quantum secure TCF is a family of functions, where each function  $f:\{0,1\}^{w} \rightarrow \{0,1\}^{w}$ in the family (a) is classically efficiently computable, (b) is 2-to-1, i.e., for every $x$ there exists a unique $x'\neq x$ such that $f(x)=f(x')$, and (c) has a trapdoor that, given $y$, can be used to find $x$ and $x'$ such that $f(x)=f(x')=y$ (when $y$ is in the image of $f$), but without the trapdoor a quantum polynomial adversary cannot find any pair $x,\ x'$ such that $f(x)=f(x')$.

In addition, we will require the adaptive hardcore bit property of a TCF that was introduced in~\cite{BCM+18}, which is explained below. Using a quantum computer, the state $\frac{1}{\sqrt 2}(\ket{x}+\ket{x'})$, where $x$ and $x'$ are two pre-images of $y$, could be measured, and one pre-image of $y$ could be found. Moreover, by measuring the state in the Hadamard basis, a non-zero string $d$ that satisfies
$d \cdot (x \oplus x')=0$ could be extracted:
 \begin{align}
 \begin{split}
 H^{\tensor w}\frac{1}{\sqrt 2}(\ket{x}+\ket{x'})&=\frac{1}{\sqrt {2^{w+1}}}\sum_{d\in\{0,1\}^{w}} (-1)^{d\cdot x}+(-1)^{d \cdot x'}\ket{d}\\
&=\frac{1}{\sqrt {2^{w-1}}}\sum_{d\in\{0,1\}^{w}|d\cdot(x\oplus x')=0} (-1)^{d \cdot x}\ket{d}\;.
\label{eq:1}
\end{split}
\end{align}
In our construction we use the following two tests: the pre-image test (providing a pre-image of $y$) and the equation test (providing a non-zero $d$ that satisfies the above condition). The adaptive hardcore bit property guarantees that, even though either test can be easily passed on its own, any quantum polynomial time (QPT) adversary can successfully pass \emph{both} tests with probability at most $\frac{1}{2} + \negl$. Brakerski et al. used these tests to construct a cryptographic test of quantumness (CTQ); our construction can be seen as a reinterpretation of this protocol in the quantum money setting, where the first part of the protocol can serve as the creation of a quantum money state, and the second part can serve as the its verification. The transition to quantum money introduces some challenges; mainly the need for a parallel repetition theorem for our NTCF-based primitive, and proving full-scheme security.  Brakerski et al. showed a construction of a \emph{noisy} trapdoor claw-free function (\nom{NTCF}{Noisy Trapdoor Claw-free Function}{NTCF}) that holds this adaptive hardcore property, based on the hardness of the Learning With Errors (\nom{LWE}{Learning With Errors}{LWE}) problem \cite{BCM+18}. For the sake of clarity, we ignore the issues related to the noisy property in this introduction.
%The full analysis (with some technical adjustments) follows the outline presented below.

A TCF on its own, however, is not hard enough to construct a money scheme with; we do not want adversaries to be able to forge banknotes with probability $\frac{1}{2}$. To that end, we would like to amplify the hardness using some sort of a parallel repetition theorem (see \cref{sec:parallel_repetition_theorem_for_1_of_2_puzzles}). Luckily, we can rephrase this setting using the framework of \emph{weakly verifiable puzzles} for which a perfect parallel repetition theorem is known~\cite{CHS05}. This perfect parallel repetition guarantees that answering both tests for $n$ puzzles correctly is as hard as trying to answer them independently, i.e., at most $\left(\frac{1}{2}\right)^{n}$ (up to negligible corrections), which is exactly our goal.

Next, we present the outline and analysis of our semi-quantum private money scheme construction. The security notion of our money scheme is rather straightforward: an adversary that receives $\ell$ banknotes, and can attempt to pass verification (polynomially) many times, cannot pass more than $\ell$ verifications. To show a construction that meets this notion, we work our way through several weaker security notions; this makes proving the security of our full scheme construction simpler.
We first show how to construct a semi-quantum money scheme (\cref{sec:strong_1-of-2_puzzles_imply_money}) that provides a weaker level of security than a full scheme. Here, we wish to show that a counterfeiter that receives $1$ quantum money state cannot create two states that will both pass verification with non-negligible probability. We call a scheme that satisfies this weaker notion of security a 2-of-2 mini-scheme --- see \cref{def:QM_2-of-2_mini_scheme}.

We now describe the construction of a 2-of-2 mini-scheme, starting with the (honest) minting protocol. The bank picks $n$ functions $f_1,\ldots,f_n$ uniformly at random from the TCF family and sends them to the user, while keeping the trapdoors $t_1,\ldots,t_n$ private. The user creates a superposition of the form $\ket{\psi_1}\tensor \ldots \tensor \ket{\psi_n}$, where $\ket{\psi_i}=\frac{1}{\sqrt {2^{w}}}\sum_{x\in \{0,1\}^{w}} \ket{x} \tensor \ket{f_i(x)}$. The user measures all the r.h.s. registers (i.e., $\ket{f_i(x)} \; \forall 1 \leq i \leq n$) and sends the resulting $y_1,\ldots,y_n$ to the bank, who saves them to its database\footnote{We deviate here slightly from the formal definitions; Since the bank does not have a ``database'', verification should only use the key. This is handled by using a message authentication code (MAC) and by returning to the user a tag for these values, and then verifying that tag during the verification. For the sake of clarity, we omit this part in the discussion --- refer to Algorithm~\ref{alg:mini_scheme_construction} to see how we work around this issue.}. Note that due to the measurement, the $i^{th}$ state collapses to $\ket{\psi_{y_i}}=\frac{1}{\sqrt{2}}(\ket{x_i}+\ket{x_i'})$, where $f_i(x_i)=f_i(x_i')=y_i$.

For verification, the bank chooses a random challenge $C_{i}\in_{R} \{0,1\}$ (which is either the pre-image or the equation challenge) for each of the $n$ registers. For the pre-image challenge, $C_{i}=0$, the user must provide a string $x_{i}$ such that $f_i(x_{i})=y_{i}$. The honest user can measure $\ket{\psi_{y_{i}}}$ to find a pre-image of $y_{i}$ to pass this test with certainty. 
In the equation challenge, $C_{i}=1$, the user must provide a non-zero string $d_{i}\in \{0,1\}^{w}$ such that $d_{i}\cdot (x_{i}\oplus x_{i}')=0$. The bank can test whether the equation challenge holds by using the trapdoor $t_i$ to calculate both $x_i$ and $x_i'$. An honest user can generate such a string by measuring $\ket{\psi_{y_{i}}}$ in the Hadamard basis, as described in Eq.~\eqref{eq:1}. The measured $d_i$ will be non-zero (except with probability exponentially small in $w$) which will allow the user to pass this test.

We emphasize that for both the minting and the verification protocols, the bank only needs a classical computer.

We now outline the security argument. Suppose the user tries to pass verification twice. Denote by $C\in\{0,1\}^n$ the challenge vector in the first attempt, denote by $C'$ the challenge vector in the second attempt, and denote by $S$ the set of coordinates in which they differ: $S=\{i\in[n]|C_i\neq C_i' \}$. With overwhelming probability, $S \neq \varnothing$, in which case for at least one coordinate the user will have to pass both challenges, and cannot succeed except with negligible probability.
%We expect $|S|$ to be roughly $\frac{n}{2}$. To pass the tests in each round $i \in S$ in both attempts, the user must pass both the pre-image \emph{and} the equation challenges. We know that the success probability of passing both tests for each $i \in S$ is $\frac{1}{2}+\negl$. To argue that the probability of passing all these tests becomes exponentially small with $n$, we need some sort of a parallel repetition theorem (see \cref{sec:parallel_repetition_theorem_for_1_of_2_puzzles}). Luckily, we can rephrase this setting using the framework of \emph{weakly verifiable puzzles} for which a (perfect) parallel repetition theorem is known~\cite{CHS05}. This parallel repetition guarantees that answering these $\frac{n}{2}$ puzzles correctly is as hard as trying to answer them independently, i.e., at most $\left(\frac{1}{2}\right)^{n/2}$ (up to negligible corrections), which is exactly our goal.

%Note that it is crucial to pick for each $i$ a different trapdoor function $f_i$: otherwise, the adversarial user could have used the first register to find a pre-image $x_1$ of $y_1$, guess a random $d$ (which would pass the equation test with probability $\frac{1}{2}$), and use the same $x_1$ and $d$ as all of its responses. Even though the success probability of that test is $\frac{1}{2}$ the success probability is perfectly correlated between all tests, and therefore, the overall forging probability is $\frac{1}{2}$. By using a different TCF for each $i$, this attack is avoided.

The construction above is a semi-quantum 2-of-2 mini-scheme (rather than a full blown scheme). There is a slightly stronger notion of security (that is still weaker than a full blown scheme) called a mini-scheme (adapted from Aaronson and Christiano~\cite{AC13}). In a mini-scheme, the counterfeiter is given a single quantum money state and can attempt to pass verification polynomially many times. The counterfeiter succeeds if at least two of these verifications are accepted. We show in \cref{sec:construction_of_mini_scheme} that the scheme above also achieves this stronger notion.

In a full quantum money scheme, the adversary can ask for $t$ money states and must pass at least $t+1$ verifications.
Aaronson and Christiano~\cite{AC13} defined the notion of a \emph{public} money mini-scheme and showed how such a mini-scheme can be lifted to a full-blown scheme. Ben-David and Sattath~\cite{BS16} showed a similar result that lifts a \emph{private} money mini-scheme to a full-blown scheme. In this work, we show how to lift an \emph{interactive} private money mini-scheme to a full-blown scheme. The goal of such a mapping is to ensure that the scheme can support the issuance of multiple money states without increasing the key-size. This is done by using an authenticated encryption scheme for the mini-scheme keys and including that authenticated ciphertext as part of the money. As part of the verification, the bank can later decrypt the mini-scheme key, and use it to run the original mini-scheme verification. It is important that the encryption scheme be authenticated to prevent the adversary from altering that information (which would be possible if, for example, the encryption scheme was malleable).

\paragraph{Related works.}
The security of private quantum money schemes is generally solid,~\cite{Wie83,MVW13,PYJ+12,YTN03,MS10,Gav12,GK15,JLS18}. Secure public quantum money is much harder to construct. The constructions of Aaronson~\cite{Aar09} was broken in ~\cite{LAF+10}, and the construction of Aaronson-Christiano~\cite{AC13} was broken using several approaches --- see the most recent attack in Ref.~\cite{PDF+19} and references therein. 
In Ref.~\cite{BS16}, a construction based on quantum-secure indistinguishability obfuscation (IO) was presented, as well as a mechanism to provide classical verifiability. 
Zhandry later proved that the quantum money is indeed secure~\cite{Zha19}, though Zhandry's proof does not lend itself to the classical verifiability construction by Ben-David and Sattath. The only two constructions that are not known to be broken are by Farhi et al.~\cite{FGHLS12} (see also~\cite{Lut11}), which does not have a security proof, and three constructions by Zhandry~\cite{Zha19}: The two quantum lightning constructions discussed above, and another one which proves the security of Aaronson-Christiano based on quantum-secure IO. We note that the (classical) security of indistinguishability obfuscation is still ongoing research (see, \url{https://malb.io/are-graded-encoding-schemes-broken-yet.html} for a list of constructions and their security status). As far as the authors are aware, no IO construction claims to be quantum-secure, which is required for Zhandry's scheme. To conclude, the security of public quantum money leaves much to be desired. 

The work of~\cite{HS20} is somewhat relevant to what we do here, but from a device-independent point of view of private quantum money. There, the bank provides a secret key to an untrusted mint, and the mint produces the quantum banknote itself. An assumption is made on the mint regarding the dimension of the state that it outputs. While \cite{HS20} do not use the definition of a mini-scheme as we do here (\cref{def:QM_mini_scheme}), their work features a private quantum money construction with mini-scheme security while not proving full scheme security (\cref{def:security_private_interactive_money}). Their scheme is also classically verifiable. To conclude, the main advantage of their scheme is that it is unconditionally secure, while the main disadvantages are that the scheme does not provide classical minting, there is an additional assumption on the dimension of the output register, and there is no security proof as a full quantum money scheme.

A very recent work of Amos et al. \cite{AGKZ20} introduced a public quantum money with completely classical communication. Their construction allows users to transfer money using classical communication only, requiring no interaction with a bank. However, the security is proved only with respect to an oracle.

A recent work of \cite{ACGH19} also shows a parallel repetition theorem, though their result is slightly different, and their proof techniques are completely different.

%These requirements could be met, but not simultaneously: the public quantum money in Refs.~\cite{FGHLS12,Zha19} could support user generated minting, which provides the second advantage. Ben-David and Sattath~\cite{BS16} extended the construction of Aaronson and Christiano~\cite{AC13}, so that a user could send money to others similarly to a classically verifiable (classical) quantum money, via the bank. As a side note, we comment there are known attacks on the original construction by Aaronson and Christiano (see~\cite{PFP15,Aar16,BS16,PDF+19}). A recovery based on indistinguishability obfuscation was suggested in~\cite{BS16}, and  proved to be secure later by Zhandry~\cite{Zha19}. Unfortunately there is still no security proof for the construction in~\cite{BS16} for this extended functionality.

\paragraph{Modes of Operation.}
To fully grasp the implications of semi-quantum money, it is important to review how quantum money in its different flavors can be used, and the communication requirements (classical or quantum) of each mode of operation. In \cref{sec:modes_of_operation} we provide a detailed overview of the different flavors of quantum money and their respective communication requirements, showing the advantage of semi-quantum money: transactions with purely classical communication.

\paragraph{Our contribution.}
Our contribution is twofold. The first is semi-quantum money, both private and public: new models of quantum money that allow transactions with no quantum communication whatsoever (see \cref{fig:semi_quantum_transaction_direct,fig:semi_quantum_transaction_bank}), and only a classical bank.
The main advantage of the new schemes compared to previous quantum money schemes is that they could be used without a quantum communication infrastructure. Classical communication has several interesting benefits over quantum communication. The most obvious one is that a classical communication infrastructure already exists; a semi-quantum money scheme --- unlike previous money schemes --- will not require a quantum communication infrastructure.
Implementing such an infrastructure on a global scale will be expensive and challenging, and might be realized years after efficient quantum computers are commonly used.
There are other benefits to classical communication, even if quantum communication infrastructure was readily available. First, due to the no-cloning theorem, quantum information cannot be re-sent. In the context of quantum money, data-loss is extremely problematic --- data loss means lost money. Quantum communication will naturally suffer more data-loss, at least initially. Second, for classical communication we can keep a record (and even a signed record) which helps with matters of dispute resolution, auditing and error-handling, whereas quantum communication cannot be logged. The same argument can be made for the banks themselves; classical banks could more easily keep records and be audited.

The second contribution is the parallel repetition theorem for 1-of-2 puzzles (described earlier in the introduction). Parallel repetition (the idea of repeating a protocol polynomially many times in parallel to gain an exponential increase in soundness) seems deceptively simple, while in reality, it sometimes behaves in unexpected ways, and such proofs are usually challenging (see~\cite{Raz11} and references therein for the non-cryptographic setting); \cite{BIN97} present several cases where parallel repetition surprisingly does not grant an exponential reduction in error rate in cryptographic-settings. The parallel repetition theorem for 1-of-2 puzzles could be useful in other cryptographic settings, as it builds on the TCF primitive to introduce a tool with an exponentially small error rate (rather than the constant error rate which is guaranteed in the original work).

\paragraph{Organization.}
\ifnum\sigconf=0
\cref{sec:modes_of_operation} provides an overview on the different ways in which different flavors of quantum money can be used, along with the communication requirements of each mode of operation.

In \cref{sec:public_semi_quantum_money} we deal with our proposed public semi-quantum money. \cref{sec:public_money_definitions} contain the relevant definitions, and in \cref{sec:public_money_construction} we construct the public scheme and prove its security, proving our main public theorem, namely, \cref{thm:main_public_theorem}.
\fi

In \cref{sec:ntcf_and_1-of-2_puzzles}, we deal with NTCF and 1-of-2 puzzles. In \cref{sec:1_of_2_puzzles}, we define a 1-of-2 puzzle. In \cref{sec:NTCF_implies_1_of_2_puzzle}, we show a construction of a $\frac{1}{2}$-hard 1-of-2 puzzle based on an NTCF. We conclude \cref{sec:ntcf_and_1-of-2_puzzles} by showing, in \cref{sec:parallel_repetition_theorem_for_1_of_2_puzzles}, a method for constructing a strong 1-of-2 puzzle using repetition of weak 1-of-2 puzzles.

In \cref{sec:strong_1-of-2_puzzles_imply_money}, we deal with our proposed private semi-quantum money. \cref{sec:private_money_definitions} contains the relevant definitions, which are adaptations of the definitions from \cref{sec:public_money_definitions} to the private setting. In \cref{sec:construction_of_mini_scheme}, we construct a semi-quantum money mini-scheme and prove its security. In \cref{sec:mini_scheme_implies_full_scheme}, we present a full semi-quantum money scheme construction based on any semi-quantum mini-scheme, and prove its security.

In \cref{sec:main_thm_proof} we combine the results of \cref{sec:ntcf_and_1-of-2_puzzles,sec:strong_1-of-2_puzzles_imply_money} to prove our main private result, namely, \cref{thm:main_private_theorem}. A structural overview of our private semi-quantum money result is shown in Fig.~\ref{fig:structure}.

In \cref{sec:computational_assumptions_are_necessary} we show that a quantum money scheme with classical minting (and therefore any semi-quantum money scheme) cannot be information-theoretically secure (i.e., it must rely on computational assumptions).

\ifnum\sigconf=1 Some of the proofs, and preliminary standard definitions are deferred to the full version~\cite{RS19}.
\else
\cref{sec:nomenclature} is a nomenclature, a "cheat sheet" describing some of our notations. \cref{sec:preliminaries} contains mainly the standard definitions of private-key encryption and message authentication code (MAC), and can be safely skipped by readers who are familiar with these notions. \cref{sec:ql_with_bolt_to_cert}, taken almost verbatim from \cite{CS20}, comprises the definitions of quantum lightning with bolt-to-certificate. \cref{sec:definition_NTCF}, taken almost verbatim from Brakerski et al. \cite{BCM+18}, comprises a definition of NTCF. \cref{sec:advantage_of_statelessness} comprises a discussion of memoryless vs memory-dependent schemes. \cref{sec:parallel_repetition_weakly_verifiable_puzzles} is a simplified overview of the parallel repetition of weakly verifiable puzzles result from \cite{CHS05} which is used for our parallel repetition of 1-of-2 puzzles result. \cref{sec:transaction_figs} comprises a series of figures illustrating the various ways transactions can be performed with different flavors of quantum money schemes (regular, classically-verifiable and semi-quantum) along with the types of communication each method and flavor requires.

\fi

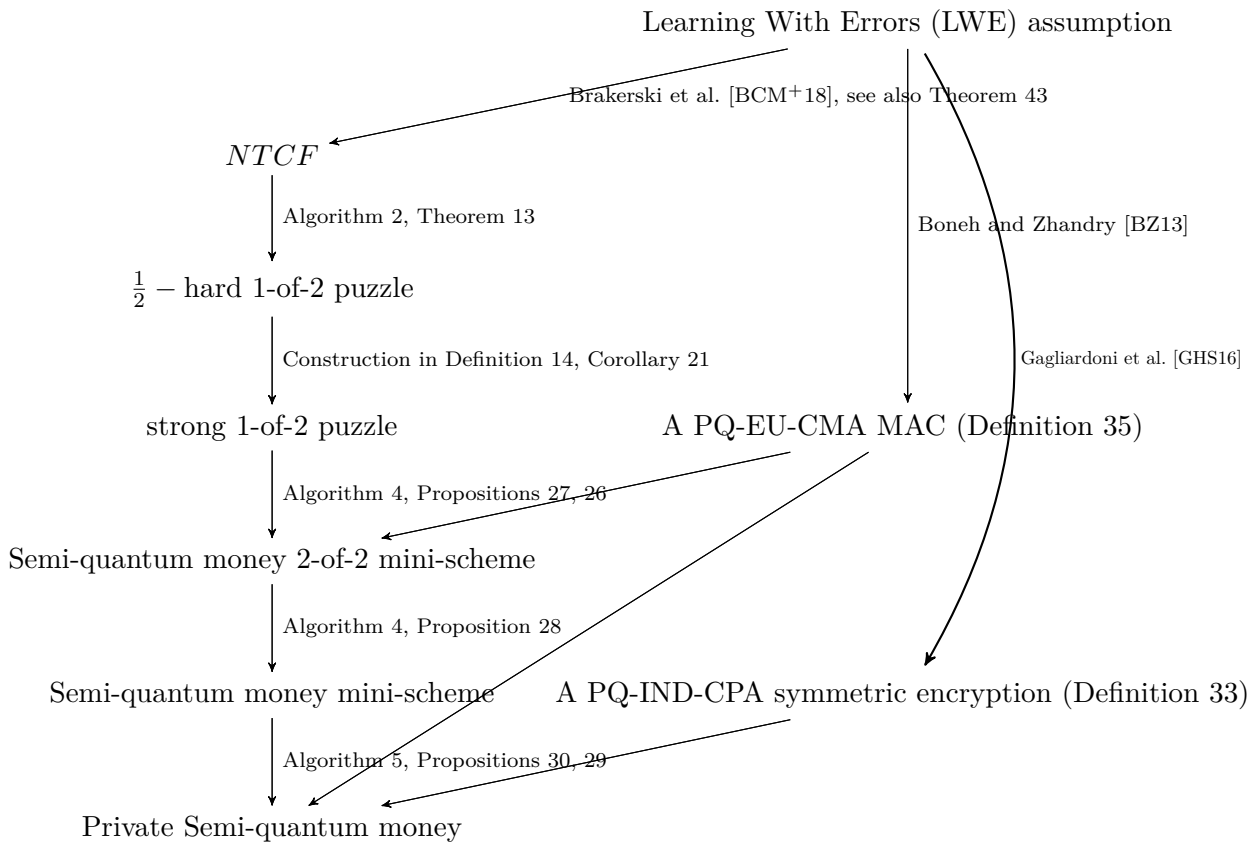
\begin{figure*}
    \centering
    \begin{tikzpicture}
  \matrix (m) [matrix of math nodes, row sep=3em, column sep=0em]
    { & \text{Learning With Errors (LWE) assumption} \\
      NTCF  &\\
      \frac{1}{2}-\text{hard 1-of-2 puzzle} & \\
      \text{strong 1-of-2 puzzle} &\text{A \ifnum\sigconf=0 PQ-EU-CMA MAC (\cref{def:unforgeability_of_MAC}) \else post-quantum universally unforgeable under CMA MAC \fi}\\
      \text{Semi-quantum money 2-of-2 mini-scheme} & \\
      \text{Semi-quantum money mini-scheme} &\text{A \ifnum\sigconf=0 PQ-IND-CPA symmetric encryption (\cref{def:indistinguishability_of_encryption}) \else post-quantum indistinguishable under CPA symmetric encryptions \fi}\\
      \text{Private Semi-quantum money} & \\
     };
  { [start chain] \chainin (m-1-2);
    \chainin (m-2-1) [join={node[right,labeled] {\text{Brakerski et al.~\cite{BCM+18}\ifnum\sigconf=0, see also \cref{thm:LWE_implies_NTCF}\fi}} }];
    \chainin (m-3-1) [join={node[right,labeled] {\text{\cref{alg:1-of-2-puzzle_from_NTCF}, \cref{thm:NTCF_implies_1-of-2_puzzles}}}  }];
    \chainin (m-4-1) [join={node[right,labeled] {\text{Construction in \cref{def:repetition}, \cref{cor:weak_implies_strong_puzzle}}}  }];
    \chainin (m-5-1) [join={node[right,labeled] {\text{\cref{alg:mini_scheme_construction}, Propositions \ref{prop:dollar_z_is_2_of_2}, \ref{prop:mini_scheme_construction_correctness}}}  }];
    \chainin (m-6-1) [join={node[right,labeled] {\text{\cref{alg:mini_scheme_construction}, \cref{prop:mini_scheme_construction}}}  }];
    \chainin (m-7-1) [join={node[right,labeled] {\text{\cref{alg:full_scheme_construction}, Propositions \ref{thm:full_scheme_construction_security}, \ref{prop:full_scheme_construction_completeness}}}  }];
    
     }

  { [start chain] \chainin (m-1-2);
    \chainin (m-4-2) [join={node[right,labeled] {\text{Boneh and Zhandry~\cite{BZ13}}} }];
    
    }
    
%   { [start chain] \chainin (m-1-2);
%     \chainin (m-6-2) [join={node[right, labeled]
%     {\text{Gagliardoni et al.~\cite{GHS16}}}}];
    
%     }
    
    \path
    (m-1-2) edge [pil, bend left=30] node[auto] [scale=0.65] {\text{Gagliardoni et al.~\cite{GHS16}}} (m-6-2);
    
  { [start chain] \chainin (m-4-2);
    \chainin (m-5-1) [join={node[right,labeled] {} }];
    
     }
  { [start chain] \chainin (m-4-2);
    \chainin (m-7-1) [join={node[right,labeled] {} }];
    
     }
  { [start chain] \chainin (m-6-2);
    \chainin (m-7-1) [join={node[right,labeled] {} }];
    
     }
\end{tikzpicture}
\caption{Structure of our private scheme construction. The right-hand side of the figure shows our \emph{assumptions}. The arrows point to constructions that make use of these assumptions.}
\label{fig:structure}
\end{figure*}

\section{Modes of Operation}
\label{sec:modes_of_operation}

This section discusses the different methods in which quantum money can be used. When discussing a transaction, we are interested in three parties: the payer \nom{$\payer$}{A payer in a transaction}{P} who has money; the receiver \nom{$\receiver$}{A receiver in a transaction}{R} who should receive $\payer$'s money by the end of the transaction; and the bank \nom{$\bank$}{A bank}{B}, a trusted third party. When discussing regular money, we can think of 2 different modes to execute a transaction: a direct transaction, or through the bank. \cref{tab:regular_money_communication_requirements} shows the communication requirements for a direct transaction and a transaction through the bank for regular (non-quantum) money.

One example of a direct transaction is online credit card payments (see \cref{fig:regular_transaction_direct}). In this case $\payer$ transfers her credit card details to $\receiver$. $\receiver$ then contacts $\bank$ with $\payer$'s credit card details and $\bank$ gives him the payment (after verifying $\payer$'s credit card details). This mode requires communication between $\payer$ and $\receiver$ and between $\receiver$ and $\bank$.

One example of a transaction through the bank is a wire transfer (see \cref{fig:regular_transaction_bank}). In this case $\payer$ contacts the bank requesting to transfer money to $\receiver$. Bank then contacts $\receiver$ to complete the transaction (deducting money from $\payer$'s account and crediting $\receiver$'s account with the payment) and sends a message alerting $\receiver$ to the transaction. This mode requires communication between $\payer$ and $\bank$ and between $\receiver$ and $\bank$.

\begin{table}[H]
    \centering
    \begin{tabular}{m{3.6cm}|m{1.2cm}m{1.5cm}|m{1.2cm}m{1.3cm}}
        & \multicolumn{2}{c}{Direct} \vline& \multicolumn{2}{c}{Through Bank}\\
        \hline
        Regular Money & \includegraphics[scale=\tablefigscale]{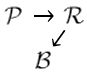} & (\cref{fig:regular_transaction_direct}) &
        \includegraphics[scale=\tablefigscale]{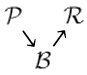} & (\cref{fig:regular_transaction_bank})
    \end{tabular}
    \caption{Types of communication required for a direct transaction and a transaction through the bank with regular (non-quantum) money. $\rightarrow$ denotes classical communication.}
    \label{tab:regular_money_communication_requirements}
\end{table}

We make the observation that different transaction methods exist for quantum money as well. In the quantum money setting this discussion is much more delicate, since there are various flavors of quantum money, and there is the question of which types of communication (classical or quantum) are required between which parties for any transaction method. An important distinction is that quantum money is held in the form of quantum states on the user's quantum computer rather than in an account kept by the bank.

The straightforward way to think of a quantum money transaction is a direct transfer in which $\payer$ sends $\receiver$ the banknote directly. Afterwards, $\receiver$ needs to verify the validity of the money, since $\payer$ is not trusted. In a public scheme this can be done locally on $\receiver$'s computer, while in a private scheme $\receiver$ must interact with $\bank$ for this purpose.

There is another way to transfer quantum money: a transaction through the bank. In such a transaction, instead of sending the money directly to $\receiver$, $\payer$ sends the money to $\bank$. $\bank$ then validates the money, and proceeds to send it to $\receiver$ (or mint a new banknote for him). This type of transaction has the following advantage: in a direct transfer, when $\payer$ sends her money to $\receiver$, $\receiver$ may claim that the money did not pass verification and refuse the payment. In this scenario there is usually no way for $\payer$ to prove that she gave $\receiver$ valid money. This is not the case, however, when performing the transaction through $\bank$; since the bank $\bank$ is assumed to be trusted for the money to function properly, any claim made by $\receiver$ that he did not receive valid money can be refuted by $\bank$ who took part in the transaction and made sure the money was valid. Along with this advantage, transactions via $\bank$ offer different communications requirements which could prove beneficial. These two benefits make it is interesting to review bank transactions even in the setting of public quantum money.

When discussing transactions it is interesting to talk about traceability; i.e., can any party learn any information regarding transactions that have been made between other parties? Generally, quantum money comes in 2 flavors regarding traceability: banknotes and coins. The difference is that while quantum banknotes are unique (each banknote has a serial numbers), while quantum coins are indistinguishable. Thus quantum banknotes are inherently traceable, while the traceability of coins depend on how they are used (see \cite{AMR20} for a formal definition). Apart from the differences between transaction methods discussed above, a direct transaction and a transaction through the bank also behave differently when considering traceability of quantum coins. In a direct transfer $\receiver$ learns the identity of the $\payer$ and $\bank$ learns only the identity of $\receiver$. In a transaction through the bank, however, $\receiver$ does not learn the identity of $\payer$ and $\bank$ learns the identities of both $\payer$ and $\receiver$. This difference gives further motivation to consider the different modes of transaction.

\cref{tab:communication_requirements} lists the communication requirements of transactions for each flavor of quantum money, and \cref{sec:transaction_figs} contains detailed figures illustrating how transactions can be made for each flavor of quantum money. Quantum communication is denoted by $\quantumcom$, classical communication is denoted by $\singleclassiccom$, and a two-sided arrow indicates two-sided communication is required.

We can see that semi-quantum money allows transactions with purely classical communication. It is an open question whether a variant of semi-quantum money could be made that would allow to classically transfer quantum money between users without the aid of the bank (i.e., make a public direct transaction with classical communication alone). Amos et al. \cite{AGKZ20} showed a construction of semi-quantum money with completely classical communication between users (one-shot signature), but in the oracle model.

\begin{remark}
In \cref{tab:communication_requirements} we state that a direct transfer using classically verifiable or semi-quantum money requires only classical communication between $\payer$ and $\receiver$. This is because $\receiver$ can initiate the verification protocol with the bank, and then act as relay between $\payer$ and $\bank$ in order to verify $\payer$'s banknote. The bank then mints a new banknote for receiver to conclude the transaction (see \cref{fig:classically_verifiable_transaction_direct,fig:semi_quantum_transaction_direct}). Of course, it is still possible for the $\payer$ to send the banknote via quantum communication, but this way requires one less quantum channel.
\end{remark}

\begin{remark}\label{rmk:stronger_security_definition}
We use a stronger security notion than that used in Refs.~\cite{Gav12, GK15}. In these previous works a quantum money scheme is defined to be unforgeable if an adversary with $\ell$ banknotes cannot generate more than $\ell$ quantum states that would each pass the \emph{honest} verification (with non-negligible probability). %This is kept informal because not all papers use the same definition; some use $\ell+1$ while others use $m>\ell$
In contrast, our unforgeability requirement (\cref{def:security_public_interactive_money,def:security_private_interactive_money}) states that an adversary with $\ell$ banknotes cannot pass more than $\ell$ \emph{potentially malicious} verifications. This is more general and therefore stronger; in the existing definition the verification protocol is run honestly, while our definition allows the adversary to act maliciously during the verification protocol itself.

It is easy to see why the stronger definition is not used in the schemes above. These schemes are designed to work in the direct mode of operation, and they provide the ability to verify a banknote multiple times, meaning the scheme does not meet the stronger security definition by design. Generally, such schemes could not utilize their classical verification for a transaction through the bank; since verification does not destroy the banknote, $\payer$ would have to send her banknote to $\bank$ (through a quantum channel) so $\bank$ could be sure she no longer holds the banknote. This means that transactions through the bank would require all-quantum communication.

Allowing multiple verifications is beneficial for the classical verification setting, but irrelevant for the semi-quantum setting. For example, in Gavinsky's scheme \cite{Gav12}, the banknote is designed to allow a (fixed) polynomial number of verifications, after which the banknote would be sent back to the bank to be verified and destroyed, and a new one would be minted instead. If instead the banknote was destroyed after a single verification, the verifier would need the bank to mint him a new banknote after every verification, which would require quantum communication. Therefore, allowing multiple verifications without the need to re-mint allows verification with no quantum communication between $\receiver$ and $\bank$ (most of the time). In the semi-quantum setting, however, minting is also classical, so allowing banknotes to pass verification without being destroyed does not give us an interesting benefit, and we can use our stronger security notion. This also allows us to use the through the bank mode of operation with only classical communication.
\end{remark}

\pagebreak % This keeps the footnotes in the same page as the table
\begin{table}[H]
    \centering
    \begin{tabular}{m{1.2cm}|m{3.6cm}|m{1.2cm}m{1.5cm}|m{1.2cm}m{1.3cm}}
                & & \multicolumn{2}{c}{Direct} \vline& \multicolumn{2}{c}{Through Bank}\\
                \hline
                & Standard          & \includegraphics[scale=\tablefigscale]{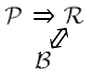} & (\cref{fig:quantum_transaction_direct}) &
                \includegraphics[scale=\tablefigscale]{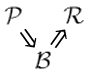} & (\cref{fig:quantum_transaction_bank})\\
                & Classically verifiable & \includegraphics[scale=\tablefigscale]{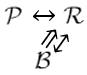} & (\cref{fig:classically_verifiable_transaction_direct}) &  \includegraphics[scale=\tablefigscale]{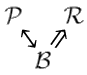} & (\cref{fig:classically_verifiable_transaction_bank})\\
        Private & Reusable classically verifiable\tablefootnote{Some classically verifiable schemes \cite{Gav12} allow to verify the same banknote multiple times before it is destroyed. Schemes with this property behave a little differently.} & \includegraphics[scale=\tablefigscale]{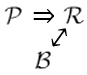} & (\cref{fig:classically_verifiable_MV_transaction_direct}) &\includegraphics[scale=\tablefigscale]{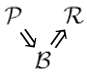} & (\cref{fig:classically_verifiable_MV_transaction_bank})\\
                & Classical minting\tablefootnote{The authors are not aware of an instantiation of a private quantum money scheme with classical minting that is not semi-quantum, but it is a valid setting to consider.} & \includegraphics[scale=\tablefigscale]{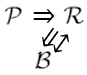} & (\cref{fig:classical_minting_transaction_direct}) & 
                \includegraphics[scale=\tablefigscale]{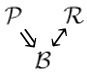} & (\cref{fig:classical_minting_transaction_bank})\\
                & Semi-quantum & \includegraphics[scale=\tablefigscale]{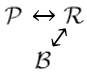} & (\cref{fig:semi_quantum_transaction_direct}) & 
                \includegraphics[scale=\tablefigscale]{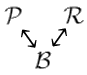} & (\cref{fig:semi_quantum_transaction_bank})\\
        \hline
                & Standard & \includegraphics[scale=\tablefigscale]{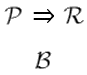} & (\cref{fig:public_transaction_direct}) & 
                \includegraphics[scale=\tablefigscale]{Transaction_Figures/Communication_Requirements/quantum_bank.png} & (\cref{fig:quantum_transaction_bank})\\
                & Classically verifiable\tablefootnote{The authors are not aware of any instantiation of a classically-verifiable public quantum money scheme (that is not also semi-quantum), but it is a valid setting to consider.} & \includegraphics[scale=\tablefigscale]{Transaction_Figures/Communication_Requirements/public_direct.png} & (\cref{fig:public_transaction_direct}) & 
                \includegraphics[scale=\tablefigscale]{Transaction_Figures/Communication_Requirements/classically_verifiable_bank.png} & (\cref{fig:classically_verifiable_transaction_bank})\\
        Public & Reusable classically verifiable & \includegraphics[scale=\tablefigscale]{Transaction_Figures/Communication_Requirements/public_direct.png} & (\cref{fig:public_transaction_direct}) & 
                \includegraphics[scale=\tablefigscale]{Transaction_Figures/Communication_Requirements/quantum_bank.png} & (\cref{fig:classically_verifiable_MV_transaction_bank})\\
                & Classical minting\tablefootnote{Quantum lightning \cite{Zha19} can be seen as quantum money with classical minting which is not necessarily classically verifiable --- in \cref{sec:public_semi_quantum_money} we construct public semi-quantum quantum money from quantum lightning with bolt-to-certificate, but we can see that classical minting is straightforward even without this quality.} & \includegraphics[scale=\tablefigscale]{Transaction_Figures/Communication_Requirements/public_direct.png} & (\cref{fig:public_transaction_direct}) & 
                \includegraphics[scale=\tablefigscale]{Transaction_Figures/Communication_Requirements/classical_minting_bank.png} & (\cref{fig:classical_minting_transaction_bank})\\
                & Semi-quantum & \includegraphics[scale=\tablefigscale]{Transaction_Figures/Communication_Requirements/public_direct.png} & (\cref{fig:public_transaction_direct}) & \includegraphics[scale=\tablefigscale]{Transaction_Figures/Communication_Requirements/semi_quantum_bank.png} & (\cref{fig:semi_quantum_transaction_bank})\\
                & One-shot signature\tablefootnote{This is a construction from \cite{AGKZ20} for public semi-quantum money which allows to transfer banknotes between users using classical communication. However, the only construction available requires an oracle.} & \includegraphics[scale=\tablefigscale]{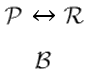} & (\cref{fig:one_shot_transaction_direct}) & \includegraphics[scale=\tablefigscale]{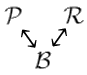} & (\cref{fig:semi_quantum_transaction_bank})
    \end{tabular}
    \caption{Types of communication required for each transaction method in each flavor of quantum money. $\Rightarrow$ denotes quantum communication and $\rightarrow$ denotes classical communication. A two-sided arrow indicates two-sided communication. Note that public schemes can also use transactions methods from private schemes.}
    \label{tab:communication_requirements}
\end{table}

\ifnum\sigconf=0
\section{Quantum Lightning with Bolt-to-Certificate Implies Public Semi-Quantum Money}
\label{sec:public_semi_quantum_money}
In this section we construct a public memory-dependent semi-quantum money scheme using Zhandry's quantum lightning from \cite{Zha19} along with Coladangelo's Bolt-to-Certificate from \cite{Col19} and its superseding work \cite{CS20}.

\subsection{Definitions of Public Semi-Quantum Money}
\label{sec:public_money_definitions}

\begin{definition}[Interactive public quantum money]
    \label{def:interactive_public_money}
    An interactive public quantum money scheme consists of a classical PPT key generation algorithm $\keygen$ and two-party interactive QPT protocols $\mint$ and $\verify$. $\keygen(\secparam)$ outputs a pair of keys $(pk, sk)$ which are the public and private keys, respectively. Both the minting protocol and the verification protocol are two-party quantum protocols: $\mint$ involves an Acquirer $\acquirer$ and a Bank $\bank$, whereas $\verify$ involves a Payer $\payer$ and a Receiver $\receiver$ (in $\verify$ either party can be either a bank or a user). During both protocols, both parties receive the public key $pk$ as input, and the bank (if it participates) receives the private key $sk$ while users do not. At the end of the honest run of $\mint$, the user holds a quantum money state that, in general, could be a mixed state. In this work, the protocols will end with a pure state, usually denoted $\ket{\$}$. In the following sections, for the sake of clarity, we work with the pure-state formalism. The banknote that the $\payer$ chooses to verify is denoted in this work as the input of the $\verify$ protocol. At the end of the verification protocol, the Receiver outputs a bit $b$ that states whether the money was accepted or not.
    \paragraph{Correctness.} The scheme is \emph{correct} if there exists a negligible function $\negl$ such that:
    \begin{align*}
        \Pr[(pk, sk) &\sample \keygen(\secparam); \ket{\$} \sample \mint_{(pk, sk)}(\secparam) : \\
        &\verify_{pk}(\ket{\$}) = 1] = 1 - \negl\;.
    \end{align*}
\end{definition}

\begin{definition}[Memory-Dependent Quantum Money]\label{def:memory_dependent}
    A quantum money scheme is \emph{memory-dependent} if the bank is required to maintain a state it uses throughout different runs of the quantum money protocols.
\end{definition}

\begin{definition}[Security against sabotage]\label{def:sec_sabotage}
    An interactive public quantum money scheme $\$$ is \emph{secure against sabotage} if for every QPT counterfeiter $\adv$ there exists a negligible function $\negl$ such that:
    \begin{equation*}
        \Pr[\sabmoney{\adv}{\$} = 1] \leq \negl\;.
    \end{equation*}
    
    The money sabotaging game $\sabmoney{\adv}{\$}$:
    \begin{enumerate}
        \item The challenger runs $(pk, sk) \gets \keygen(\secparam)$ and sends $pk$ to $\adv$.
        \item The adversary outputs a money state $\ket{\$}$ and runs $\verify(\ket{\$})$ with the challenger two consecutive times.
        \item $\adv$ wins if the first verification accepts and the second rejects, in which case the result of the game is 1 (and otherwise it is 0).
    \end{enumerate}
\end{definition}

\begin{definition}
    \label{def:security_public_interactive_money}
    We say that an interactive public quantum money scheme $\$$ is secure if it is secure against sabotage (\cref{def:sec_sabotage}) and if for every QPT counterfeiter $\adv$ there exists a negligible function $\negl$ such that:
    \begin{equation*}
        \Pr[\counterfeit{\adv}{\$}{} = 1] \leq \negl\;.
    \end{equation*}
    
    The money counterfeiting game $\counterfeit{\adv}{\$}{}$:
    \begin{enumerate}
      \item The bank generates a key pair $(pk, sk) \sample \keygen(\secparam)$ such that $pk$ is publicly known.
      \item The bank and the counterfeiter interact. The counterfeiter can ask the bank to run $\mint_{k}(\cdot)$ and $\verify_{k}(\cdot)$ polynomially many times, in any order the counterfeiter wishes. The counterfeiter is not bound to following his side of the protocols honestly. The counterfeiter can keep ancillary registers from earlier runs of these protocols and use them in later steps. Let $w$ be the number of successful verifications, $\ell$ the number of times that mint was called by the counterfeiter and $v$ the number of times that verify was called by the counterfeiter.
      \item The bank outputs $(w,\ell,v)$.
    \end{enumerate}
    
    The value of the game is 1 iff $w > \ell$. In this case we sometimes simply say that the counterfeiter wins.
\end{definition}

We state here again that this security definition is stronger than the one previously used --- this is further explained in \cref{rmk:stronger_security_definition}.

Public semi-quantum money has two types of verifications: the first is a quantum verification algorithm we denote $\qverify$ that does not destroy the banknote and can be performed between any two entities with quantum computation and communication resources, and the second is a classical verification protocol we denote $\cverify$ that destroys the money and can only be performed with a bank. $\qverify$ allows money transfer between any two users, which could be performed without the aid of a bank. $\cverify$ can be thought of as depositing or spending the money with the bank --- a quantum user can destroy his banknote in a way that he can prove to the bank that the note was destroyed. In the counterfeiting game above, the counterfeiter can choose which type of verification to run each time. Note that if $\qverify$ is chosen, the bank must perform the verification; i.e., the adversary must give the banknote to the bank (and the bank does not return it afterwards).

Note that the bank is no different to a user when participating in $\qverify$, since no secret information is used there --- so if $\qverify$ is secure with a bank, it is also secure with other users (i.e., any counterfeiting method that would work against a user would also work against the bank). Therefore the above definition of security ensures the scheme is secure even when considering transactions between users, even though the security game is defined with the bank.

Note that if two users wish to transfer a banknote but do not share a quantum communication channel, they can do so via a bank: the Payer would perform $\cverify$ with the bank, after which the bank will run $\mint$ with the Receiver (see \cref{fig:semi_quantum_transaction_bank}).

\begin{definition}[Public semi-quantum money]
    \label{def:public_semi_quantum}
    We say that a protocol with the bank has classical minting (resp. verification) if the bank is classical in $\mint$ (resp. $\verify$). We define public semi-quantum money as any secure interactive public quantum money scheme that has classical minting and two types of verification, denoted $\qverify$ and $\cverify$, such that in the end of $\qverify$ the banknote is not destroyed (meaning it could pass further verifications), and that $\cverify$ is a classical verification which destroys the banknote and can only be performed between a user and a bank. In $\qverify$ no party receives the private key, and in $\cverify$ the bank receives the private key while the user does not.
\end{definition}

\subsection{Construction of a Public Semi-Quantum Scheme}
\label{sec:public_money_construction}
Following is an informal explanation of our public scheme construction, which is defined formally in \cref{alg:public_scheme}. The construction uses a secure quantum lightning scheme with bolt-to-certificate ---  a brief informal overview of both is given in \cref{sec:introduction} below \cref{thm:main_public_theorem}, and their formal definitions (taken almost verbatim from their respective papers) can be found on \cref{sec:ql_with_bolt_to_cert}.

Let \nom{$\ql$}{A quantum lightning scheme}{QL} be a secure quantum lightning scheme with bolt-to-certificate (see \cref{def:quantum_lightning,def:ql_security,def:bolt_to_certificate}), and let \nom{$\signscheme$}{A digital signature scheme}{DS} be a PQ-EU-CMA digital signature scheme (see \cref{def:digital_signature,def:unforgeability_of_digital_signature}). In $\keygen$ the bank runs the $\ql$ setup algorithm to randomly generate a set of algorithms that generate and verify bolts and certificates, and runs $\signscheme.\keygen$ to generate a private and public key. $\keygen$ outputs the public digital signature key and the four $\ql$ algorithms as the scheme's public key, and the digital signature's private key as the scheme's private key. For the minting process, the user generates a bolt $\ket{\psi}$ along with its serial key $s$ using the generation algorithm $\genbolt$ (that is part of the public key). He sends $s$ to the bank, which sends back a signature $\sigma$ for it (using its private key). Each banknote $\ket{\$}$ consists of the bolt $\ket{\psi}$, its serial number $s$, and its signature $\sigma$ (without the signature any user could generate by themselves notes that would pass verification).

In $\qverify$ the Payer sends to the Receiver his bolt along with its signed serial number, and the receiver verifies both the bolt and the signature. In $\cverify$ the user uses the bolt to generate a certificate and sends the certificate to the bank along with the bolt's signed serial number. The bank verifies the certificate and the signature, and checks its database for the serial number. If the serial number appears there then the certificate has been given before and so verification will fail.

Note that according to \cref{def:public_semi_quantum}, in a general public semi-quantum money scheme the bank uses the secret key $sk$ in $\cverify$, but in our scheme the bank does not --- meaning it is stronger in this respect. \cite{BS16} also have classical verification with only a public key, so this is not the first instance where $\cverify$ uses only the public key.

\begin{algorithm*}
    \caption{The Interactive Public Money Scheme \nom{$\$_P$}{Our public semi-quantum money scheme based on quantum lightning with bolt-to-certificate}{DOLLAR_P}}
    \label{alg:public_scheme}
    
    \fbox{
    \procedure[linenumbering]{$\$_P.\keygen(\secparam)$}{
        (\genbolt, \verbolt, \gencert, \vercert) \gets \ql.\setup(\secparam)\\
        (pk_\sigma, sk_\sigma) \gets \signscheme.\keygen(\secparam) \\
        pk \gets (\genbolt, \verbolt, \gencert, \vercert, pk_\sigma) \\
        sk \gets sk_\sigma \\
        \textbf{return} \; (pk, sk)
    }
    }
    
    \fbox{
    \pseudocode[head = $\$_p.\mint_{(pk, sk)}$]{
    \textbf{Acquirer} \<\< \textbf{Bank} \\[] [\hline]
    \pcln \protect{\label{dollar_p_mint_bolt}} (\ket{\psi}, s) \gets \ql.\genbolt(\secparam)\\
    \pcln \< \sendmessageright*{s} \< \\
    \pcln \protect{\label{dollar_p_mint_signature}} \<\< \sigma \gets \signscheme.\sign_{sk_\sigma}(s) \\
    \pcln \< \sendmessageleft*{\sigma} \< \\
    \pcln \ket{\$} \gets (\ket{\psi}, s, \sigma)
    }
    }
    
    \fbox{
    \procedure[linenumbering]{$\$_p.\qverify_{pk}(\ket{\$})$}{
    \protect{\label{dollar_p_qverify_input}} \text{interpret $\ket{\$}$ as } (\ket{\psi}, s, \sigma) \\
    \protect{\label{dollar_p_qverify_signature}} r_\sigma \gets \signscheme.\verify_{pk_\sigma}(s, \sigma) \\
    \protect{\label{dollar_p_qverify_bolt}} r_b \gets \ql.\verbolt(\ket{\psi}, s) \\
    \protect{\label{dollar_p_qverify_out}} \textbf{return} \; r_\sigma \cdot r_b
    }
    }
    
    \fbox{
    \pseudocode[head = $\$_p.\cverify_{pk}(\ket{\$})$]{
    \text{\Comment{$D \gets \varnothing$ before first run}}\\
    \textbf{Payer} \<\< \textbf{Bank} \\[] [\hline]
    \pcln \text{interpret $\ket{\$}$ as } (\ket{\psi}, s, \sigma) \\
    \pcln \protect{\label{dollar_p_cverify_gencert}} c \gets \ql.\gencert(\ket{\psi}, s) \\
    \pcln \protect{\label{dollar_p_cverify_input}} \< \sendmessageright*{s, \sigma, c} \< \\
    \pcln \protect{\label{dollar_p_cverify_signature}} \<\< r_\sigma \gets \signscheme.\verify_{pk_\sigma}(s, \sigma) \\
    \pcln \protect{\label{dollar_p_cverify_vercert}} \<\< r_c \gets \ql.\vercert(s, c) \\
    \pcln \<\< r_d \gets s \notin D \\
    \pcln \<\< \pcif r_\sigma \cdot r_c \cdot r_d = 1: \\
    \pcln \<\< \;\;\;\; D \gets D \cup \{s\} \\
    \pcln \<\< \textbf{return} \; r_\sigma \cdot r_c \cdot r_d
    }
    }
\end{algorithm*}

\begin{proposition}[Correctness of $\$_P$]
Assuming $\ql$ is a secure quantum lightning scheme with bolt-to-certificate (see Definitions \ref{def:quantum_lightning} and \ref{def:bolt_to_certificate}) and $\signscheme$ is a digital signature scheme with \emph{perfect completeness} (see \cref{def:digital_signature}), $\$_P$, which is defined in \cref{alg:public_scheme}, is a \emph{correct} (\cref{def:interactive_public_money}) \emph{memory-dependent} (\cref{def:memory_dependent}) public semi-quantum money scheme (see \cref{def:public_semi_quantum}).
\end{proposition}

\begin{proof}
    Let $(\genbolt, \verbolt, \gencert, \vercert) \gets \ql.\setup(\secparam)$ and assume an honest run of $\ket{\$} \gets \$_P.\mint_{(pk, sk)}()$. We need to prove correctness for both implementations of $\verify$: we begin by proving the correctness of $\qverify$. In an honest run of $\$_P.\qverify$, the banknote sent on line \ref{dollar_p_qverify_input} is $\ket{\$}$ generated in $\mint$. Recall that $\ket{\$} \coloneqq (\ket{\psi}, s, \sigma)$. From the \emph{perfect completeness} of $\signscheme$ we get:
    \begin{equation*}
        \Pr[\signscheme.\verify_{pk_\sigma}(s, \signscheme.\sign_{sk_\sigma}(s)) = 1] = 1\;.
    \end{equation*}
    So the signature verification on line \ref{dollar_p_qverify_signature} passes since $\sigma \gets \signscheme.\sign_{sk_\sigma}(s)$ on line \ref{dollar_p_mint_signature} of $\$_P.\mint$, meaning $\Pr[r_\sigma = 1] = 1$. From the definition of quantum lightning (\cref{def:quantum_lightning}) we get:
    \begin{align*}
        \Pr[(\genbolt, \verbolt) &\gets \ql.\setup(\secparam); (\ket{\psi}, s) \gets \genbolt() : \\
        &\verbolt(\ket{\psi}, s) = 1] \\
        & = 1 - \negl\;.
    \end{align*}
    So the bolt verification on line \ref{dollar_p_qverify_bolt} passes except with negligible probability since $(\ket{\psi}, s)$ is a valid bolt generated on line \ref{dollar_p_mint_bolt} of $\$_P.\mint$, meaning $\Pr[r_b = 1] = 1 - \negl$. Therefore, from the union bound:
    \begin{align*}
        \Pr[\ql.\qverify_{pk}(\ket{\psi}) = 0] &\leq \Pr[r_\sigma = 0] + \Pr[r_b = 0] \\
        &= \negl
    \end{align*}
    meaning $\Pr[\ql.\qverify_{pk}(\ket{\psi}) = 1] = 1 - \negl$.
    
    We now prove correctness of $\$_P.\cverify$. Like with $\qverify$, from the \emph{perfect completeness} of $\signscheme$ we get that the signature verification on line \ref{dollar_p_cverify_signature} passes, meaning $\Pr[r_\sigma = 1] = 1$. From the definition of bolt-to-certificate (\cref{def:bolt_to_certificate}) we get that the certificate verification on line \ref{dollar_p_cverify_vercert} passes except with negligible probability, since $c$ was generated on line \ref{dollar_p_cverify_gencert} using the valid bolt $(\ket{\psi}, s)$ that was generated in line \ref{dollar_p_mint_bolt} of $\$_P.\mint$, meaning $\Pr[r_c = 1] = 1 - \negl$. From the security of quantum lightning (\cref{def:ql_security}) we get that for each $s' \in D$, $\Pr[s = s']$ is negligible\footnote{If there is non-negligible probability that a generated bolt would have the same serial number as one already in $D$, assuming $D$ contains an amount of serial numbers polynomial in $\secpar$, we could construct a bolt forger $\boltforger$ that would generate $\abs{D}+1$ bolts and with non-negligible probability end up with two bolts with the same serial number that pass $\verbolt$, winning the bolt forging security game (\cref{def:ql_security}).}, meaning $\Pr[r_d = 1] = 1 - \negl$ assuming D is polynomial in $\secpar$. Therefore, from the union bound:
    \begin{align*}
        \Pr[\ql.\cverify_{pk}(\ket{\psi}) = 0] &\leq \Pr[r_\sigma = 0] + \Pr[r_c = 0] + \Pr[r_d = 0] \\
        &= \negl\;.
    \end{align*}
\end{proof}

\begin{theorem}[Security of $\$_P$]
    \label{thm:public_scheme_security}
    Assuming $\ql$ is a secure quantum lightning scheme (Definitions \ref{def:quantum_lightning} and \ref{def:ql_security}) with bolt-to-certificate capability (\cref{def:bolt_to_certificate}) and that $\signscheme$ is a PQ-EU-CMA digital signature scheme (Definitions \ref{def:digital_signature} and \ref{def:unforgeability_of_digital_signature}), then $\$_P$, which is defined in \cref{alg:public_scheme}, is secure according to \cref{def:security_public_interactive_money}. If $\ql$ is secure against sabotage (\cref{def:sec_sabotage}), $\$_P$ is also secure against sabotage.
\end{theorem}

\begin{proof}
    In order to prove the security of the scheme, we need to prove that no counterfeiter $\adv$ can win the security game $\counterfeit{\adv}{\$_P}{}$ (\cref{def:security_public_interactive_money}) with non-negligible probability. In the game, the counterfeiter has access to a verification oracle, meaning that in our case the counterfeiter can run either $\qverify$ or $\cverify$ a polynomial amount of times, and wins if he manages to pass a total of more than $\ell$ $\qverify$ and $\cverify$ verifications (when $\ell$ is the number of times $\mint$ was run).
    
    We show that any adversary capable of breaking the security of $\$_P$ must break the underlying security of either $\ql$, $\ql$'s bolt-to-certificate capability, or $\signscheme$. Assume a QPT counterfeiter $\adv$ with non-negligible success probability. Recall that $w$, $\ell$ and $v$ are the numbers of successful verifications, runs of $\$_P.\mint$ and runs of $\$_P.\verify$ in the counterfeiting security game, respectively ($w$ and $v$ include runs of $\qverify$ and of $\cverify$). This means that $w > \ell$ with non negligible probability. We assume $w$ and $v$ are polynomial in $\secpar$, otherwise the counterfeiter would not be QPT. Denote by $\ket{\$_j} = (\ket{\psi_j}, s_j, \sigma_j)$ the banknote minted in the $j^{th}$ run of $\$_P.\mint$\footnote{We arbitrarily number the runs of $\$_P.\mint$ according to the order they were initiated.}. Due to the unforgeability of $\signscheme$ (\cref{def:unforgeability_of_digital_signature}), for every successful verification, $\adv$ either sent $\ket{\$} = (\ket{\psi}, s_j, \sigma)$ for some $j \in [\ell], \ket{\psi}, \sigma$ on line \ref{dollar_p_qverify_input} of $\qverify$ or $s_j, \sigma, c$ for some $j \in [\ell], \sigma, c$ on line \ref{dollar_p_cverify_input} of $\cverify$\footnote{Suppose $\adv$ passes with non-negligible probability a verification of some $(\ket{\psi}, s, \sigma)$ on $\qverify$ or a verification of $s, \sigma, c$ on $\cverify$ such that $s \neq s_j \; \forall j \in [\ell]$. In that case $\signscheme.\verify_{pk_\sigma}(s, \sigma) = 1$ with non-negligible probability (either on line \ref{dollar_p_qverify_signature} of $\qverify$ or on line \ref{dollar_p_cverify_signature} of $\cverify$). We could use $\adv$ to construct a digital signature forger $\forger$ with non-negligible success probability: $\forger$ simulates a bank, but instead of signing with $sk_\sigma$ and verifying with $pk_\sigma$ that were generated in $\keygen$ he uses his signing oracle and the real $pk$. He then runs $\adv$ against the simulated bank, and will be able to present $(s, \sigma)$ which pass $\signscheme$ verification with non-negligible probability, while $s \notin Q$ since $s \neq s_j \; \forall j \in [\ell]$.}. Denote by $w_j$ the amount of successful verifications --- either $\qverify$ or $\cverify$ --- made with $s_j$ as input. By the pigeonhole principle, since there are only $\ell$ $s_j$'s but $w > \ell$ successful verifications, $w_i \geq 2$ for some $i \in [\ell]$. For that $i$, there are three possibilities: either the two verifications were $\qverify$, the two verifications were $\cverify$, or there was one of each.
    
    Two successful $\cverify$ runs with the same $s_i$ are not possible, since after each successful run of $\cverify$ the bank adds $s_i$ to $D$, and every subsequent verifications with $s_i \in D$ fails. One successful run of $\qverify$ and one successful run of $\cverify$ are also not possible --- if it were, we could construct a certificate forger $\certforger$ with non-negligible probability: $\certforger$ will simulate the bank and run $\adv$ against it. With non-negligible probability, he will succeed both on a run of $\qverify$ and on a run of $\cverify$: meaning he possesses a quantum state $\ket{\psi}$ such that $\ql.\verbolt(\ket{\psi}, s_i) = 1$ and a certificate $c$ such that $\ql.\vercert(s_i, c) = 1$, in contradiction to the bolt-to-certificate capability of $\ql$ (\cref{def:bolt_to_certificate}).
    
    The only option left is two successful runs of $\qverify$: in which case $\adv$ is in possession of two quantum states $\ket{\psi_1}, \ket{\psi_2}$ such that $\ql.\verbolt(\ket{\psi_1}, s_j)=1$ and $\ql.\verbolt(\ket{\psi_2}, s_j)=1$ (he could not have passed the two verifications with the same quantum state since $\qverify$ entails sending the quantum state to the bank --- hence he necessarily possess two such states). We could then construct a quantum lightning adversary $\boltforger$ that simulates the bank and runs $\adv$ against it. With non-negligible probability, he will end up with some $\ket{\psi_1}$ and $\ket{\psi_2}$ that pass $\ql.\verbolt$ with the same serial number $s_i$, and could use them to win $\boltforge{\ql}$ (see \cref{def:ql_security}), in contradiction to the security of $\ql$.
    
    None of the three options are possible, meaning that any counterfeiter $\adv$ has negligible success probability, i.e., $\$_P$ is secure.
    
    Security against sabotage of $\ql$ directly implies security against sabotage of $\$_P$.
\end{proof}

For convenience, we restate the main theorem of our public semi-quantum money result:

\publicresult*

We can see that \cref{thm:public_scheme_security} proves our main public result.

\fi

%%%%%%%%%%%%%%%%%%%%%%%%%%%%%%%%%%%%%%%%%%%
\section{Trapdoor Claw-Free Families and 1-of-2 Puzzles}
\label{sec:ntcf_and_1-of-2_puzzles}
%%%%%%%%%%%%%%%%%%%%%%%%%%%%%%%%%%%%%%%%%%%
In this section, as the name suggests, we discuss the concepts of NTCF and  1-of-2 puzzles. \ifnum\sigconf=0 For completeness, we restate the formal definition of NTCF by Brakerski et al. in ~\cref{sec:definition_NTCF}. \fi In Section~\ref{sec:1_of_2_puzzles}, we introduce 1-of-2 puzzles. In Section~\ref{sec:NTCF_implies_1_of_2_puzzle} we show how to construct a 1-of-2 puzzle using an NTCF, and in Section~\ref{sec:parallel_repetition_theorem_for_1_of_2_puzzles} we show a parallel repetition theorem for 1-of-2 puzzles that is subsequently used to construct strong 1-of-2 puzzles.

%%%%%%%%%%%%%%%%%%%%%%%%%%%
\subsection{1-of-2 Puzzles}
%%%%%%%%%%%%%%%%%%%%%%%%%%%
\label{sec:1_of_2_puzzles}

\begin{definition}[1-of-2 puzzle]
A 1-of-2 puzzle consists of four efficient algorithms: the puzzle generator $G$, an obligation algorithm $O$, a 1-of-2 solver $S$, and a verification algorithm $V$. $G$ is a classical algorithm, $V$ is a classical deterministic algorithm, and $O$ and $S$ are quantum algorithms. 
\begin{enumerate}
	\item $G$ outputs, on security parameter $\secparam$, a random puzzle $p$ and some verification key $v$: $ (p,v)\sample G(\secparam) $.
	\item $O$ receives a puzzle $p$ as input and outputs a classical string $o$ called the obligation and a quantum state $\rho$: $(o,\rho) \sample O(p)$.
	\item $S$ receives $p,o,\rho$ and a bit $b\in \{0,1\}$ as input and outputs an answer string $a$: $a \sample S(p,o,\rho,b) $.
	\item $V$ receives $p,v,o,b,a$ as input and outputs $0$ or $1$: $V(p,v,o,b,a)\in \{0,1\}$.
\end{enumerate}

Completeness: Let $\eta$ be some arbitrary function $\eta:\NN\mapsto [0,1]$. We say that the 1-of-2 puzzle has completeness $\eta$ if there exists a negligible function $\negl$ such that
\begin{align}
\begin{split}
	\Pr[(p,v) &\sample G(\secparam),(o,\rho) \sample O(p),b\sample \{0,1\},a \sample S(p,o,\rho,b):\\
	&V(p,v,o,b,a)=1]\\
	&\geq \eta(\secpar) - \negl\;.
	\label{eq:1-of-2-completeness}
\end{split}
\end{align}

We define the \nom{$V_2$}{An algorithm that verifies solutions for both challenges of a puzzle}{V2} algorithm as:
\begin{equation}
	V_2(p,v,o,a_0,a_1)=V(p,v,o,0,a_0)\cdot V(p,v,o,1,a_1)\;.
	\label{eq:V_2_definition}
\end{equation}

Hardness: Let $h:\NN \mapsto [0,1]$ be an arbitrary function. We say that the 1-of-2 puzzle $\mathcal{Z}$ is $1-h$-hard if for any QPT 2-of-2 solver \nom{$\solvertwo$}{A QPT 2-of-2 solver --- an adversary that attempts to solve both challenges of a 1-of-2 puzzle}{S} there exists a negligible function $\negl$ such that
\begin{equation}
\Pr[\solveboth{\mathcal{Z}}=1] \leq h(\secpar) + \negl\;.
\label{eq:succ_2_upperbound}
\end{equation}

The 2-of-2 solving game $\solveboth{\mathcal{Z}}$:
\begin{enumerate}
    \item The puzzle giver runs $(p,v) \sample G(\secparam)$
    \item The 2-of-2 solver $\solvertwo$ receives input $p$ and outputs a triple $(o, a_o, a_1)$
    \item The puzzle giver runs $r \gets V_2(p,v,o,a_0,a_1)$ and outputs $r$
    \item $\solvertwo$ wins the game if and only if $r=1$, in which case the output of the game is defined to be 1.
\end{enumerate}

We say that the 1-of-2 puzzle is strong if $\eta=1$ and $h=0$ (i.e., the puzzle is 1-hard). We say that the 1-of-2 puzzle is weak if $\eta=1$ and $1-h$ is noticeable.

\label{def:1-of-2_puzzle}
\end{definition}

\subsection{An NTCF Implies a 1-of-2 Puzzle}
\label{sec:NTCF_implies_1_of_2_puzzle}

This section presents how an NTCF can be used to construct a 1-of-2 puzzle. The formal definition of an NTCF and its properties used in this section can be found in \cref{sec:definition_NTCF}, taken from~\cite{BCM+18}. 

\begin{theorem}
An NTCF implies a 1-of-2 puzzle with completeness $\eta=1$ and hardness $h=\frac{1}{2}$. 
\label{thm:NTCF_implies_1-of-2_puzzles}
\end{theorem}
Note that the 1-of-2 puzzle above is a weak 1-of-2 puzzle.
\begin{proof}
The proof contains arguments similar to those used by Brakerski et al.~\cite{BCM+18}.

Given an NTCF family $ \mathcal F $ that consists of the algorithms 
\begin{equation*}
    \keygen_{\mathcal F}, \textrm{Inv}_{\mathcal F}, \textrm{CHK}_{\mathcal F}, \textrm{SAMP}_{\mathcal F}, J_{\mathcal F}    
\end{equation*}
we construct the 1-of-2 puzzle \nom{$ \mathcal Z$}{A 1-of-2 puzzle}{Z} $= (\keygen_\mathcal Z, O_\mathcal Z,S_\mathcal Z,V_\mathcal Z)$ as specified in Algorithm~\ref{alg:1-of-2-puzzle_from_NTCF}.
\begin{algorithm*}
    \caption{The 1-of-2 Puzzle $\mathcal Z$}
    \begin{algorithmic}[1] % The number tells where the line numbering should start
        \Procedure{$\keygen_{ \mathcal Z}$}{$\secpar$}
            \State $(k,t_k) \sample\keygen_{\mathcal F}(\secpar)$
            \State Set $p \equiv k$, $v \equiv t_k$ 
            \State \textbf{return} $(p,v)$
        \EndProcedure

    \end{algorithmic}
    \begin{algorithmic}[1] % The number tells where the line numbering should start
        \Procedure{$O_{ \mathcal Z}$}{$p$}
            \State $\ket{\psi'} \sample \textrm{SAMP}_{\mathcal F}(p,\ket{+})$\label{line:O_SAMP}
            \State Measure the last register to obtain an $o \in \mathcal Y$. Denote the post-measurement state $\ket{\psi}$\label{line:O_measure} \Comment{In the completeness we show that  $\ket{\psi}\approx \frac{1}{\sqrt 2} (\ket{0} \ket{x_0}+\ket{1}\ket{x_1})$.}
            \State \textbf{return} $(o,\ket{\psi})$
        \EndProcedure

    \end{algorithmic}
	\begin{algorithmic}[1] % The number tells where the line numbering should start
        \Procedure{$S_{ \mathcal Z}$}{$p,o,\ket{\psi},b$} \Comment{$p$ and $o$ are not used in this construction.}
            \If {$b=0$}
            	\State Measure both registers of $\ket{\psi}$ to obtain a result $i \in \{0,1\}$ and $x \in \mathcal X$.
            	\State Set $a \equiv (i,x)$
            \ElsIf{$b=1$}
            	\State Evaluate the function $J$ on the second register of $\ket{\psi}$, and apply Hadamard transform on both registers.\label{line:apply_Hadamard} 
            	\State Measure both registers to obtain the result $i \in \{0,1\}$ and $d$.\label{line:S_measurement_b_0}
            	\State Set $a \equiv (i,d)$
            \EndIf
            \State \textbf{return} $a$
        \EndProcedure

    \end{algorithmic}
    \begin{algorithmic}[1] % The number tells where the line numbering should start
        \Procedure{$V_{ \mathcal Z}$}{$p,v,o,b,a$}
        	\State Set $x_0\equiv \mathrm{INV}_{\mathcal F}(v,0,o)$ and $x_1\equiv \mathrm{INV}_{\mathcal F}(v,1,o)$ \Comment{Recall that $v$ is the trapdoor, and $o$ is an image of the NTCF.}
            \If{$b=0$}
            	\State Interpret $a$ as $i,x$
            	\If{$x=x_i$} 
            		\State \textbf{return} $1$ \label{line:V_b_equal_1_accept}
            	\Else
            		\State \textbf{return} $0$
            	\EndIf
            \ElsIf{$b=1$}
            	\State Interpret $a$ as $i,d$.
            	\If{$d \in G_{p,0,x_0}\cap G_{p,1,x_1}$ and $d \cdot (J(x_0) \oplus J(x_1))=i$}\label{line:V_if_for_acceptance_b_0} \Comment{This membership test uses $\textrm{CHK}_{\mathcal F}$}
            		\State \textbf{return} $1$
            	\Else
            		\State \textbf{return} $0$
            	\EndIf
            \EndIf 
       	\EndProcedure

    \end{algorithmic}
    \label{alg:1-of-2-puzzle_from_NTCF}
\end{algorithm*}

Completeness: we need to show that Eq.~\eqref{eq:1-of-2-completeness} holds for $\mathcal Z$ defined above.
By \ifnum\sigconf=1 the efficient range superposition property of NTCF~\cite{RS19}, \else NTCF property~\ref{it:NTCF_efficient_range_superposition}, \fi the state $\ket{\psi'}$ in line~\ref{line:O_SAMP} of the algorithm $O_{\mathcal Z}$ is negligibly close in trace distance to:
\begin{equation*}
\ket{\tilde \psi}=\frac{1}{\sqrt{ 2|\mathcal X|}} \sum_{x \in \mathcal X, y\in \mathcal Y, b \in \{0,1\}} \sqrt{(f_{k,b}'(x))(y)}\ket{b,x}\ket{y}\;.
\end{equation*}
For the sake of the analysis, therefore, we can replace $\ket{\psi}$ with $\ket{\tilde \psi}$, and the algorithm will behave the same, up to a negligible probability. By \ifnum\sigconf=1 the injective pair property of NTCF, \else NTCF property~\ref{it:NTCF_injective_pair}, \fi the post-measurement state $\ket{\psi}$ generated by $O_{\mathcal Z}$ is 
$\ket{\psi}=\frac{1}{\sqrt 2} (\ket{0} \ket{x_0}+\ket{1}\ket{x_1})$, where $(x_0,x_1)\in \mathcal{R}_p$. Since $o$ was the outcome of the measurement in line~\ref{line:O_measure}, we know that $o \in \supp{f_{p,i}(x_i)}$. By \ifnum\sigconf=1 the trapdoor property of NTCF, \else NTCF property~\ref{it:NTCF_trapdoor}, \fi for $i\in \{0,1\}$:
\begin{equation}
x_i=\mathrm{INV}_{\mathcal F}(v,i,o)\;.
\label{eq:inv_equal_x_i}
\end{equation}

Consider the case $b=0$. In this case, the output of $S_{\mathcal Z}$ is $a\equiv (i,x_i)$, where, by Eq.~\eqref{eq:inv_equal_x_i}, $x_i=\mathrm{INV}_{\mathcal F}(v,i,o)$. Therefore, $V_{\mathcal Z}$ will return $1$ in line~\ref{line:V_b_equal_1_accept}. In the case of $b=1$, before line~\ref{line:apply_Hadamard} in $S_{\mathcal Z}$ the state is $\frac{1}{\sqrt 2} (\ket{0} \ket{x_0}+\ket{1}\ket{x_1})$, after the evaluation of $J$ on the second register the state is  $\frac{1}{\sqrt 2} \sum_{j\in \{0,1\}}\ket{j} \ket{J(x_j)}$, and after the Hadamard on both registers (which consist of $w+1$ qubits), the state is 
\begin{align*}
&\frac{1}{\sqrt {2^{w+2}}} \sum_{i \in \{ 0,1\}, d\in \{0,1\}^w } \left(\sum_{j \in \{0,1\} } (-1)^{ij + d\cdot J(x_j) }\right)\ket{i}\ket{d} \\
=& \frac{1}{\sqrt {2^{w}}} \sum_{d\in \{0,1\}^w } (-1)^{d\cdot J(x_0)}\ket{d\cdot (J(x_0) \oplus J(x_1)) }\ket{d}\;.
\end{align*}
Therefore, the outcome of the measurement in line~\ref{line:S_measurement_b_0} will provide a random $d\in \{0,1\}^w$ and an $i\in \{0,1\}$ that satisfy $i=d\cdot (J(x_0) \oplus J(x_1))$. Since $d$ is random, \ifnum\sigconf=0 property \ref{it:NTCF_adaptive_hardcore_bit_a} \else the adaptive hardcore bit property \fi guarantees that the first condition in line~\ref{line:V_if_for_acceptance_b_0} of $V_{\mathcal Z}$ will be met (up to a negligible probability), and the analysis in the previous sentence guarantees that the second condition will be met. Overall, the probability that $V_{\mathcal Z}$ outputs $1$ is $1-\negl$, for some negligible function $\negl[]$, as required. 

Soundness: We need to show that Eq.~\eqref{eq:succ_2_upperbound} holds for every QPT $\solvertwo$. In Algorithm~\ref{alg:adversary}, we show a reduction that maps a 2-of-2 solver $\solvertwo$ for the 1-of-2 puzzle as in Eq.~\eqref{eq:succ_2_upperbound} to an NTCF adversary \ifnum\sigconf=0 $\adv$ as in Eq.~\eqref{eq:adaptive-hardcore}. \else $\adv$. \fi   

\begin{algorithm}
    \caption{The Adversary $\adv$}
    \begin{algorithmic}[1] % The number tells where the line numbering should start
        \Procedure{$\adv_{\mathcal F}$}{$k$}
            \State $(o,a_0,a_1) \sample \solvertwo(k)$
            \State Interpret $a_0$ as $i,x$ and $a_1$ as $i',d$.
            \State \textbf{return} $(i,x,d,i')$
        \EndProcedure

    \end{algorithmic}
    \label{alg:adversary}
\end{algorithm}

If $\solvertwo$ succeeds with probability $\frac{1}{2}+\epsilon(\secpar)$ (where $\epsilon(\secpar)$ is not necessarily negligible), \ifnum\sigconf=0 then the l.h.s. in Eq.~\eqref{eq:adaptive-hardcore} is lower-bounded by $2\epsilon(\secpar)$ with respect to $\adv$. Plugging \else then from the adaptive hardcore bit property, plugging \fi the definition of $V_2$ (see Eq.~\eqref{eq:V_2_definition}) and the acceptance criteria of $V_{\mathcal Z}$ into lines~\ref{line:V_b_equal_1_accept} and~\ref{line:V_if_for_acceptance_b_0}, we see that the 2-of-2 solver $\solvertwo$ needs to find $o,i,x,d,i'$ such that $d\in G_{p,0,x_0} \cap G_{p,1,x_1}$ and $x=x_i$, where $x_0=\mathrm{INV}_{\mathcal F}(v,0,o)$, $x_1=\mathrm{INV}_{\mathcal F}(v,1,o)$ and $i'=d\cdot(J(x_0) \oplus J(x_1))$.
This implies the membership of $(i,x,d,i')$ in $H_k$ \ifnum\sigconf=0 (see Eq.~\eqref{eq:defsetsH}\fi).
Therefore, $\Pr_{(k,t_k)\sample \keygen_{\mathcal F}(\secparam)}[\adv(k)\in H_k]\geq \frac{1}{2} + \epsilon(\secpar)$. Since $H_k$ and $\overline{H}_k$ are disjoint, $\Pr_{(k,t_k)\sample \keygen_{\mathcal F}(\secparam)}[\adv(k)\in \overline{H}_k]\leq \frac{1}{2}- \epsilon(\secpar)$, and 
\begin{align*}
\Big|\Pr_{(k,t_k)\leftarrow \keygen_{\mathcal{F}}(\secparam)}&[\mathcal{A}(k) \in H_k] - \Pr_{(k,t_k)\leftarrow \keygen_{\mathcal{F}}(\secparam)}[\mathcal{A}(k) \in\overline{H}_k]\Big| \\&\geq\, 2\epsilon(\secpar)\;.
\end{align*}

 \ifnum\sigconf=0 Since by property \ref{it:NTCF_injective_pair} \else By the injective pair property \fi \ifnum\sigconf=0 (\cref{def:trapdoorclawfree}) the l.h.s. of Eq.~\eqref{eq:adaptive-hardcore} is upper-bounded by the negligible function $\mu(\secpar)$, \fi we conclude that $\epsilon(\secpar)$ must be negligible, as required for $h=\frac{1}{2}$. 

\end{proof}

% section strong_NTCF (end)

%%%%%%%%%%%%%%%%%%%%%%%%%%%%%%%%%%%%%%%%%%%%%%%%%%%%%%%%%%%%%%%%%%%%%%%%%%
\subsection{A Parallel Repetition Theorem for 1-of-2 Puzzles} % (fold)
%%%%%%%%%%%%%%%%%%%%%%%%%%%%%%%%%%%%%%%%%%%%%%%%%%%%%%%%%%%%%%%%%%%%%%%%%%
\label{sec:parallel_repetition_theorem_for_1_of_2_puzzles}
\begin{definition}[Parallel repetition of 1-of-2 puzzles]
Let $\mathcal{Z}$ be a 1-of-2 puzzle system, and let $n\in \NN$. We denote by $G^n$ the algorithm that, on security parameter $\secpar$, runs $G(\secparam)$ for $n(\secpar)$ times and outputs all the $n$ puzzles with their verification keys: 
\begin{equation}
	((p_1,\ldots,p_n)),(v_1,\ldots,v_n))\sample G^n(\secparam)
	\label{eq:repetition_G_n}
\end{equation}
where $(p_i,v_i)\sample G(1^{\secpar})$. A similar approach is used for all other algorithms in $\mathcal Z$:
\begin{equation}
	((o_1,\ldots,o_n)),(\rho_1\tensor \cdots \tensor \rho_n))\sample O^n(p_1,\ldots,p_n)
	\label{eq:repetition_O_n}
\end{equation}
where $(o_i,\rho_i)\sample O(p_i)$.
\begin{equation*}
	(a_1,\ldots,a_n)\sample S^n((p_1,\ldots,p_n),(o_1,\ldots,o_n),\rho_1\otimes\cdots\otimes \rho_n,b)
\end{equation*}
where $a_i \sample S(p_i,o_i,\rho_i,b)$. The algorithm
\begin{equation*}
	V^n((p_1,\ldots,p_n),(v_1,\ldots,v_n),(o_1,\ldots,o_n),b,(a_1,\ldots,a_n))
\end{equation*}
outputs $1$ iff for all $i \in$ \nom{$[n]$}{We denote $[n] \equiv \{1, \dots, n\}$}{N}, $V(p_i,v_i,o_i,b,a_i)=1$.

The $n$-fold parallel repetition of $\mathcal Z$ is the 1-of-2 puzzle
\begin{equation*}
    \mathcal Z^n=(G^n,O^n,S^n,V^n)\;.
\end{equation*}
\label{def:repetition}
\end{definition}
We emphasize that $\mathcal Z^n$ is a 1-of-2 puzzle (and not a 1-of-$2^n$ puzzle), which explains why the algorithm contains a single challenge bit $b$ rather than $n$ bits. The reason for that should be made clear later --- see Fact~\ref{fct:equivalence_of_hardness}.

It is natural to ask what is the hardness parameter of $\mathcal Z^n$, relative to the hardness of $\mathcal Z$. In the settings of 2-prover games: on the one hand, the soundness parameter decreases exponentially~\cite{Raz98}; yet the exponential decrease is not as fast as one would expect~\cite{Raz11}. 
We say the parallel repetition is \emph{perfect} if the fact that $\mathcal{Z}$ is $1-h$-hard means that $\mathcal{Z}^n$ is $1-h^n$-hard.
Note that a \emph{perfect} parallel repetition means that the adversary can do no better than solving each 1-of-2 puzzle independently (up to a negligible correction). This is indeed the case for 1-of-2 puzzles:

\begin{theorem}[Perfect parallel repetition of 1-of-2 puzzles]
Let $\mathcal Z$ be a 1-of-2 puzzle with completeness $\eta$ and hardness parameter $h$. For a function $n(\secpar)$ that satisfies $n(\secpar)=\poly$, the 1-of-2 puzzle $\mathcal Z^n$ has completeness $\eta^n$ and hardness parameter $h^n$.
\label{thm:parallel_repetition_1_of_2_puzzle}
\end{theorem}
\begin{proof}

First we prove the completeness property (see Eq.~\eqref{eq:1-of-2-completeness}). For ease of notation, we write $n,\negl[],\eta$ ,etc., instead of $n(\secpar)$, $\negl$, $\eta(\secpar)$. Suppose that the success probability of $\mathcal Z$ is $\eta - \negl[]$ for some negligible function $\negl[]$. Since the repeated game $\mathcal Z^n$ is an independent repetition of $\mathcal Z$, its success probability is $(\eta - \negl[])^n$. We show that for the negligible function $\negl[]'=n^2 \negl$, indeed $(\eta - \negl[])^n\geq \eta^n - \negl[]'$:
\begin{align*}
(\eta - \negl[])^n &= \eta^n + \sum_{k=1}^{n} (-1)^k \binom{n}{k}\eta^{n-k}\negl[]^k\\
&\geq \eta^n - \sum_{k=1}^n n^k \negl[]^k \\
&\geq \eta^n - \sum_{k=1}^n n \cdot \negl[]  = \eta^n-\negl[]'\;,
\end{align*}
where the last inequality holds for all $\secpar \geq \secpar_0$ (where $n \cdot \negl[] \leq 1 $).

We are now ready to prove the soundness. Our main tool is the notion of a weakly verifiable puzzle system defined by Canetti, Halevi and Steiner:
\begin{definition}[A weakly verifiable puzzle, adapted from~\cite{CHS05}]
A system for weakly verifiable puzzles consists of a pair of efficient classical algorithms \nom{$\hat{\mathcal Z}$}{A weakly verifiable puzzle}{Z} $= (G,V)$ such that
\begin{enumerate}
	\item The puzzle generator $G$ outputs, on security parameter $\secpar$, a random puzzle $p$ along with some verification  information $v$: $ (p,v)\sample G(\secparam) $.
	\item The puzzle verifier $V$ is a deterministic efficient classical algorithm that, on input of a puzzle $p$, verification key $v$, and answer $a$, outputs either zero or one: $V(p,v,a)\in \{0,1\}$.
\end{enumerate}
% A \emph{solver} for this puzzle system is a QPT algorithm $S$ that gets a puzzle $p$ as input and outputs an answer $a$. The \emph{success probability} of $S$ is the probability that the answer is accepted by the puzzle verifier:
% \begin{equation*}
% succ_{\hat{\mathcal Z}}[S]\equiv \Pr\left[ (p,c)\sample G(\secparam), a \sample S(p): V(p,v,a)=1 \right]
% \end{equation*}
\label{def:puzzle}
\end{definition}
The hardness of a weakly verifiable puzzle is defined as follows:
\begin{definition}[Hardness of a weakly verifiable puzzle, adapted from~\cite{CHS05}]
Let $h:\NN \mapsto [0,1]$ be an arbitrary function. A weakly verifiable puzzle $\hat{\mathcal Z}$ is said to be $1-h$-hard if, for any QPT\footnote{In Ref.~\cite{CHS05} this notion is defined for any PPT algorithm.} algorithm $S$, there exists a negligible function $\negl$ such that:
\begin{equation*}
    \Pr[\solve{\hat{\mathcal{Z}}}] \leq h(\secpar) + \negl\;.
\end{equation*}
\label{def:hardness_puzzle}
\end{definition}

The event $\solve{\hat{\mathcal{Z}}}$ is defined by the following security game:
\begin{enumerate}
    \item The puzzle giver runs $(p,v) \sample G(\secparam)$.
    \item The solver $\solver$ is given input $p$ and outputs an answer $a$.
    \item The puzzle giver runs $r \gets V(p,v,a)$. The event $\solve{\hat{\mathcal{Z}}}$ is when $r=1$.
\end{enumerate}

To avoid confusion, we always use $\mathcal{Z}$ to denote a 1-of-2 puzzle and $\hat{\mathcal Z}$ to denote a weakly verifiable puzzle. 

\begin{definition}[Repetition of weakly verifiable puzzles, from~\cite{CHS05}]
Let $\hat{\mathcal Z} = (G,V)$ be a weakly verifiable puzzle system, and let $n:\NN \mapsto \NN$ be some function. We denote by $G^n$ the algorithm that, on security parameter $\secpar$, runs $G(\secparam)$ for $n(\secpar)$ times and outputs all the $n$ puzzles with their respective verification keys: 
\begin{equation*}
	((p_1,\ldots,p_n)),(v_1,\ldots,v_n))\sample G^n(\secparam)
\end{equation*}
where $(p_i,v_i)\sample G^n(\secparam)$. $V^n$ receives $n$ inputs and accepts if and only if all $n$ runs of $V$ accept:
\begin{equation*}
	 V^n((p_1,\ldots,p_n),(v_1,\ldots,v_n),(a_1,\ldots,a_n))\equiv \prod_{i=1}^{n(\secpar)} V(p_i,v_i,a_i)\;.
\end{equation*}
\label{def:repetition_weakly_verifiable_puzzle}
\end{definition}
There is a tight relation between the hardness of a 1-of-2 puzzle and the hardness of a weakly verifiable puzzle. Given a 1-of-2 puzzle $\mathcal Z=(G,O,S,V)$, we define the weakly verifiable puzzle 
\begin{equation*}
	\mathcal{\hat Z}=(G,V_2)
\end{equation*}
(where $V_2$ is defined in Eq.~\eqref{eq:V_2_definition}).
\begin{fact}
For every polynomially bounded function $n:\NN \mapsto \NN$, the 1-of-2 puzzle $\mathcal Z^n$ is $1-h$-hard if and only if the weakly verifiable puzzle $\mathcal{\hat Z}^n$ is $1-h$-hard.
\label{fct:equivalence_of_hardness}
\end{fact}
This fact follows from the observation that the hardness property of the 1-of-2 puzzle $\mathcal Z$ is equivalent to the hardness of the weakly verifiable puzzle $\hat{\mathcal Z}$ (see Definitions~\ref{def:puzzle} and~\ref{def:1-of-2_puzzle}). Furthermore, the hardness of $\mathcal Z^n$ is equivalent to the hardness of $\hat{\mathcal Z}^n$ (see Definitions~\ref{def:repetition} and~\ref{def:repetition_weakly_verifiable_puzzle}).

Canetti, Halevi and Steiner proved a parallel repetition theorem for weakly verifiable puzzles. 
\begin{theorem}[{\cite[Theorem 1]{CHS05}}]
Let $h:\NN \mapsto [0,1]$ be an efficiently computable function, let $n:\NN \mapsto \NN$ be efficiently computable and polynomially bounded, and let $\hat{\mathcal Z}=(G,V)$ be a weakly verifiable puzzle system. If $\hat{\mathcal Z}$ is $1-h$-hard, then $\hat{\mathcal Z}^n$, the n-fold repetition of $\hat{\mathcal Z}$, is $1-h^n$-hard.
\label{thm:parallel_repetition_weakly_verifiable_puzzles}
\end{theorem}

Although the original proof of Canetti, Halevi and Steiner assumed that the hardness is with respect to a classical solver, the result holds also when we consider our definition, in which the solvers are quantum. The reason is as follows. Their proof maps an efficient solver of the n-fold repetition of a puzzle, which succeeds with probability which is non-negligibly higher than $1-h^n$, to an efficient solver that succeeds with probability non-negligibly higher than $1-h$ for a single puzzle, which is of course a contradiction. This reduction is black-box, and in particular there is no rewinding (which, of course, could cause an issue in the quantum setting)\footnote{The weakly verifiable puzzle $\hat{Z}$ constructed from a 1-of-2 puzzle is classical --- each attempt to solve the puzzle creates a new quantum state, so even when running a solver for a puzzle multiple times there is indeed no rewinding.}. For completeness, a sketch of their proof is provided in \cref{sec:parallel_repetition_weakly_verifiable_puzzles}.

We use Theorem~\ref{thm:parallel_repetition_weakly_verifiable_puzzles} to prove the soundness of the 1-of-2 puzzle $\mathcal Z^n$. We assume $\mathcal Z=(G,O,S,V)$ is $1-h$-hard. We define the weakly verifiable puzzle $\hat{ \mathcal Z}=(G,V_2)$. By the equivalence in Fact~\ref{fct:equivalence_of_hardness}, we know that $\hat{ \mathcal Z}$ is also $1-h$-hard. By Theorem~\ref{thm:parallel_repetition_weakly_verifiable_puzzles}, we know that $\hat{ \mathcal Z}^n$ is $1-h^n$-hard. Using the equivalence in Fact~\ref{fct:equivalence_of_hardness} again, we conclude that $\mathcal Z^n$ is $1-h^n$-hard, which completes the proof.      
\end{proof}

\begin{corollary}
A weak 1-of-2 puzzle implies a strong 1-of-2 puzzle.
\label{cor:weak_implies_strong_puzzle}
\end{corollary}

Note that we define a weak 1-of-2 puzzle to have completeness $\eta=1$. We refrain from answering the question whether any puzzle in which $\eta(\secpar)-h(\secpar)$ is noticeable, implies a strong puzzle.

\begin{proof}
By using Theorem~\ref{thm:parallel_repetition_1_of_2_puzzle} with $n(\secpar)=\frac{\log^2(\secpar)}{\log(\frac{1}{h})}$ repetitions\footnote{Note that $n(\secpar)$ is indeed polynomial in $\secpar$ - since a weak 1-of-2 puzzle holds that $1-h$ is noticeable (see \cref{def:1-of-2_puzzle}), by using the inequality $\ln(1-\varepsilon) \leq -\varepsilon$ we get that $\log(1/h)$ is noticeable.} of the weak $h$-hard 1-of-2 puzzle, we construct a $1$-complete \footnote{Recall that a weak 1-of-2 puzzle has completeness $\eta=1$ (see \cref{def:1-of-2_puzzle}).}, $1-h^{n}=1-\frac{1}{\secpar^{\log(\secpar)}}=1-\negl$-hard 1-of-2 puzzle. Note that a $1-\negl$-hard 1-of-2 puzzle is equivalent to a $1$-hard 1-of-2 puzzle, which completes the proof.
\end{proof}

% section parallel_repetition_theorem_for_1_of_2_puzzles (end)

\section{Strong 1-of-2 Puzzles Imply Semi-Quantum Money}
\label{sec:strong_1-of-2_puzzles_imply_money}
In this section, we show a construction of a private semi-quantum money scheme using strong 1-of-2 puzzles.

In \cref{sec:private_money_definitions}, we define \emph{interactive} private quantum money. We define three degrees of security. Full scheme security means that every QPT counterfeiter cannot pass $t+1$ verifications given $t$ quantum money states. We define mini-scheme security as a weaker variant of full security, which is secure only when the adversary is given a single banknote. Finally, we define 2-of-2 mini-scheme security as an even weaker variant wherein the adversary does not have a banknote verification oracle. We also formally define semi-quantum money.

In \cref{sec:construction_of_mini_scheme}, we show the construction of a 2-of-2 mini-scheme, and show that our 2-of-2 mini-scheme is in fact a mini scheme (see \cref{def:QM_mini_scheme}).

In \cref{sec:mini_scheme_implies_full_scheme}, we show that any (interactive private quantum money) mini scheme can be elevated to a full (interactive private quantum money) scheme --- see \cref{def:security_private_interactive_money}.

\subsection{Definitions of Private Semi-Quantum Money} % (fold)
\label{sec:private_money_definitions}

The following definitions are slight variations of the definitions in \cref{sec:public_money_definitions}.

\begin{definition}[Interactive memoryless private quantum money]\label{def:interactive_private_quantum_money} An interactive memoryless private quantum money scheme consists of a classical PPT key generation algorithm $\keygen$ and two-party interactive memoryless QPT protocols $\mint$ and $\verify$. $\keygen(\secparam)$ outputs a key $k$. Both the minting protocol and the verification protocol are two-party quantum protocols involving an Acquirer $\acquirer$ and a Bank $\bank$. During both protocols, the bank receives the key $k$ as input, and the user does not.
At the end of the honest run of $\mint$, the user holds a quantum money state that, in general, could be a mixed state. In this work, the protocols will end with a pure state, usually denoted $\ket{\$}$. In the following sections, for the sake of clarity, we work with the pure-state formalism. The banknote the user chooses to verify is denoted in this work as the input of the $\verify$ protocol. At the end of the verification protocol, the bank outputs a bit $b$ that states whether the money is valid or not.
\paragraph{Correctness.} The scheme is \emph{correct} if there exists a negligible function $\negl$ such that:
\begin{align*}
    \Pr(k &\sample \keygen(\secparam); \ket{\$}\sample \mint_{k}(\secparam); b \sample\verify_{k}(\ket{\$}):\\
 &b=1) = 1 - \negl\;.
\end{align*}
\end{definition}

\begin{definition} We say that the protocol has classical minting (verification) if $\bank$ is classical in $\mint$ $(\verify)$. To emphasize that the verification is classical, we use $\cverify$ to denote the (classical) verification algorithm. We define private semi-quantum money as any secure memoryless interactive private quantum money protocol that has classical minting \emph{and} classical verification.  
\label{def:semi_classical}
\end{definition}

In the quantum setting, there are a number of possible verifications with different qualities; a notable quality is whether the verification "destroys" the banknote (i.e., whether the banknote can be used again after verification). This distinction can be thought of as the difference between verifying --- proving that a legal money state exists --- and spending --- proving a legal money state does not exist --- and it becomes more interesting when considering the public setting. There, a banknote can be spent with the bank in the same manner as in the private setting, but it can also be verified with other users. In such a case it is important that the banknote is preserved, so it could be transferred. Another distinction is added by the introduction of classically verified money: whether the verification is a classical or quantum protocol. Moreover, a classical verification must be a challenge-response protocol --- otherwise the same proof can be passed twice, effectively spending the same banknote twice. In our scheme, verification is classical and does not preserve the banknote, proving both that it existed and that it does not exist anymore.

In this definition, we emphasize that the protocols $\mint$ and $\verify$ are \emph{memoryless}: i.e., all outgoing messages depend solely on the key and the input from the user. In other words, the bank does not maintain a variable state that changes between different runs of the protocols --- each run is independent. Constructing a stateful scheme is trivial even in the classical setting, as discussed in \ifnum\sigconf=0 \cref{sec:advantage_of_statelessness}\else the full version~\cite{RS19}\fi. In addition, it is interesting to note that our protocols are composed of a fixed number of messages, independent of the security parameter: $\verify$ has 2 messages (a single round) and $\mint$ has 3 messages.

\begin{definition}\label{def:security_private_interactive_money}
We say that an interactive private quantum money scheme $\$$ is secure if for every QPT counterfeiter \nom{$\adv$}{Counterfeiter of a full blown private quantum money scheme (also NTCF adversary in \cref{sec:ntcf_and_1-of-2_puzzles})}{A} there exists a negligible function $\negl$ such that:
\begin{equation*}
    \Pr[\counterfeit{\fulladv}{\$}{full} = 1] \leq \negl\;.
\end{equation*}

The money counterfeiting game $\counterfeit{\adv}{\$}{full}$:
\begin{enumerate}
  \item The bank generates a key $k\sample \keygen(\secparam)$.
  \item The bank and the counterfeiter interact. The counterfeiter can ask the bank to run $\mint_{k}(\cdot)$ and $\verify_{k}(\cdot)$ polynomially many times, in any order the counterfeiter wishes. The counterfeiter is not bound to following his side of the protocols honestly. The counterfeiter can keep ancillary registers from earlier runs of these protocols and use them in later steps. Let $w$ be the number of successful verifications, $\ell$ the number of times that mint was called by the counterfeiter and $v$ the number of times that verify was called by the counterfeiter.
  \item The bank outputs $(w,\ell,v)$.
\end{enumerate}

The value of the game is 1 iff $w > \ell$. In this case we sometimes simply say that the counterfeiter wins.
\end{definition}

We state here again that this security definition is stronger than the one previously used --- this is further explained in \cref{rmk:stronger_security_definition}.

Following previous works~\cite{AC13,GK15}, we define a private quantum money mini-scheme, with a slight deviation. Additionally, we define a 2-of-2 mini-scheme, which is a weaker variant of the mini-scheme.

\begin{definition}[quantum money mini-scheme and 2-of-2 mini-scheme]
\label{def:QM_mini_scheme}
\label{def:QM_2-of-2_mini_scheme}
    We define mini-scheme security as we defined full scheme security but with regard to $\counterfeit{\miniadv}{\$}{mini}$, wherein the counterfeiter \nom{$\miniadv$}{Counterfeiter of a private quantum money mini-scheme}{B} wins iff $w > \ell \wedge \ell = 1$.

    We define 2-of-2 mini-scheme security as we did above but with regard to $\counterfeit{\twominiadv}{\$}{2-of-2}$, where the counterfeiter \nom{$\twominiadv$}{Counterfeiter of a private quantum money 2-of-2 mini-scheme (also quantum lightning certificate forger)}{C} wins iff $w > \ell \wedge \ell = 1 \wedge v = 2$.
\end{definition}

Note that the definitions in this sections could be naturally extended to the public settings.
% section definitions (end)

\subsection{Construction of a Mini-Scheme}
\label{sec:construction_of_mini_scheme}
In this section, we show the construction of a scheme that we then prove to be a 2-of-2 semi-quantum mini-scheme. Later we prove that our construction in fact achieves a stronger security notion --- a semi-quantum mini-scheme.

We now give an informal description of our construction, which is defined formally in \cref{alg:mini_scheme_construction}. The construction uses
a strong 1-of-2 puzzle \ifnum\sigconf=1 (see \cref{def:1-of-2_puzzle}) \fi and a post-quantum existentially unforgeable under an adaptive chosen-message attack (PQ-EU-CMA) MAC \ifnum\sigconf=0 (see Definitions~\ref{def:1-of-2_puzzle},~\ref{def:MAC} and~\ref{def:unforgeability_of_MAC})\fi. In $\keygen$, the bank generates a MAC signing key and $n$ pairs of strong 1-of-2 puzzles and their respective verification keys. The minting process is done as follows. The bank sends these $n$ puzzles to the user, who then runs the obligation protocol $\mathcal Z.O$ on all the $n$ puzzles. The user keeps the quantum output of $O$ and sends the classical outputs (called the obligations) to the bank. The bank signs these obligations using the classical MAC scheme and sends these tags back to the user. The verification starts with the bank sending random challenges to the user. The user then has to present a set of signed obligations (which the user should have from the $\mint$ protocol) together with a set of solutions to the challenges of these puzzles. The bank verifies the solution to each puzzle with its respective verification key (the set of verification keys is part of the key). Due to the fact that this verification is classical, it is denoted $\cverify$. We show that a counterfeiter cannot double-spend a banknote without breaking the soundness of a strong 1-of-2 puzzle (or the security of the MAC).

Intuitively, an adversary could try to double-spend the banknote using the solutions he received from the first verification, while hoping to be given the same challenges. However, assuming a sufficiently large number of puzzles (say, $n=\log^2(\secpar)$), the probability of encountering the exact same set of challenges more than once is negligible. Passing two verifications of any banknote in which the challenges were not the same both times essentially requires one to pass the $\mathsf{SOLVE-2}$ security game for the 1-of-2 puzzle. Insofar as this is considered a strong 1-of-2 puzzle, the probability that it can occur is therefore negligible.

\begin{algorithm*}
\caption{The Interactive Private Money Scheme \nom{$\$_\mathcal{Z}$}{Our private semi-quantum money construction based on 1-of-2 puzzles}{DOLLAR_Z}}
\label{alg:mini_scheme_construction}

    \fbox{
    \procedure[linenumbering]{$\$_\mathcal{Z}.\keygen(\secparam)$}{
        n \gets \log^2(\secpar)\\
        \pcforeach i \in [n]:\\
        \;\;\;\;(p_i, v_i) \gets \mathcal{Z}.G(\secparam)\\
        \overline{p}^n \gets (p_1, \dots, p_n), \overline{v}^n \gets (v_1, \dots, v_n)\\
        k \gets MAC.\keygen(\secparam)\\
        k_\$ \gets (\overline{p}^n, \overline{v}^n, k)\\
        \textbf{return} \; k_\$
    }
    }
    
\fbox{
\pseudocode[head = $\$_\mathcal{Z}.\mint_{k_\$}$]{
    \textbf{Acquirer} \<\< \textbf{Bank} \\[] [\hline]
    \pcln \< \sendmessageleft*{\overline{p}^n} \< \\
    \pcln \pcforeach i \in [n]:\\
    \pcln \;\;\;\;(o_i, \psi_i) \gets \mathcal{Z}.O(p_i)\\
    \pcln \overline{o}^n \gets (o_1, \dots, o_n), \overline{\psi}^n \gets (\psi_1, \dots, \psi_n)\\
    \pcln \label{{dollar_z:mint:obligation}}\< \sendmessageright*{\overline{o}^n} \< \\
    \pcln \<\< t_o \gets MAC.\mac_{k}(\overline{o}^n) \\
    \pcln \< \sendmessageleft*{t_o} \<
}
}

\fbox{
\pseudocode[head = $\$_\mathcal{Z}.\cverify_{k_\$}({\overline{o}^n, t_o, \overline{\psi}^n})$]{
    \textbf{Acquirer} \<\< \textbf{Bank} \\[] [\hline]
    \pcln \label{{dollar_z:verify:b}}\< \sendmessageleft*{\overline{b}^n \text{\nom{$\in_R$}{For a finite set $S$ we denote $s \in_R S$ to be the process in which $s$ is sampled uniformly from $S$}{IN_R}} \{0,1\}^n} \< \\
    \pcln \pcforeach i \in [n]:\\
    \pcln \;\;\;\;a_i \gets \mathcal{Z}.S(p_i, o_i, \ket{\psi_i}, b_i) \\
    \pcln \label{{dollar_z:verify:answer}}\< \sendmessageright*{\overline{a}^n, \overline{o}^n, t_o} \< \\
    \pcln \label{{dollar_z:verify:mac}}\<\< r_{MAC} \gets MAC.\verify_{k}(\overline{o}^n, t_o) \\
    \pcln \<\< \pcforeach i \in [n]:\\
    \pcln \label{{dollar_z:verify:z_v}} \<\< \;\;\;\; r_i \gets \mathcal{Z}.V(p_i, v_i, o_i, b_i, a_i) \\
    \pcln \label{{dollar_z:verify:result}} \<\< r \gets r_{MAC} \cdot \prod_{i=1}^n r_i \\
    \pcln \label{{dollar_z:verify:return}} \<\< \pcreturn r
}
}
\end{algorithm*}

For ease of notation, we write:

\begin{itemize}
    \item $\overline{p}^n \coloneqq (p_1, \dots, p_n)$
    \item $\overline{v}^n \coloneqq (v_1, \dots, v_n)$
    \item $\overline{o}^n \coloneqq (o_1, \dots, o_n)$
    \item $\overline{\psi}^n \coloneqq \ket{\psi_1} \otimes \dots \otimes \ket{\psi_n}$
    \item $\overline{b}^n \coloneqq (b_1, \dots, b_n)$
    \item $\overline{a}^n \coloneqq (a_1, \dots, a_n)$
\end{itemize}

\begin{proposition}[Correctness of $\$_\mathcal{Z}$]
    Assuming MAC has perfect completeness and $\mathcal{Z}$ is a 1-of-2 puzzle with completeness $\eta=1$, $\$_\mathcal{Z}$ (\cref{alg:mini_scheme_construction}) is a semi-quantum money scheme that satisfies the correctness property (see \cref{def:interactive_private_quantum_money}).
    \label{prop:mini_scheme_construction_correctness}
\end{proposition}

\begin{proof}
    Clearly, the communication and the bank's operation in $\mint$ and $\cverify$ are classical --- therefore, the scheme is semi-quantum. 
    
    From the \textit{perfect completeness} property of MAC \ifnum\sigconf=0(see \cref{def:MAC}) \else (see the full version~\cite{RS19})\fi we get:
    \begin{equation*}
    \Pr[MAC.\verify_{k}(\overline{o}^n, MAC.\mac_{k}(\overline{o}^n))=1]=1
    \end{equation*}
    meaning $\Pr[r_{MAC}=1]=1$.

    From the completeness $\eta=1$ of $\mathcal{Z}$ we get:
    \begin{align*}
        \Pr[(p, v) &\sample \mathcal{Z}.G(\secpar); (o, \ket{\psi}) \sample \mathcal{Z}.O(p); b \sample \{0, 1\};\\
        &a \sample \mathcal{Z}.S(p, o, \ket{\psi}, b):\\
        &\mathcal{Z}.v(p, v, o, b, a)=1]\\
        &\geq 1 - \negl\;.
    \end{align*}
    Let $b_i$ be the event of failing verification on the $i^{th}$ puzzle. From the previous equation, $\Pr[b_i] \leq \negl$ for some negligible function $\negl$. Let $\negl[]'(\secpar) \coloneqq n \cdot \negl = \log^2(\secpar) \cdot \negl$. Using the union bound:
    
    \begin{align*}
        \Pr[\cup_{i=1}^n b_i] \leq \sum_{i=1}^n \Pr[b_i] = \log^2(\secpar) \cdot \negl = \negl[]'(\secpar)
    \end{align*}

    meaning $\Pr[\left(\prod_{i=1}^n r_i\right)=1] \geq 1 - \negl[]'(\secpar)$. Thus:
    \begin{align*}
        \Pr[k_\$ &\sample \$_{\mathcal{Z}}.\keygen(\secparam); (\overline{p}^n, \overline{o}^n, t_o, \overline{\psi}^n) \sample \$_{\mathcal{Z}}.\mint_{k_\$}();\\
        &\$_{\mathcal{Z}}.\cverify_{k_\$}(\overline{p}^n, \overline{o}^n, t_o, \overline{\psi}^n) = 1] \\
        = &\Pr[r_{MAC} = 1 \bigcap \left(\prod_{i=1}^n r_i\right) = 1] \\
        \geq &1-\negl[]'(\secpar)\;.
    \end{align*}
\end{proof}

\begin{proposition}[$\$_\mathcal{Z}$ is a 2-of-2 mini-scheme]
 Assuming $\mathcal{Z}$ is a strong 1-of-2 puzzle and MAC is a PQ-EU-CMA MAC, the scheme $\$_\mathcal{Z}$ (Algorithm \ref{alg:mini_scheme_construction}) is a 2-of-2 mini-scheme (see \cref{def:QM_2-of-2_mini_scheme}).
 \label{prop:dollar_z_is_2_of_2}
\end{proposition}

\begin{proof}

    We show that the probability of a QPT counterfeiter to win the 2-of-2 mini-scheme security game against $\$_\mathcal{Z}$ (\cref{alg:mini_scheme_construction}) is bound by the negligible probability to solve both challenges of the strong 1-of-2 puzzle $\mathcal{Z}$. Intuitively, double-spending a banknote entails solving both challenges for at least one of its $n$ puzzles, which is intractable. For this proof, as well as the following security proofs of our money scheme (\cref{prop:mini_scheme_construction} and \cref{thm:full_scheme_construction_security}), we use a sequence-of-games based technique adapted from \cite{Sho04}. The following sequence of games binds the success probability of any QPT 2-of-2 mini-scheme counterfeiter to that of a QPT 2-of-2 puzzle solver (see \cref{eq:succ_2_upperbound}):

   \paragraph{Game 0.} Let $\twominiadv$ be a QPT 2-of-2 mini-scheme counterfeiter. We assume w.l.o.g. that $\twominiadv$ performs exactly two verifications and one mint (i.e., $\ell=1$ and $v=2$) --- an adversary which does not comply with this assumption will necessarily fail (see \cref{def:QM_mini_scheme}). We define Game 0 to be $\counterfeit{\twominiadv}{\$_\mathcal{Z}}{2-of-2}$.

    Let $S_0$ be the event where $w > 1$ (see \cref{def:QM_2-of-2_mini_scheme}) in Game 0 (this is the original win condition for $\twominiadv$, since we assume $\ell=1 \wedge v=2$).

    \paragraph{Game 1.} We now transform Game 0 into Game 1, simply by changing the win condition: Game 1 is identical to Game 0, but we define the following event: let $\overline{b^1}^n, \overline{b^2}^n$ be the random bit strings that were generated in line \ifnum\sigconf=0 \ref{{dollar_z:verify:b}} \else \ref{dollar_z:verify:b} \fi of $\$_\mathcal{Z}.\cverify$ the first and second times $\twominiadv$ asked for verification, respectively. Let $S_1$ be the event where $w > 1 \wedge \overline{b^1}^n \neq \overline{b^2}^n$ in Game 1.

    Let $F$ be the event where $\overline{b^1}^n = \overline{b^2}^n$ in Game 1, and $F'$ the event where $w > 1 \wedge \overline{b^1} = \overline{b^2}$ in Game 1. Since $\overline{b^1}^n$ and $\overline{b^2}^n$ are generated uniformly and independently, $\Pr[F]=\frac{1}{2^n} \leq \negl$ for some negligible function $\negl$. Therefore: $\Pr[S_0] = \Pr[S_1 \cup F'] \leq \Pr[S_1 \cup F] \leq \Pr[S_1] + \Pr[F] \leq \Pr[S_1] + \negl$. So $\Pr[S_0] \leq \Pr[S_1] + \negl$.
    %Let $F$ be the event where $b^1 = b^2$. It is easy to see that $S_0 \wedge \neg F \iff S_1 \wedge \neg F$. Therefore, from the Difference Lemma (\cref{lem:difference}), we get $\abs{\Pr[S_1] - \Pr[S_0]} \leq \Pr[F]$. Since $b^1$ and $b^2$ are generated uniformly and independently, $\Pr[F]=\frac{1}{2^n} \leq \negl$ for some negligible function $\negl$, meaning $\Pr[S_1]$ is negligibly close to $\Pr[S_0]$.

    \paragraph{Game 2.} We now add a small change to the game above: at the start of the game, a uniform $i' \in_R [n]$ is chosen by the bank. Let $j$ be the first index such that $b^1_j \neq b^2_j$ ($j=\infty$ if $b^1 = b^2$).

    Let $S_2$ be the event where $w > 1 \wedge b^1 \neq b^2 \wedge i' = j$ in Game 2.

    $S_1 \Rightarrow b^1 \neq b^2$, so since $i'$ was chosen uniformly and independently of $w$, $b^1$, $b^2$ and $j$, we get that $\Pr[S_2|S_1] = \frac{1}{n}$. Moreover, it is easy to see that $\Pr[S_2|\neg S_1] = 0$. So $\Pr[S_2] = \frac{1}{n} \cdot \Pr[S_1]$, meaning $\Pr[S_1]$ is a polynomial multiplicative factor of $\Pr[S_2]$.

    \paragraph{Game 3.} Game 3 is identical to Game 2, but we now add an additional constraint to the win condition. Let $\overline{o}^n$ be the set of obligations $\twominiadv$ sent in line \ifnum\sigconf=0 \ref{{dollar_z:mint:obligation}} \else \ref{dollar_z:mint:obligation} \fi of $\$_\mathcal{Z}.\mint$, and let $\overline{o^1}^n$, $\overline{o^2}^n$ be the sets of obligations sent by $\twominiadv$ during line \ifnum\sigconf=0 \ref{{dollar_z:verify:answer}} \else \ref{dollar_z:verify:answer} \fi of $\$_\mathcal{Z}.\cverify$ the first and second times $\twominiadv$ asks for verification, respectively.

    Let $S_3$ be the event where $w > 1 \wedge b^1 \neq b^2 \wedge i' = j \wedge \overline{o^1}^n = \overline{o^2}^n = \overline{o}^n$ in Game 3.

    Let $F$ be the event where $\twominiadv$ passes one or more verifications such that $\overline{o^1}^n \neq \overline{o}^n$ or $\overline{o^2}^n \neq \overline{o}^n$. It is easy to see that $S_2 \wedge \neg F \iff S_3 \wedge \neg F$. Therefore, from the Difference Lemma \ifnum\sigconf=0 (\cref{lem:difference}) \else (\cite{Sho04}, see also the full version~\cite{RS19}) \fi we get
    \begin{equation*}
        \abs{\Pr[S_3] - \Pr[S_2]} \leq \Pr[F]\;.
    \end{equation*}
    From the unforgeability of MAC \ifnum\sigconf=0 (see \cref{def:unforgeability_of_MAC})\fi, $\Pr[F]$ is negligible\footnote{Otherwise, we could construct a MAC forger $\forger$ with non-negligible success probability. Assume towards a contradiction that with non-negligible probability, $\twominiadv$ passes verification by sending in line \ifnum\sigconf=0 \ref{{dollar_z:verify:answer}} \else \ref{dollar_z:verify:answer} \fi $\overline{o'}^n, t_o'$ such that $\overline{o'}^n \neq \overline{o}^n$. That means that the MAC verification in line \ifnum\sigconf=0 \ref{{dollar_z:verify:mac}} \else \ref{dollar_z:verify:mac} \fi passed. So $\forger$ could simulate a bank, but instead of signing and verifying with $k$ generated in $\$_\mathcal{Z}.\keygen$, $\forger$ uses the signing and verification oracles. $\forger$ runs $\twominiadv$ against the simulated bank, and present $\overline{o'}^n, \tilde{o'}^n$. With non-negligible probability, MAC verification passes, and since $\overline{o'}^n \neq \overline{o}^n$, and no other signings are requested ($\mint$ was run only once), $\forger$ wins $\macforge{MAC}$.}. Therefore, $\Pr[S_2] \leq \Pr[S_3] + \negl$.

    \paragraph{Game 4.} We now change the behavior of verifications. Let $\overline{a^1}^n$, $\overline{a^2}^n$ be the sets of answers sent by $\twominiadv$ in line \ifnum\sigconf=0 \ref{{dollar_z:verify:answer}} \else \ref{dollar_z:verify:answer} \fi of $\$_\mathcal{Z}.\cverify$ the first and second times $\twominiadv$ asks for verification, respectively\footnote{$\twominiadv$ can, of course, run both verification protocols simultaneously. We number the verifications according to the one that got to line \ifnum\sigconf=0 \ref{{dollar_z:verify:answer}} \else \ref{dollar_z:verify:answer} \fi of the protocol first.}. Instead of performing verifications both times, the bank now performs both verifications only on the second time: the first time $\$_\mathcal{Z}.\cverify$ is called, after line \ifnum\sigconf=0 \ref{{dollar_z:verify:answer}} \else \ref{dollar_z:verify:answer} \fi the bank returns 1 and stops. The second time $\$_\mathcal{Z}.\cverify$ is called, the bank performs both verifications: i.e., on the second verification we replace everything from line \ifnum\sigconf=0 \ref{{dollar_z:verify:z_v}} \else \ref{dollar_z:verify:z_v} \fi with:

    \pseudocode[lnstart=5, linenumbering]{
         \pcforeach i \in [n]: \\
         \;\;\;\; r_i \gets \mathcal{Z}.V(p_i, v_i, o_i, \fbox{$b^1_i, a^1_i$})\\
         \;\;\;\; \fbox{$r'_i \gets \mathcal{Z}.V(p_i, v_i, o_i, b^2_i, a^2_i$)}\\
         r \gets r_{MAC} \cdot \prod_{i=1}^n r_i \fbox{$\cdot r'_i$}\\
         \pcreturn r
    }
    
    Let $S_4$ be the event where $w > 1 \wedge b^1 \neq b^2 \wedge i' = j \wedge \overline{o^1}^n = \overline{o^2}^n = \overline{o}^n$ in Game 4.

    Verifying both inputs on the second request is equivalent to verifying them individually: $S_3 \Rightarrow S_4$ since if both verifications pass in Game 3, then both pass in Game 4 (the first one always passes, the second one runs both verifications that passed in $S_3$), and $\neg S_3 \Rightarrow \neg S_4$ since that means one of the verifications in Game 3 fail, which means the second verification in Game 4 fails. So $\Pr[S_3]=\Pr[S_4]$.

    \paragraph{Game 5.} We now change the second verification: on the $i'^{th}$ pair of puzzles, if $b^1_{i'} \neq b^2_{i'}$ (we note that this always holds when $i'=j$), we perform $V_2$ instead of normal verification --- i.e., we replace everything from line \ifnum\sigconf=0 \ref{{dollar_z:verify:z_v}} \else \ref{dollar_z:verify:z_v} \fi forward in $\$_\mathcal{Z}.\cverify$ in the second verification with:

    \pseudocode[lnstart=5, linenumbering]{
        \pcforeach i \in [n]: \\
        \;\;\;\; \fbox{$\pcif i=i' \wedge b^1_{i'} \neq b^2_{i'}:$} \\
        \;\;\;\;\;\;\;\; \fbox{$\pcif b^1_i=0: \hat{a_0} \gets a^1_i, \hat{a_1} \gets a^2_i$} \\
        \;\;\;\;\;\;\;\; \fbox{$\pcelse: \hat{a_0} \gets a^2_i, \hat{a_1} \gets a^1_i$} \\
        \label{{game_6:verify:v_2}} \;\;\;\;\;\;\;\; \fbox{$r_i, r'_i \gets V_2(p_i, v_i, o_i, \hat{a_0}, \hat{a_1})$} \\
        \;\;\;\; \fbox{$\pcelse:$} \\
        \;\;\;\;\;\;\;\; r_i \gets \mathcal{Z}.V(p_i, v_i, o_i, b^1_i, a^1_i)\\
        \;\;\;\;\;\;\;\; r'_i \gets \mathcal{Z}.V(p_i, v_i, o_i, b^2_i, a^2_i)\\
        \;\;\;\; \fbox{$\pcendif$} \\
        r \gets r_{MAC} \cdot \prod_{i=1}^n r_i \cdot r'_i\\
        \pcreturn r
    }

    Let $S_5$ be the event where $w > 1 \wedge b^1 \neq b^2 \wedge i' = j \wedge \overline{o^1}^n = \overline{o^2}^n = \overline{o}^n \wedge V_2(p_i, v_i, o_i, \hat{a_0}, \hat{a_1})=1$ in Game 5.

    In the case where $i=i' \wedge b^1_i \neq b^2_i$, running $V_2(p_i, v_i, o_i, \hat{a_0}, \hat{a_1})$ is equivalent to running $\mathcal{Z}.v$ twice, since we assign $\hat{a_0}$ and $\hat{a_1}$ respective to $b^1_i$ and $b^2_i$. So $S_4 \iff S_5$, meaning $\Pr[S_4]=\Pr[S_5]$.

    \paragraph{Game 6.} We now simply relax the win condition: Game 6 goes exactly the same as Game 5, but we define the following event: let $S_6$ be the event where $V_2(p_i, v_i, o_i, \hat{a_0}, \hat{a_1})=1$ in Game 6. Since this is a relaxation of the conditions of $S_5$, we get $\Pr[S_5] \leq \Pr[S_6]$.

    \paragraph{Bound on success probability.} We show a reduction mapping a 2-of-2 mini-scheme counterfeiter to a 2-of-2 solver (see \cref{def:hardness_puzzle}):

    Let $\twominiadv$ be a QPT 2-of-2 mini-scheme counterfeiter. We construct a QPT 2-of-2 solver $\solvertwo$ in the following manner:

    Let $(p,v)$ be the output of $G(\secparam)$ at step 1 of the solving game. On step 2, $\solvertwo$ simulates a Game 6 bank (by honestly running mints and verifications as defined in Game 5, as well as choosing $i'$ uniformly) with two changes:
    \begin{enumerate}
        \item The $i'^{th}$ puzzle is replaced with $p$.
        \item In line \ifnum\sigconf=0 \ref{{game_6:verify:v_2}} \else \ref{game_6:verify:v_2} \fi of the second verification, $\solvertwo$ outputs $(o_i, \hat{a_0}, \hat{a_1})$ to the puzzle giver instead of running $V_2$. The honest puzzle giver runs $V_2(p, v, o_i, \hat{a_0}, \hat{a_1})$ and returns the result, which $\solvertwo$ uses as $r_i$ and $r_{i'}$.
    \end{enumerate}
    We can see that for any $\twominiadv$, $\Pr[S_6]$ is not affected by the above changes: in the first change we replace a random puzzle with another random puzzle, which has no affect on $\Pr[S_6]$. In the second change, the honest puzzle giver runs $V_2$ with exactly the same input as the bank in the original Game 6 should, and returns the result --- this also does not affect $\Pr[S_6]$.

    $\solvertwo$ runs $\twominiadv$ against Game 6. $S_6$ is exactly the win condition of the 2-of-2 solving game, which means $\solvertwo$ wins the 2-of-2 solving game with probability $\Pr[S_6]$. Since $\mathcal{Z}$ is a strong 1-of-2 puzzle, the success probability of any QPT 2-of-2 solver is negligible --- meaning $\Pr[S_6]$ is negligible for any QPT counterfeiter.

    For each pair of consecutive games $i$ and $i+1$, we have shown that $\Pr[S_i] \leq \poly \cdot \Pr[S_{i+1}] + \negl$ for some $\poly, \negl$. Finally, we have shown that $\Pr[S_6]$ is negligible in $\secpar$, so we can conclude that $\Pr[S_0]$ is negligible in $\secpar$. Since Game 0 is defined as the 2-of-2 mini-scheme counterfeiting game, and $S_0$ is defined as its win condition, no QPT 2-of-2 mini-scheme counterfeiter can win the game with more than negligible probability.
\end{proof}

We now prove that $\$_\mathcal{Z}$ (\cref{alg:mini_scheme_construction}) is, in fact, a mini-scheme (see \cref{def:QM_mini_scheme}). Unlike the others, this proof is not modular --- not every 2-of-2 mini-scheme is a mini-scheme. For example, consider a scheme wherein the bank shares with the counterfeiter a single bit of the key on each verification. This scheme could have 2-of-2 mini-scheme security, but obviously, it would not be secure for a counterfeiter with a verification oracle, which could easily discern the key.

\begin{proposition}[$\$_\mathcal{Z}$ is a mini-scheme]
\label{prop:mini_scheme_construction}
    Assuming $\$_\mathcal{Z}$ is a 2-of-2 mini-scheme (where $\$_\mathcal{Z}$ is given in \cref{alg:mini_scheme_construction}, and a 2-of-2 mini-scheme is defined in \cref{def:QM_2-of-2_mini_scheme}), $\$_\mathcal{Z}$ is a mini-scheme (see \cref{def:QM_mini_scheme}).
\end{proposition}

\begin{proof}
    We use an idea very similar to that used in Ref.~\cite[Theorm 5]{PYJ+12} (a slightly different variation also appeared in Ref.~\cite[Appendix C]{BS16}); we show that if a counterfeiter with access to a verification oracle can ask for $v$ verifications and have two of them succeed, a 2-of-2 counterfeiter could guess the two success indices randomly and apply the same strategy, thus breaking the security of the 2-of-2 mini-scheme. The following sequence of games binds the success probability of any QPT mini-scheme counterfeiter to that of a QPT 2-of-2 mini-scheme counterfeiter against $\$_\mathcal{Z}$:

    \paragraph{Game 0.} Let $\miniadv$ be a QPT mini-scheme counterfeiter. We assume w.l.o.g. that $\miniadv$ asks for mint only once (i.e., $\ell=1$), and for verification $v$ times such that $v$ is polynomial in $\secpar$ --- an adversary which does not comply with this assumption necessarily fails (see \cref{def:QM_mini_scheme}). We define the first game to be $\counterfeit{\miniadv}{\$_\mathcal{Z}}{mini}$ (see \cref{def:QM_mini_scheme}).

    Let $S_0$ be the event where $w > 1$ (see \cref{def:QM_mini_scheme}) in Game 0 (this is the original win condition for $\miniadv$ since we assume $\ell=1$ and $v$ is polynomial in $\secpar$).

    \paragraph{Game 1.} We now make one small change to Game 0, namely, that the game stops after $\miniadv$ receives two successful verifications (i.e., the counterfeiter is not allowed to make additional verifications after receiving two successful ones. We model this by defining additional verification attempts as failures).

    Let $S_1$ be the event where $w=2$ in Game 1.

    It is obvious why $S_1 \Rightarrow S_0$. In addition, $S_0 \Rightarrow S_1$, since any run of Game 0 with more than two successful verifications is equivalent to a run of Game 1 in which all verifications beyond the second successful one are ignored. So $\Pr[S_0]=\Pr[S_1]$.

    \paragraph{Game 2.} We model a run of $v$ verifications using a string $r \in \{0,1\}^v$, such that $r_i=1$ if and only if the $i^{th}$ time $\miniadv$ asked for verification was successful\footnote{$\miniadv$ can, of course, run several verification protocols simultaneously. We number the verifications according to the order in which they were initiated.}.
    At the beginning of Game 2, a uniform binary string $r' \in_R \{0,1\}^v$ is generated such that $\sum_{i=1}^v r'_i = 2$.

    Let $S_2$ be the event where $w=2 \wedge r'=r$ in Game 2.

    Given $S_1$, we know that the string $r$ representing the verifications in Game 1, like $r'$, also holds $\sum_{i=1}^v r_i = 2$. There are $\binom{v}{2}$ such strings, so since $r'$ was chosen uniformly and independently of $r$, there is a $\frac{1}{\binom{v}{2}}$ probability that $r'=r$. So $\Pr[S_2]=\frac{1}{\binom{v}{2}} \cdot \Pr[S_1]$, meaning $\Pr[S_1] = \binom{v}{2} \cdot \Pr[S_2]$.

    \paragraph{Game 3.} We transform Game 2 into Game 3 by changing the following: for each $i \in [v]$, for the $i^{th}$ time $\miniadv$ runs a verification protocol with the bank, instead of receiving the actual result of the MAC and puzzle verifications ($r$), it receives $r'_i$; i.e., we change line \ifnum\sigconf=0 \ref{{dollar_z:verify:return}} \else \ref{dollar_z:verify:return} \fi with $\pcreturn r'_i$.

    Let $S_3$ be the event where $w=2 \wedge r'=r$ in Game 3.

    Given $S_2$, since $r'=r$ in both Game 2 and Game 3, the fact that $\miniadv$ receives $r'_i$ instead of $r_i$ changes nothing. So $\Pr[S_3|S_2]=1$. Trivially, $\Pr[S_3|\neg S_2]=0$. So $\Pr[S_2]=\Pr[S_3]$.

    \paragraph{Game 4.} Let $k, h$ be the two indices such that $r'_k = r'_h = 1, k \neq h$ (by construction there are exactly two such indices). In Game 4, for every verification other than the $k^{th}$ and the $h^{th}$, the MAC verification and puzzle verifications are not called at all --- $b_i$ is generated and $r'_i$ is returned; i.e., lines \ifnum\sigconf=0 \ref{{dollar_z:verify:mac}} \else \ref{dollar_z:verify:mac} \fi to \ifnum\sigconf=0 \ref{{dollar_z:verify:result}} \else \fi are removed.

    Let $S_4$ be the event where $w=2 \wedge r'=r$ in Game 4.

    It is easy to see that $\Pr[S_3]=\Pr[S_4]$, since for every verification but the $k^{th}$ and the $h^{th}$, the bank did nothing with the result of the MAC or puzzle verifications, so whether we run them at all changes nothing.

    \paragraph{Bound on success probability.} We show a reduction mapping a mini-scheme counterfeiter to a 2-of-2 mini-scheme counterfeiter (see \cref{def:QM_2-of-2_mini_scheme}):

    Let $\miniadv$ be a QPT mini-scheme counterfeiter. We construct a a QPT 2-of-2 mini-scheme counterfeiter $\twominiadv$ in the following manner:

    $\twominiadv$ simulates a Game 4 bank with the following difference: when asked to run $\$_\mathcal{Z}.\mint$, it, in turn, asks the real bank to run $\$_\mathcal{Z}.\mint$ and returns the result, and on the $k^{th}$ an $h^{th}$ verifications, it asks the real bank to run $\$_\mathcal{Z}.\cverify$ and returns the result. We note that for any other verification, $\twominiadv$ can simulate the bank since MAC and puzzle verifications are not performed; all it needs to do is choose a uniform $b$ and return $r'_i$. $\twominiadv$ runs $\miniadv$ against the simulated Game 4 bank.

    So $\Pr[S_4]=\Pr[\counterfeit{\twominiadv}{\$_\mathcal{Z}}{2-of-2}=1] \leq \negl$ for some negligible function $\negl$.
    Therefore, by construction, we get that $\Pr[S_0] \leq \poly \cdot \Pr[S_4]$ for some $\poly$ and therefore is also negligible for any QPT counterfeiter. Game 0 is defined as the original mini-scheme security game, and $S_0$ is defined as its original win condition; therefore, $\$_\mathcal{Z}$ (Algorithm~\ref{alg:mini_scheme_construction}) is a mini-scheme.
\end{proof}

\subsection{A Mini-Scheme Implies a Full Blown Scheme}
\label{sec:mini_scheme_implies_full_scheme}
We show how a mini-scheme $\$$ can be used to construct a full blown scheme $\hat{\$}$. The construction is based on a very similar idea to those in Refs.~\cite[Appendix C]{BS16} and \cite[Section 3.3]{AC13}.

Here we provide an informal description of our full scheme $\hat{\$}$. The construction is defined formally in \cref{alg:full_scheme_construction}. Our full scheme is constructed by minting mini-scheme banknotes, and including the key of the mini-scheme in each one. To that end, a MAC and a private-key encryption scheme are used: on minting, the bank mints a mini-scheme banknote, encrypts the mini-scheme key that was generated in the process, signs it in its encrypted form, and hands it to the user together with the mini-scheme banknote. The secure nature of the encryption scheme prevents the user from exploiting the mini-scheme key to break the mini-scheme's underlying security. On verification, the bank uses the MAC scheme to verify that the note was indeed minted by a bank, after which it decrypts the mini-scheme key to verify the mini-scheme banknote itself.

In both \cite{BS16} and \cite{AC13}, the core idea of the construction is the same, with minor differences: in \cite{BS16} algorithms are used instead of interactive protocols, and \cite{AC13} is in the public setting, so a digital signature scheme is used instead of MAC, and an encryption scheme is not necessary.

We prove the security of the full-blown scheme by showing a reduction mapping a full-blown scheme counterfeiter to a mini-scheme counterfeiter, such that the mini-scheme counterfeiter generates fake bank notes for the full-blown counterfeiter.

\begin{algorithm*}
\caption{The Interactive Private Money Scheme $\hat{\$}$}
\label{alg:full_scheme_construction}    

\fbox{
\procedure[linenumbering]{$\hat{\$}.\keygen(\secparam)$}{
    k_m \gets MAC.\keygen\\
    k_e \gets ENC.\keygen\\
    \textbf{return } (k_m, k_e)
}
}

\fbox{
\pseudocode[head = $\hat{\$}.\mint_{(k_m, k_e)}$]{
    \textbf{Acquirer} \<\< \textbf{Bank} \\[] [\hline]
    \pcln \label{{hat_dollar:mint:gen}}\<\< k_\$ \gets \$.\keygen(\secparam) \\
    \pcln \label{{hat_dollar:mint:enc}}\<\< c \gets ENC.\encrypt_{k_e}(k_\$) \\
    \pcln \label{{hat_dollar:mint:sign}}\<\< t \gets MAC.\mac_{k_m}(c) \pclb
    \pcintertext[center]{\fbox{\begin{minipage}{\dimexpr\textwidth-2\fboxsep-2\fboxrule\relax}
        \centering
        \ \\
        \ \\
        $\ket{\$} \gets \$.\mint_{k_\$}()$ \\
        \ \\
    \end{minipage}}}
    \pcln \< \sendmessageleft*{c, t} \<
}
}

\fbox{
\pseudocode[head = $\hat{\$}.\cverify_{(k_m, k_e)}({c, t, \ket{\$}})$]{
    \textbf{Acquirer} \<\< \textbf{Bank} \\[] [\hline]
    \pcln \label{{hat_dollar:verify:user_input}} \< \sendmessageright*{c, t} \< \\
    \pcln \label{{hat_dollar:verify:mac_verify}} \<\< r_m \gets MAC.\verify_{k_m}(c, t) \pclb
    \pcintertext[dotted]{\pcif $r_m = 1$}\\
    \pcln \label{{hat_dollar:verify:decrypt}} \<\< k_\$ \gets ENC.\decrypt_{k_e}(c) \pclb
    \pcintertext[center]{\fbox{\begin{minipage}{\dimexpr\textwidth-2\fboxsep-2\fboxrule\relax}
        \centering
        \ \\
        \ \\
        $r_v \gets \$.\cverify_{k_\$}(\ket{\$})$
        \ \\
        \ \\
    \end{minipage}}}
    \pcln \<\< \pcreturn r_v \pclb
    \pcintertext[dotted]{\pcif $r_m = 0$}
    \pcln \<\< \pcreturn 0
}
}
\end{algorithm*}

\begin{proposition}[Correctness of $\hat{\$}$]
    Assuming $\$$ is a \emph{correct} mini-scheme (see \cref{def:interactive_private_quantum_money}) and that both MAC and \nom{ENC}{A symmetric encryption scheme}{ENC} have \emph{perfect completeness}, $\hat{\$}$ (\cref{alg:full_scheme_construction}) is \emph{correct} (see \cref{def:interactive_private_quantum_money}).
    \label{prop:full_scheme_construction_completeness}
\end{proposition}

\begin{proof}
    From the \emph{perfect completeness} of MAC \ifnum\sigconf=0 (see \cref{def:MAC})\else (see the full version~\cite{RS19})\fi, we get that $\Pr[S_m] = 1$, where
    \begin{equation*}
        S_m \coloneqq MAC.\verify_{k_m}(c, MAC.\mac_{k_m}(c)) = 1\;.
    \end{equation*}
    Therefore, when the acquirer is honest, we know that he will send $t = MAC.sign_{k_m}(c)$ (that he received during the run of $\hat{\$}.\mint$) to the bank on line \ifnum\sigconf=0 \ref{{hat_dollar:verify:user_input}} \else \ref{hat_dollar:verify:user_input} \fi of $\hat{\$}.\cverify$. Thus, the MAC verification on line \ifnum\sigconf=0 \ref{{hat_dollar:verify:mac_verify}} \else \ref{hat_dollar:verify:mac_verify} \fi will succeed.

    From the \emph{perfect completeness} of ENC \ifnum\sigconf=0(see \cref{def:ENC})\else(see the full version~\cite{RS19})\fi, we get that $\Pr[S_e] = 1$, where
    \begin{align*}
        S_e &\coloneqq ENC.\decrypt_{k_e}(ENC.\encrypt_{k_e}(k_\$)) = k_\$\;.
    \end{align*}
    Therefore, when the acquirer is honest, we know that he will send $c = ENC.\encrypt_{k_e}(k_\$)$ (that he received during the run of $\hat{\$}.\mint$) to the bank on line \ifnum\sigconf=0 \ref{{hat_dollar:verify:user_input}} \else \ref{hat_dollar:verify:user_input} \fi of $\hat{\$}.\cverify$. Thus, the decryption in line \ifnum\sigconf=0 \ref{{hat_dollar:verify:decrypt}} \else \ref{hat_dollar:verify:decrypt} \fi will succeed.

    From the above, we conclude that for an honest acquirer both the decryption and MAC verification in $\hat{\$}.\cverify$ always succeeds. As such, the verification can only fail in $\$.\cverify_{k_\$}$. We know that the result of the decryption is $k_\$$ as it was generated in $\hat{\$}.\mint$, and that this $k_\$$ was generated by running $\$.\keygen$. Thus, from the \emph{correctness} of the mini-scheme $\$$ (see \cref{def:interactive_private_quantum_money}), we get that $\Pr[S_\$] \geq 1 - \negl$ for some negligible function $\negl$, where
    \begin{align*}
        S_\$ &\coloneqq k_\$ \sample \$.\keygen(\secparam); \ket{\$} \sample \$.\mint_{k_\$}();\\
        &\$.\cverify_{k_\$}(\ket{\$}) = 1\;.
    \end{align*}
    $\hat{\$}.\cverify$ passes when $MAC.\verify$, $ENC.\decrypt$ and $\$.\cverify$ all pass, so for an honest acquirer:
    \begin{align*}
        \Pr[(k_m, k_e) &\sample \hat{\$}.\keygen(\secparam); (c, t, \ket{\hat{\$}}) \sample \hat{\$}.\mint_{(k_m, k_e)}();\\
         &\hat{\$}.\cverify_{(k_m, k_e)}(c, t, \ket{\hat{\$}}) = 1]\\
         &= 1 - \Pr[\neg S_m \cup \neg S_e \cup \neg S_\$]\\
         &\geq 1 - \negl\;.
    \end{align*}
\end{proof}

\begin{theorem}[$\hat{\$}$ is a secure interactive private quantum money scheme]
    Assuming $\$$ is an interactive private quantum money mini-scheme, MAC is a PQ-EU-CMA MAC \ifnum\sigconf=0(see \cref{def:unforgeability_of_MAC}) \fi and ENC has PQ-IND-CPA \ifnum\sigconf=0 \ (see \cref{def:indistinguishability_of_encryption})\else(see the full version~\cite{RS19})\fi, $\hat{\$}$ (\cref{alg:full_scheme_construction}) is a secure interactive private quantum money scheme (see \cref{def:security_private_interactive_money}). Moreover, if $\$$ is semi-quantum, $\hat{\$}$ is also semi-quantum.
    \label{thm:full_scheme_construction_security}
\end{theorem}
\ifnum\sigconf=0

\begin{proof}
     The proof idea is very similar to that used in \cite[Appendix C]{BS16}: we show that the success probability of any full-scheme counterfeiter able to verify more banknotes than he received is upper-bounded by the success probability of a mini-scheme counterfeiter; a mini-scheme counterfeiter could guess which banknote the full-scheme counterfeiter will double-spend, generate fake banknotes, and with non-negligible probability double-spend the single mini-scheme banknote. The following sequence of games binds the success probability of any QPT full-scheme counterfeiter to that of a QPT mini-scheme counterfeiter:

    \paragraph{Game 0.} Let $\adv$ be a QPT full scheme counterfeiter. We assume that the amount of mints and verifications requested by $\adv$ is polynomial in $\secpar$ (i.e., $\ell$ and $v$ are polynomial in $\secpar$) since an adversary that does not comply with this assumption is not QPT. We define the first game to be the original interactive private quantum money security game, $\counterfeit{\adv}{\hat{\$}}{full}$ (see \cref{def:security_private_interactive_money}).

    Let $S_0$ be the event where $w > \ell$ (see \cref{def:security_private_interactive_money}) in Game 0 (this is the original win condition of the interactive private quantum money security game, since we assume $\ell$ and $v$ are polynomial in $\secpar$).

    \paragraph{Game 1.} We change Game 0 slightly by adding the condition that a specific banknote is double-spent: recall that $\ell$ and $v$ are the numbers of times $\hat{\$}.\mint$ and $\hat{\$}.\cverify$ are run during Game 0, respectively. In the start of Game 1 a uniform $i \in_R [\ell]$ is chosen by the bank. Let $(c_j, t_j, \ket{\$_j})$ be the result of the $j^{th}$ mint, and let $w_j$ be the amount of verifications such that $\hat{\$}.\cverify(c_j, t, \ket{\$}) = 1$ for some $t$, $\ket{\$}$. Let $\hat{j}$ be the smallest $j$ such that $w_j \geq 2$ ($\hat{j}=\infty$ if for all $j \in [\ell]: w_j < 2$).

    Let $S_1$ be the event where $w > \ell \wedge i = \hat{j}$ in Game 1.

    Assume $S_0$ occurred. Due to the unforgeability of MAC (see \cref{def:unforgeability_of_MAC}), we know that in every successful verification, $\adv$ presented $(c_j, t, \ket{\$})$ for some $j \in [\ell], t, \ket{\$}$\footnote{Suppose $\adv$ passes with non-negligible probability a verification of $(c, t, \ket{\$})$ such that $c \neq c_j \; \forall j \in [\ell]$. In that case $MAC.\verify(c, t)=1$ with non-negligible probability. We could use $\adv$ to construct a forger $\forger$ with non-negligible success probability: $\forger$ simulates a bank, but instead of signing and verifying with the MAC key generated in $\hat{\$}.\keygen$, he uses the signing and verification oracles. He then runs $\adv$ against the simulated bank, and will be able to present $c, t$ which pass MAC verification with non-negligible probability, while he did not ask for a tag of $c$ before since $c \neq c_j \; \forall j \in [\ell]$.}. Therefore, since $\adv$ was given only $\ell$ pairs $(c_j, t_j)$ (from the $\ell$ times that $\hat{\$}.\mint$ was run), and there were $w > \ell$ successful verifications from the assumption that $S_0$ occurred, then from the pigeonhole principle we conclude that $w_j \geq 2$ for some $j$, meaning $1 \leq \hat{j} \leq \ell$. Since $i\in[\ell]$ was chosen randomly and independently to $\hat{j}$, given $S_0$, there is a $\frac{1}{\ell}$ probability that $i = \hat{j}$; in which case $S_1$ occurs --- therefore $\Pr[S_1|S_0] = \frac{1}{\ell} \cdot \Pr[S_0]$.

    Assume $S_0$ did not occur: then we know $w \leq \ell$, meaning $S_1$ also did not occur --- namely, $\Pr[S_1|\neg S_0] = 0$. So $\Pr[S_1]=\frac{1}{\ell} \cdot \Pr[S_0]$, meaning $\Pr[S_0]$ is $\Pr[S_1]$ times some polynomial in $\secpar$.

    \paragraph{Game 2.} We now change the above game such that now, on the $i^{th}$ mint\footnote{$\adv$ can, of course, run several mint protocols simultaneously. We number them according to the order they were initiated.}, instead of encrypting and signing the mini-scheme key from line \ifnum\sigconf=0 \ref{{hat_dollar:mint:gen}} \else \ref{hat_dollar:mint:gen} \fi (the one later used in $\$.\mint$), the bank encrypts and signs $0_\$$ (where $0_\$$ is a string of $0$'s the length of a mini-scheme key); i.e., on the $i^{th}$ mint we replace lines \ifnum\sigconf=0 \ref{{hat_dollar:mint:enc}} \else \ref{hat_dollar:mint:enc} \fi and \ifnum\sigconf=0 \ref{{hat_dollar:mint:sign}} \else \ref{hat_dollar:mint:sign} \fi with:
    \begin{align*}
        &c \gets ENC.\encrypt_{k_e}(0_\$)\\
        &t \gets MAC.\mac_{k_m}(c)\;.
    \end{align*}
    On $\hat{\$}.\cverify(c, t, \ket{\$})$, if $ENC.\decrypt_{k_e}(c) = 0_\$$, then the bank runs $\$.\cverify$ with the original mini-scheme key that was used in the $i^{th}$ mint (the one originally generated in line \ifnum\sigconf=0 \ref{{hat_dollar:mint:gen}} \else \ref{hat_dollar:mint:gen} \fi of the $i^{th}$ mint) rather than with $0_\$$.

    Let $S_2$ be the event where $w > \ell \wedge i = \hat{j}$ in Game 2.

    Game 2 is different from Game 1 only in the $i^{th}$ mint, and the sole difference in the $i^{th}$ mint is that $\adv$ receives an encrypted and signed $0_\$$ rather than the key that was used in $\$.\mint$. Similarly, in a verification for $(c, t, \ket{\$})$ such that $c=ENC.\encrypt_{k_e}(0_\$)$ for some $t, \ket{\$}$, the mini-scheme bank verifies $\ket{\$}$ with the mini-scheme key that was generated in the $i^{th}$ mint, that in Game 1 is sent instead of $0_\$$.

    That means that the only difference between Game 2 and Game 1 is in what $\adv$ receives on the $i^{th}$ mint; on the $i^{th}$ mint, $\adv$ receives a signed encryption of a random key rather than the key used to sign the mini-scheme banknote he received, but the same key will be used to verify it, just like on a normal verification. So, due to the indistinguishability of ENC, replacing an encryption of one message with the encryption of another message of the same length\footnote{Indistinguishability works for messages of the same length. Here we assume, without loss of generality, that $\keygen(\secparam)$ always outputs keys of the same length.} cannot change the behavior of $\adv$, i.e., $|\Pr[S_2]-\Pr[S_1]| \leq  \negl$\footnote{Assume $\abs{\Pr[S_2]-\Pr[S_1]}$ is non-negligible. Assume without loss of generality that $\Pr[S_2] \leq \Pr[S_1]$. In that case, we could construct a distinguisher $\dist$ with non-negligible success probability: $\dist$ will simulate a bank, but instead of encrypting with the ENC key generated in $\hat{\$}.\keygen$, he will use the encryption oracle, and instead of decrypting he will "remember" each encryption he made and thus could match each encryption to the relevant key (any unrecognized encryption would not have passed the real bank verification either, because the encryptions are MAC signed). On the $i^{th}$ mint, he will present a random key and the actual mini-scheme key used in that mint as $m_0$ and $m_1$ (the chosen messages whose encryptions he needs to recognize in the CPA game) respectively, and proceed with the encryption he received to finish the game. $\dist$ returns $b'=1$ (guessing the encryption he received was of the real key) if and only if he wins the counterfeiting game (since he wins with higher probability when he receives encryption of the real key). $\dist$ has a $\frac{1}{2} + \frac{\abs{\Pr[S_2]-\Pr[S_1]}}{2}$ probability to win, which is non-negligible, in contradiction to the security of ENC.}.

    \paragraph{Bound on success probability.} We show a reduction mapping a Game 2 counterfeiter to a mini-scheme counterfeiter:

    Let $\adv$ be a QPT full-scheme counterfeiter. We construct a mini-scheme counterfeiter $\miniadv$ in the following manner:

    $\miniadv$ simulates the bank of Game 2, with one exception: on the $i^{th}$ mint, instead of generating the actual mini-scheme key and banknote, $\miniadv$ asks the actual mini-scheme bank to run $\$.\mint$. Similarly, when performing a verification for $0_\$$, $\miniadv$ asks the actual mini-scheme bank for verification. $\miniadv$ runs $\adv$ against the altered version of Game 2. The only difference from the original Game 2 is that on the $i^{th}$ mint $\miniadv$ asks the bank to generate the banknote, and when he receives $0_\$$ he asks the bank to verify that same note. The honest mini-scheme bank runs minting and verification on that banknote in the exact same way as the bank in Game 2 should, meaning that $\Pr[S_2]$ is unchanged for any $\adv$ by the simulated Game 2. In the case that $S_2$ occurred, $\miniadv$ passed at least two verifications with $0_\$$, meaning he passed two verifications with the actual bank, while only asking mint once. So $\miniadv$ has a probability of $\Pr[S_2]$ to pass win the mini-scheme counterfeiting game, and we showed that the success probability of any QPT counterfeiter to do so is negligible --- meaning $\Pr[S_2]$ must be negligible.

    From construction, $\Pr[S_0]$ is also negligible. Game 0 was defined to be the original full-scheme security game, and $S_0$ was defined as its win condition; so $\hat{\$}$ is secure (see \cref{def:security_private_interactive_money}).
\end{proof}

\else
The proof of \cref{thm:full_scheme_construction_security} is given in the full version~\cite{RS19}. \fi

\section{LWE Implies Semi-Quantum Money}
\label{sec:main_thm_proof}
For convenience, we restate the main theorem of our private semi-quantum money result:

\privateresult*

\begin{proof}
    From \ifnum\sigconf=0 Theorem \ref{thm:LWE_implies_NTCF} \else \cite[Theorem 26]{BCM+18} (see also the full version~\cite{RS19})\fi we get that the hardness of LWE with certain parameters implies that an NTCF family exists. From Theorem \ref{thm:NTCF_implies_1-of-2_puzzles} we get that an NTCF implies $\frac{1}{2}$-hard 1-of-2 puzzles, and from Corollary \ref{cor:weak_implies_strong_puzzle} we get that weak 1-of-2 puzzles (and in particular, $\frac{1}{2}$-hard 1-of-2 puzzles) imply strong 1-of-2 puzzles.

     From \cite{BZ13,GHS16} we get that the hardness of LWE with certain parameters\footnote{Both \cite{BZ13} and \cite{GHS16} rely on Quantum Pseudorandom Functions (QPRF). From Banerjee et al.~\cite{BPR12} and Zhandry~\cite{Zha12} we get that QPRFs can be constructed from LWE with certain parameters.} (that are different to those used for NTCF) implies that a PQ-EU-CMA MAC and a PQ-IND-CPA encryption exist. By combining these with  Propositions~\ref{prop:mini_scheme_construction_correctness},~\ref{prop:dollar_z_is_2_of_2},~\ref{prop:mini_scheme_construction} and ~\ref{prop:full_scheme_construction_completeness} and \cref{thm:full_scheme_construction_security} (based on the constructions of \cref{alg:mini_scheme_construction} and \cref{alg:full_scheme_construction}), we get secure semi-quantum private money from 1-of-2 puzzles.
\end{proof}

\section{Semi-Quantum Money Requires Computational Assumptions}
\label{sec:computational_assumptions_are_necessary}
Both our public and private semi-quantum money constructions rely on computational assumptions. For our public construction we know this is necessary directly from the fact that public quantum money schemes cannot be secure against computationally unbounded adversaries~\cite{AC13}. However, it is known that private quantum money schemes could be information-theoretically secure\footnote{Wiesner's original quantum money scheme \cite{Wie83} was proven to be information-theoretically secure \cite{MVW13}, and Pastawski et al. \cite{PYJ+12} even showed a slightly modified version of the scheme that is also noise-tolerant.}. In this section we show that for semi-quantum money schemes, and in fact for any quantum money scheme with classical minting, computational assumptions are necessary. We do so by showing an adversary that, after running a single minting protocol with the bank, can create 2 states which pass verification with non-negligible probability:

\begin{theorem}\label{thm:no_it_secure_classical_minting}
Let \nom{$\$_C$}{A quantum money scheme with classical minting}{DOLLAR_C} be a quantum money scheme with classical minting. Then there is a computationally unbounded adversary $\adv$ such that:
\begin{align*}
    \Pr[k &\sample \$_C.\keygen(\secparam); \ket{\$} \sample \$_C.\mint_k(\secparam); (\ket{\$_1}, \ket{\$_2}) \gets \adv(\secparam, \ket{\$});\\
    &\$_C.verify(\ket{\$_1}) = \$_C.verify(\ket{\$_2}) = 1] = p
\end{align*}
where $p$ is non-negligible.
\end{theorem}

\begin{proof}
    The adversary's strategy is simple --- she simply recreates the exact run of the protocol on her own by aborting and rerunning whenever a measurement result differs from the original run with the bank, and uses the messages from the original run to substitute for the bank. Eventually she will be able to rerun the protocol with the same measurement results (this would occur only after an exponential number of tries, but this is fine since the adversary is unbounded), resulting in another quantum money state identical to that minted with the bank.
\end{proof}

Note that \cref{thm:no_it_secure_classical_minting} does not hold for a memory-dependant scheme (see \cref{def:memory_dependent}) where the database can be modified during the verification protocol, since the bank could keep track of money that has been verified. However, a private\footnote{This discussion is not relevant to the public setting.} scheme where this is allowed is not interesting, since if such modifications are considered we can even construct a simple secure classical scheme (see \cref{sec:advantage_of_statelessness}). Therefore, when considering memory-dependant private quantum money schemes, the database can be modified only during the minting protocol.

It is known that only private quantum money schemes which are memory-dependant can be information-theoretically secure --- Aaronson \cite[Theorem 8]{Aar18} showed that any semi-quantum money scheme with a fixed key (i.e., any scheme that is not memory-dependant) is vulnerable to counterfeiting by an unbounded adversary. Obviously, this also holds for schemes with classical minting that are not memory-dependant. However, \cref{thm:no_it_secure_classical_minting} also holds for schemes with classical minting which \textit{are} memory-dependant, thus getting a stronger result --- that no quantum money scheme with classical minting can be information-theoretically secure.

\section{Discussion}
\label{sec:discussion}
The main question that is raised in this work is the following. There are many multi-party quantum cryptographic protocols which require that both parties have quantum resources. This work elicits an important question: is there a way (preferably, as general as possible) to convert some of these protocols to ones in which at least one of the parties does not need a quantum computer? A weaker open question can be posed from the perspective of device-independent cryptography: can at least one party use an \emph{untrusted} quantum computer in unison with a trusted classical computer? We emphasize that device independent protocols (see~\cite{VV19,FRV19} and references therein), such as DI quantum key distribution, DI randomness expansion\footnote{Also known as certified randomness.} and randomness amplification, use unconditional (information theoretic) security notions, while our protocols are only computationally secure.

Additional questions are raised when considering Remote State Preparation (RSP). RSP is a protocol that allows a classical client to create quantum states on an untrusted quantum server, with different levels of security; a protocol can promise blindness~\cite{DK16} (i.e., the server learns nothing about the state) and verifiability~\cite{GV19, CCKW19} (i.e., the client can verify the right states were created). This raises the following question: is there a general way to turn a classically-verifiable quantum money scheme into a semi-quantum money scheme using remote state preparation for the user-side minting? We note that Pastawski et al.~\cite{PYJ+12} proved that a simple variant of Wiesner's scheme is classically verifiable. Their scheme also tolerates constant level of noise. Their construction only requires preparations of the states $\ket{0},\ket{1},\ket{+}$ and $\ket{-}$ which the RSP constructions above support. A recent work by Badertscher et al. \cite{BCC+20} shows that an RSP protocol cannot be composable in the abstract cryptography framework --- this result is somewhat discouraging, but there might still be a way to use RSP to create semi-quantum money schemes.

Moreover, can the randomness generation protocol from~\cite{BCM+18} be amplified by our parallel repetition result? Currently the protocol has $N$ rounds; could this number be made constant using parallel repetition?

%Is it crucial that the $n$ puzzles  $((p_1,\ldots,p_n),(v_1,\ldots,v_n))$ are generated \emph{independently} (see Eq.~\eqref{eq:repetition_G_n} in Definition~\ref{def:repetition}) to get the exponential hardness in Theorem~\ref{thm:parallel_repetition_1_of_2_puzzle}? The answer is positive: An $\solvertwo$ adversary for \emph{repeated copies} of 1-of-2 puzzle could solve the first 1-of-2 puzzle with probability $\frac{1}{2}$, and copy that same answer  for all the $n-1$ other 1-of-2 puzzles. This gives a probability of $\frac{1}{2}$ (rather than $\frac{1}{2^n}$ to solve all $n$ 1-of-2 puzzles.

%In addition to using the same puzzle as in the preceding paragraph, one could also add the requirement that the solution to the $n$ 1-of-2 puzzles must be unique (i.e., require that for $V^n$ in Eq.~\eqref{eq:repetition_V_n} to accept, $a_i\neq a_j$ for all $i\neq j \in [n]$). That would not affect the completeness in our construction. We leave it as an open question whether the soundness parameter improves exponentially as in Theorem~\ref{thm:parallel_repetition_1_of_2_puzzle} in that setting, for the 1-of-2 puzzle in Algorithm~\ref{alg:1-of-2-puzzle_from_NTCF}, based on the LWE construction. For our application, such an improvement would not change  the asymptotic running times and overall communication. 

\subsection*{Acknowledgments}
We wish to thank Zvika Brakerski, Urmila Mahadev and Thomas Vidick for fruitful discussions. We also wish to thank an anonymous referee for insightful comments. 
\ifnum \anon=0
O.S. and R.R. are supported by the Israel Science Foundation (ISF) grant No. 682/18 and 2137/19, and by the the Cyber Security Research Center at Ben-Gurion University.
\fi

\ifnum\sigconf=1
    \bibliographystyle{ACM-Reference-Format}
\else
    \ifnum\cryptology=1
        \bibliographystyle{abbrvurl}
    \else
        \bibliographystyle{alphaabbrurldoieprint}
    \fi
\fi

\ifnum\masterthesis=0
    \ifnum\cryptology=0
        \mybib
    \fi
\fi

\ifnum\sigconf=0

\appendix
\section{Nomenclature}
\label{sec:nomenclature}
\renewcommand{\nomname}{} \vspace{-10pt} \printnomenclature %Removed the "Nomenclature" title, since this is already the heading of the appendix

\section{Preliminaries}
\label{sec:preliminaries}
This appendix contains mainly the standard definitions of private-key encryption and message authentication codes (MAC), and can be safely skipped by readers already familiar with these notions.

We use the standard definitions for negligible, non-negligible and noticeable functions --- see, e.g.,~\cite{Gol01}. %We use the shorthand \nom{PPT}{Probabilistic Polynomial Time}{P} for Probabilistic Polynomial Time and \nom{QPT}{Quantum Polynomial Time}{Q} for Quantum Polynomial Time. We denote $[n] \equiv \{1, \dots, n\}$. For a finite set $S$, we denote $s \in_R S$ to be the process in which $s$ is sampled uniformly from $S$.

\begin{definition}[{Private-key encryption system, \cite[Definition 3.7]{KL14}}]\label{def:ENC}
    A private-key encryption scheme consists of three PPT algorithms $\keygen$, $\encrypt$ and $\decrypt$ such that:
    \begin{enumerate}
        \item The randomized key-generation algorithm $\keygen$ takes as input $\secparam$ and outputs a key $k \gets \keygen(\secparam)$.
        \item The (possibly randomized) encryption algorithm $\encrypt$ takes as input a key $k$ and a plaintext message $m \in \{0, 1\}^*$, and outputs a ciphertext $c \gets \encrypt_k(m)$.
        \item The deterministic decryption algorithm $\decrypt$ takes as input a key $k$ and a ciphertext $c$, and outputs a message $m \coloneqq Dec_k(c)$.
    \end{enumerate}
    A private-key encryption system is required to have \emph{perfect completeness}, meaning that for every $\secpar$, every $k$ output by $\keygen(\secparam)$, and every $m \in \{0, 1\}^*$, it holds that $\decrypt_k(\encrypt_k(m)) = m$.
\end{definition}

\begin{definition}[{\nom{PQ-IND-CPA}{Post-quantum indistinguishable encryptions under a chosen-plaintext attack}{PQINDCPA}, adapted from \cite[Definition 3.22]{KL14}}]\label{def:indistinguishability_of_encryption}
    A private-key encryption scheme $\Pi$ has \emph{post-quantum indistinguishable encryptions under a chosen-plaintext attack} (PQ-IND-CPA) if for every QPT distinguisher \nom{$\dist$}{An encryption distinguisher --- i.e., an adversary whose goal is to distinguish between an actual encryption and a random string}{D} there is a negligible function $\negl$ such that, for all $\secpar$:
    \begin{equation*}
        \Pr[\indcpa_{\dist,\Pi}(\secpar) = 1] \leq \frac{1}{2} + \negl\;.
    \end{equation*}
    
    The indistinguishability game $\indcpa_{\dist,\Pi}(\secpar)$:
    \begin{enumerate}
        \item A key $k$ is generated by running $\keygen(\secparam)$.
        \item The distinguisher $\dist$ is given input $\secparam$ and classical oracle access to $\encrypt_k(\cdot)$, and outputs a pair of messages $m_0, m_1$ of the same length.
        \item A uniform bit $b \in_R \{0, 1\}$ is chosen, and then a ciphertext $c \gets \encrypt_k(m_b)$ is computed and given to $\dist$.
        \item $\dist$ continues to have oracle access to $\encrypt_k(\cdot)$ and outputs a bit $b'$.
        \item The output of the game is defined to be $1$ if $b'=b$, and $0$ otherwise. In the former case, we say that $\dist$ \emph{succeeds}.
    \end{enumerate}
\end{definition}

\begin{definition}[{Message authentication code \cite[Definition 4.1]{KL14}}]\label{def:MAC}
    A message authentication code (\nom{MAC}{Message Authentication Code}{MAC}) consists of 3 PPT algorithms $\keygen,\ \mac$ and $\verify$ satisfying:
    \begin{enumerate}
        \item $\keygen$ takes as input the security parameter $\secparam$ and outputs a key $k$
        \item $\mac$ takes as input a key $k$ and a message $m \in \{0,1\}^*$ and outputs a tag $t \gets \mac_k(m)$.
        \item $\verify$ takes as input a key $k$, a message $m$, and a tag $t$. It outputs a bit $b \coloneqq \verify_k(m,t)$, with $b=1$ meaning \textbf{valid} and $b=0$ meaning \textbf{invalid}.
    \end{enumerate}
    A MAC is required to have \textit{perfect completeness}, i.e., for every $\secpar$, every key $k \gets \keygen(\secparam)$ and every $m \in \{0,1\}^*$, it holds that $\verify_k(m, \mac_k(m))=1$.
\end{definition}

\begin{definition}[{\nom{PQ-EU-CMA}{Post-quantum existentially unforgeable under an adaptive chosen-message attack}{PQEUCMA} MAC, adapted from \cite[Definition 4.2]{KL14}}]\label{def:unforgeability_of_MAC}
    A message authentication code $\Pi$ is \emph{Post-Quantum Existentially Unforgeable under an adaptive Chosen-Message Attack} (PQ-EU-CMA) if for every QPT forger \nom{$\forger$}{A MAC or digital signature forger --- i.e., an adversary whose goal is to break the security of a MAC or digital signature scheme}{F}, there exists a negligible function $\negl$ such that:
    \begin{equation*}
        \Pr[\macforge{\Pi} = 1] \leq \negl\;.
    \end{equation*}
    
    The CMA message authentication game $\macforge{\Pi}$:
    \begin{enumerate}
        \item A key $k$ is generated by running $\keygen(\secparam)$.
        \item The forger $\forger$ is given input $\secparam$, classical oracle access to $\mac_k(\cdot)$ and classical oracle access to $\verify_k(\cdot)$ (note that the forger cannot query the oracles in superposition). The forger eventually outputs $(m,t)$. Let $\mathcal{Q}$ denote the set of all queries that $\forger$ asked its signing oracle. 
        \item $\forger$ succeeds if and only if \textit{(1)} $\verify_k(m,t) = 1$ and \textit{(2)} $m \notin \mathcal{Q}$. In that case the output of the game is defined to be $1$.
    \end{enumerate}
\end{definition}

\begin{definition}[{Digital signature scheme \cite[Definition 12.1]{KL14}}]\label{def:digital_signature}
    A digital signature scheme consists of three PPT algorithms $\keygen$, $\sign$ and $\verify$ such that:
    \begin{enumerate}
        \item The key-generation algorithm $\keygen$ takes as input a security parameter $\secparam$ and outputs a pair of keys $(pk, sk)$. These are called the public key and the private key, respectively. We assume that $pk$ and $sk$ each has length of at least $\secpar$, and that $\secpar$ can be determined from either.
        \item The signing algorithm $\sign$ takes as input a private key $sk$ and a message $m$. It outputs a signature $\sigma \gets \sign_{sk}(m)$.
        \item The deterministic verification algorithm $\verify$ takes as input a public key $pk$, a message $m$ and a signature $\sigma$. It outputs a bit $b \gets \verify_{sk}(m,\sigma)$, with $b=1$ meaning \textbf{valid} and $b=0$ meaning \textbf{invalid}.
    \end{enumerate}
    A digital signature scheme is required to have \emph{perfect completeness}, meaning that except with negligible probability over $(pk, sk)$ output by $\keygen(\secparam)$, it holds that $\verify_{pk}(m, \sign_{sk}(m))=1$ for every legal message $m$.
\end{definition}

\begin{definition}[{PQ-EU-CMA digital signature scheme, adapted from \cite[Definition 12.2]{KL14}}]\label{def:unforgeability_of_digital_signature}
    A digital signature scheme $\Pi$ is \emph{Post-Quantum Existentially Unforgeable under an adaptive Chosen Message Attack} (PQ-EU-CMA) if for every QPT forger $\forger$, there exists a negligible function $\negl$ such that:
    \begin{equation*}
        \Pr[\sigforge{\Pi} = 1] \leq \negl\;.
    \end{equation*}
    
    The signature experiment $\sigforge{\Pi}$:
    \begin{enumerate}
        \item $\keygen$ is run to generate to obtain keys $(pk, sk)$.
        \item Forger $\forger$ is given $pk$ and access to a signing oracle $\sign_{sk}(\cdot)$. The forger than outputs $(m, \sigma)$. Let $Q$ denote the set of all queries that $\forger$ asked its oracle.
        \item $\forger$ succeeds iff $\verify_{pk}(m, \sigma)=1$ and $m \notin Q$. In this case the output of the experiment is defined to be 1 (and otherwise 0).
    \end{enumerate}
\end{definition}

\begin{lemma}[{Difference Lemma \cite[Lemma 1]{Sho04}}]
    \label{lem:difference}
    Let $A, B, F$ be events defined in some probability distribution, and suppose that $A \wedge \neg F \iff B \wedge \neg F$. Then $\abs{\Pr[A] - \Pr[B]} \leq \Pr[F]$.
\end{lemma}

\section{Quantum Lightning with Bolt-to-Certificate}
\label{sec:ql_with_bolt_to_cert}
The following definitions are taken almost verbatim from \cite{CS20}. The definitions originate in \cite{Zha19} and \cite{Col19}, but in this work we use the notations of the superseding work \cite{CS20}.

\begin{definition}[{Quantum Lightning \cite{Zha19}}]\label{def:quantum_lightning}
    A quantum lightning scheme consists of a PPT algorithm $\ql.\setup(\secparam)$ (where $\secpar$ is a security parameter) which samples a pair of QPT algorithms $(\genbolt, \verbolt)$. $\genbolt$ outputs a pair of the form $\ket{\psi} \in \mathcal{H}_\$, s \in \{0,1\}^\secpar$. We refer to $\ket{\psi}$ as a "bolt" and to $s$ as a "serial number". $\verbolt$ takes as input a pair of the same form, and outputs either "accept" (1) or "reject" (0). They satisfy the following:
    \begin{itemize}
        \item
        \begin{align*}
            \Pr[(\genbolt, \verbolt) &\gets \ql.\setup(\secparam); (\ket{\psi}, s) \gets \genbolt() : \\
            &\verbolt(\ket{\psi}, s) = 1] \\
            &= 1 - \negl\;.
        \end{align*}
        \item For all $s' \in \{0,1\}^\secpar$:
        \begin{align*}
            \Pr[(\genbolt, \verbolt) &\gets \ql.\setup(\secparam); (\ket{\psi}, s) \gets \genbolt() : \\
            &s \neq s' \wedge \verbolt(\ket{\psi}, s') = 1] \\
            &= \negl\;.
        \end{align*}
    \end{itemize}
\end{definition}

\begin{definition}[Security \cite{Zha19}]\label{def:ql_security}
    A quantum lightning scheme $\ql$ is secure if, for all QPT bolt forgers \nom{$\boltforger$}{A quantum lightning adversary, i.e., an adversary whose goal is to pass verification for two bolts with the same serial number}{L}:
    \begin{equation*}
        \Pr[\boltforge{\ql} = 1] = \negl\;.
    \end{equation*}
    
    The bolt forging game $\boltforge{\ql}$:
    \begin{enumerate}
        \item The challenger runs $(\genbolt, \verbolt) \gets \ql.\setup(\secparam)$ and sends $(\genbolt, \verbolt)$ to $\boltforger$.
        \item $\boltforger$ produces a pair $\ket{\Psi_{12}} \in \mathcal{H}_\$^{\otimes 2}, s \in \{0,1\}^\secpar$.
        \item The challenger runs $\verbolt(\cdot, s)$ on each half of $\ket{\Psi_{12}}$. The output of the game is 1 if both outcomes are "accept" (and otherwise 0).
    \end{enumerate}
\end{definition}

\begin{definition}[Bolt-to-certificate]\label{def:bolt_to_certificate}
    For a quantum lightning scheme $\ql$ to have bolt-to-certificate capability, we change the procedure $\ql.\setup(\secparam)$ slightly, so that it outputs a quadruple
    \begin{equation*}
        (\genbolt, \verbolt, \gencert, \vercert)\;,
    \end{equation*}
    where $\gencert$ is a QPT algorithm that takes as input a quantum money state and a serial number and outputs a classical string of some fixed length $l(\secpar)$ for some polynomially bounded function $l$, to which we refer as a \emph{certificate}, and $\vercert$ is a PPT algorithm that takes as input a serial number and a certificate, and outputs "accept" (1) or "reject" (0).
    
    Let $\secpar \in \mathbb{N}$. We say that a quantum lightning scheme $\ql$ has bolt-to-certificate capability if:
    \begin{itemize}
        \item
        \begin{align*}
            \Pr[(\genbolt&, \verbolt, \gencert, \vercert) \gets \ql.\setup(\secparam); \\
            &(\ket{\psi}, s) \gets \genbolt(); c \gets \gencert(\ket{\psi}, s) : \\
            & \vercert(s, c) = 1] \\
            &= 1 - \negl\;.
        \end{align*}
        \item For all QPT algorithms $\certforger$:
        \begin{equation*}
            \Pr[\certforge{\ql} = 1] = \negl\;.
        \end{equation*}
    \end{itemize}
    
    The certificate forging game $\certforge{\ql}$:
    \begin{enumerate}
        \item The challenger runs
        \begin{equation*}
            (\genbolt, \verbolt, \gencert, \vercert) \gets \ql.\setup(\secparam)\;,
        \end{equation*}
        and sends the quadruple to $\certforger$.
        \item $\certforger$ returns $c \in \{0,1\}^{l(\secpar)}$ and $(\ket{\psi}, s)$.
        \item The challenger runs $\vercert(s,c)$ and $\verbolt(\ket{\psi}, s)$, and outputs 1 if they both accept (otherwise outputs 0).
    \end{enumerate}
\end{definition}

\section{Trapdoor Claw-Free Families}
\label{sec:definition_NTCF}
Most of this section is taken verbatim from Brakerski et al.~\cite{BCM+18}.
Let $\secpar$ be a security parameter, and let $\sX$ and $\sY$ be finite sets (depending on $\secpar$). For our purposes, an ideal family of functions $\mathcal{F}$ would have the following properties. For each public key $k$, there are two functions $\{f_{k,b}:\sX\rightarrow \sY\}_{b\in\{0,1\}}$ that are both injective, that have the same range (equivalently, $(b,x)\mapsto f_{k,b}(x)$ is $2$-to-$1$), and that are invertible given a suitable trapdoor $t_k$ (i.e., $t_k$ can be used to compute $x$ given $b$ and $y=f_{k,b}(x)$). Furthermore, the pair of functions should be claw-free: it must be hard for an attacker to find two pre-images $x_0,x_1\in\sX$ such that $f_{k,0}(x_0) = f_{k,1}(x_1)$. Finally, the functions should satisfy an adaptive hardcore bit property, which is a stronger form of the claw-free property: assuming for convenience that $\sX= \{0,1\}^w$, we want it to be computationally infeasible to simultaneously generate $(b,x_b)\in\{0,1\}\times \sX$ and a non-zero string $d\in \{0,1\}^w$ such that with a non-negligible advantage over $\frac{1}{2}$ the equation $d\cdot (x_0\oplus x_1)=0$ holds, where $x_{1-b}$ is defined as the unique element such that $f_{k,1-b}(x_{1-b})=f_{k,b}(x_b)$.

Unfortunately, we (as well as Brakerski et al.) do not know how to construct a function family that exactly satisfies all these requirements under standard cryptographic assumptions. Instead, Brakerski et al. construct a family that satisfies slightly relaxed requirements based on the hardness of the learning with errors (LWE) problem, and we will show that these are still adequate for our purposes. The requirements are relaxed as follows. First, the range of the functions is no longer a set $\sY$; instead, it is  $\mathcal{D}_{\sY}$, the set of probability densities over $\sY$. That is, each function returns a density, rather than a point. The trapdoor injective pair property is then described in terms of the support of the output densities: these supports should either be identical for a colliding pair or be disjoint in all other cases. 

The consideration of functions that return densities elicits an additional requirement of efficiency: there should exist a quantum polynomial-time procedure that efficiently prepares a superposition over the range of the function, i.e., for any key $k$ and $b\in\{0,1\}$, the procedure can prepare a state that is close (up to a negligible trace distance) to the state
\begin{equation*}
\frac{1}{\sqrt{\sX}}\sum_{x\in \sX, y\in \sY}\sqrt{f_{k,b}(x)(y)}\ket{x}\ket{y}\;.
\end{equation*}
%In their instantiation based on LWE, it is not possible to prepare~\eqref{eq:perfectsuperposition} perfectly, but it is possible to create a superposition with coefficients $\sqrt{f'_{k,b}(x)}$, such that the resulting state is within negligible trace distance of~\eqref{eq:perfectsuperposition}. The density $f'_{k,b}(x)$ is required to satisfy two properties used in our protocol. First, it must be easy to check, without the trapdoor, if an $y\in \sY$ lies in the support of $f'_{k,b}(x)$. Second, the inversion algorithm should operate correctly on all $y$ in the support of $f'_{k,b}(x)$.

We modify the adaptive hardcore bit requirement slightly. Since the set $\sX$ may not be a subset of binary strings, we first assume the existence of an injective, efficiently invertible map $\inj:\sX\to\{0,1\}^w$. Next, we only require the adaptive hardcore bit property to hold for a subset of all nonzero strings rather than for the set $\{0,1\}^w\setminus \{0^w\}$. Finally, membership in the appropriate set should be efficiently checkable, given access to the trapdoor. 

\begin{definition}[NTCF family]\label{def:trapdoorclawfree}
Let $\secpar$ be a security parameter. Let $\sX$ and $\sY$ be finite sets.
 Let $\mathcal{K}_{\mathcal{F}}$ be a finite set of keys. A family of functions 
$$\mathcal{F} \,=\, \big\{f_{k,b} : \sX\rightarrow \mathcal{D}_{\sY} \big\}_{k\in \mathcal{K}_{\mathcal{F}},b\in\{0,1\}}     $$
is called a \textbf{noisy trapdoor claw free (NTCF) family} if the following conditions hold:

\begin{enumerate}
\item{\textbf{Efficient Function Generation.}}\label{it:NTCF_efficient_function_generation}
 There exists an efficient probabilistic algorithm $\keygen_{\mathcal{F}}$ which generates a key $k\in \mathcal{K}_{\mathcal{F}}$ together with a trapdoor $t_k$: 
$$(k,t_k) \leftarrow \keygen_{\mathcal{F}}(\secparam)\;.$$
\item{\textbf{Trapdoor Injective Pair.}} \label{it:NTCF_trapdoor_injective_pair} For all keys $k\in \mathcal{K}_{\mathcal{F}}$ the following conditions hold. 
	\begin{enumerate}
	\item \textit{Trapdoor}:\label{it:NTCF_trapdoor} For all $b\in\{0,1\}$ and $x\neq x' \in \sX$, $\supp(f_{k,b}(x))\cap \supp(f_{k,b}(x')) = \emptyset$. Moreover, there exists an efficient deterministic algorithm $\textrm{INV}_{\mathcal{F}}$ such that for all $b\in \{0,1\}$,  $x\in \sX$ and $y\in \supp(f_{k,b}(x))$, $\textrm{INV}_{\mathcal{F}}(t_k,b,y) = x$. 

	\item \textit{Injective pair}:\label{it:NTCF_injective_pair}  There exists a perfect matching $\sR_k \subseteq \sX \times \sX$ such that $f_{k,0}(x_0) = f_{k,1}(x_1)$ if and only if $(x_0,x_1)\in \sR_k$. 
	\end{enumerate}

\item{\textbf{Efficient Range Superposition.}\label{it:NTCF_efficient_range_superposition}%
\footnote{Here we use a slightly weaker (and simpler) definition compared to Brakerski et al. Our definition follows from theirs by using Lemma 2 in~\cite{BCM+18}, which relates the Hellinger distance to the trace distance.}}
There exists an efficient procedure SAMP$_{\mathcal{F}}$ that on input $k$ and $b\in\{0,1\}$ prepares a state $\ket{\psi'}$ which has a negligible trace distance to the state 
\begin{equation*}
    \ket{\psi} = \frac{1}{\sqrt{|\sX|}}\sum_{x\in \sX,y\in \sY}\sqrt{(f_{k,b}(x))(y)}\ket{x}\ket{y}\;.
\end{equation*}

\item{\textbf{Adaptive Hardcore Bit.}} \label{it:NTCF_adaptive_hardcore_bit}
For all keys $k\in \mathcal{K}_{\mathcal{F}}$ the following conditions hold, for some integer $w$ that is a polynomially bounded function of $\secpar$. 
	\begin{enumerate}
	\item \label{it:NTCF_adaptive_hardcore_bit_a} For all $b\in \{0,1\}$ and $x\in \sX$, there exists a set $\dset_{k,b,x}\subseteq \{0,1\}^{w}$ such that $\Pr_{d\leftarrow_U \{0,1\}^w}[d\notin \dset_{k,b,x}]$ is negligible, and moreover there exists an efficient algorithm that checks for membership in $\dset_{k,b,x}$ given $k,b,x$ and the trapdoor $t_k$. 
	\item \label{it:NTCF_adaptive_hardcore_bit_b} There is an efficiently computable injection $\inj:\sX\to \{0,1\}^w$, such that $\inj$ can be inverted efficiently on its range, and such that the following holds. If
	\begin{eqnarray}\label{eq:defsetsH}
	\begin{aligned}
	H_k &=& \big\{(b,x_b,d,d\cdot(\inj(x_0)\oplus \inj(x_1)))\,|\\ &&b\in \{0,1\},\; (x_0,x_1)\in \mathcal{R}_k,\; d\in \dset_{k,0,x_0}\cap \dset_{k,1,x_1}\big\}\;,\text{\footnotemark}\\
	\overline{H}_k &=& \{(b,x_b,d,c)\,|\; (b,x,d,c\oplus 1) \in H_k\big\}\;,
	\end{aligned}
	\end{eqnarray}
	\footnotetext{Note that although both $x_0$ and $x_1$ are referred to to define the set $H_k$, only one of them, $x_b$, is explicitly specified in any $4$-tuple that lies in $H_k$.}
	then for any quantum polynomial-time procedure $\mathcal{A}$ there exists a negligible function $\mu(\cdot)$ such that 
	\begin{align}\label{eq:adaptive-hardcore}
	\begin{split}
	\Big|\Pr_{(k,t_k)\leftarrow \keygen_{\mathcal{F}}(\secparam)}&[\mathcal{A}(k) \in H_k] - \Pr_{(k,t_k)\leftarrow \keygen_{\mathcal{F}}(\secparam)}[\mathcal{A}(k) \in\overline{H}_k]\Big| \\&\leq\, \mu(\secpar)\;.
	\end{split}
	\end{align}

	\end{enumerate}
\end{enumerate}
\end{definition}
\begin{theorem}[Informal]
Under the assumption that the Learning With Errors (LWE) problem with certain parameters is hard for $\BQP$, an NTCF family exists.	
\label{thm:LWE_implies_NTCF}
\end{theorem}
The hardness definition of LWE and the exact parameters required for the theorem above are given in~\cite[Theorem 26]{BCM+18}.

\section{The Advantage of Memoryless Money}
\label{sec:advantage_of_statelessness}
When discussing any form of quantum money, we must consider the motivation, i.e., the benefits over classical constructions --- for example, we could construct a rudimentary private classical money scheme in the following way: upon minting, the bank would produce a random serial number significantly long for some security parameter $\secpar$ and sign it using a MAC. The bank would maintain a database of all banknotes that have already been spent, and upon verification, after verifying the MAC tag of the banknote, the bank would search for its serial number within the database --- if it is not there, the verification succeeds and the serial number is added to the database, and if it is there the bank would know the money was already spent and thus verification will fail (of course, the bank would have to mint a new banknote for the user after a successful verification). Gavinsky~\cite{Gav12} discusses a similar notion.

This scheme is counterfeit-resistant according to our security definitions. However, it is \emph{memory-dependent} (also known as \emph{state-based}); i.e., the bank has to maintain a database to represent an ongoing state, remembering the banknotes that were spent. On its own, a memory-dependent protocol is not a terrible problem; many services maintain a database. This, however, becomes a liability when considering multiple branches of the same bank: a central database must maintain the shared state and synchronize the access to it (otherwise information would have to propagate between the branches, causing potential security breaches during the propagation time); this has a toll in terms of response time and communication.

Constructing a memory-dependent (classical) private money scheme is trivial --- the scheme above is an extremely simple example --- so such a construction is not particularly interesting. The case is different, however, in the public setting; constructing even a memory-dependent quantum money scheme that is publicly secure is challenging (and impossible to achieve classically), and thus such a construction is an interesting result.

\section{Parallel Repetition of Weakly Verifiable Puzzles}
\label{sec:parallel_repetition_weakly_verifiable_puzzles}
As seen in \cref{sec:parallel_repetition_theorem_for_1_of_2_puzzles}, our main tool for proving parallel repetition for 1-of-2 puzzles is the notion of weakly verifiable puzzles introduced by \cite{CHS05}. This section provides a brief informal overview of their parallel repetition proof for weakly verifiable puzzles (\cref{thm:parallel_repetition_weakly_verifiable_puzzles}).

The goal is to show a reduction from an algorithm $A$ (which in our case may be quantum) that solves $n$ weakly verifiable puzzles in parallel with probability at least $h^n +\negl$ for some non-negligible $h$ to an algorithm $A'$ solving a single puzzle with probability at least $h + \negl$. This shows that the best strategy for solving $n$ puzzles has the same probability (up to a negligible difference) of solving each puzzle separately.

Denote the event where $A$ solves the puzzle $p_i$ correctly by $S_i$, and denote the event where $A$ solves puzzles $p_k, \dots, p_n$ correctly by $R_k$. We call a puzzle coordinate $i \in [n]$ \textit{good} if it holds that $\underset{p_1,\dots,p_n}{\Pr}[S_i|R_{i+1}] \geq h + \negl$, i.e., if the probability that $A$ solved $p_i$ correctly conditioned on the probability that $A$ solved $p_{i+1}, \dots, p_{n}$ correctly is at least $h$, where the probability is taken over the randomness of $A$ as well as over the random choices of puzzles $p_1, \dots, p_n$. We make the following statistical observation:
\begin{align*}
    \Pr[R_1] &= \Pr[S_1 \wedge R_2] = \Pr[S_1|R_2] \cdot \Pr[R_2] \geq h^n + \negl\;.
\end{align*}
The last inequality holds because we know the probability of $A$ to solve all $n$ puzzles is at least $h^n$. Therefore, we conclude that at least one of the following must hold: either $\underset{p_1,\dots,p_n}{\Pr}[S_1|R_2] \geq h$, or $\underset{p_1,\dots,p_n}{\Pr}[R_2] \geq h^{n-1}$. Meaning either 1 is a good coordinate, or the probability that $A$ solves the last $n-1$ puzzles of its input is at least $h^{n-1}$.

If 1 is a good coordinate, then $A'$ solves its original puzzle in the following way: $A'$ runs $A$ with $(p, p_2, \dots, p_n)$ --- where $p$ is the original challenge puzzle and $p_2, \dots, p_n$ are randomly generated by $A'$ --- until puzzles $p_2, \dots, p_n$ were solved correctly\footnote{This can be verified since $A'$ generated those puzzles along with their verification information and therefore can verify them efficiently. Moreover, $A$ solves $n-1$ puzzles with non-negligible probability, meaning $A'$ can get correct solutions in polynomial time with high probability.}. Since 1 is a good coordinate, we know that if $A$ solved puzzles $p_2, \dots p_n$ correctly, the first puzzle is also solved correctly with probability at least $h$, as needed.

If 1 is not a good coordinate, than we know the probability of $A$ to solve $n-1$ puzzles is at least $h^{n-1}$. This provides us with a solver for $n-1$ puzzles with probability at least $h^{n-1} + \negl$, and the reduction continues by induction. Eventually, $A'$ will either reach a good coordinate $i$ and solve it with probability at least $h$ as described above, or arrive at a solver that solves a single puzzle --- the last puzzle --- with probability at least $h$.

To find out whether $i$ is a good coordinate, $A$ estimates $\underset{p_1,\dots,p_n}{\Pr}[S_i|R_{i+1}]$. This is done by generating $n$ puzzles, and running $A$ multiple times, all the while checking whether puzzles $p_{i+1}, \dots, p_n$ were solved correctly, and if so checking whether puzzle $p_i$ was also solved correctly. Since $\Pr[R_{i+1}] \geq \Pr[R_1] \geq h^n + \negl$, this estimation can be done efficiently with high probability.

The proof itself in \cite{CHS05} is more complex than the outline given here. For example, we cannot know for sure if a coordinate is a good coordinate, since we can only estimate the probabilities within some polynomial bounds --- the original proof deals with these issues.

Note that this parallel repetition is "perfect" --- i.e., it shows that the best strategy for solving multiple puzzles has the same probability as solving them separately (up to a negligible difference). This would not be the case if we had more than 2 rounds --- for 3 rounds the repetition would not be perfect, and for 4 rounds the soundness error does not decrease exponentially at all \cite{BIN97}. We can see that the proof above indeed would not work for more than 2 rounds: in our reduction, when we find a good coordinate $i$, we run the multiple-puzzle adversary multiple times until the puzzles $p_{i+1}, \dots, p_n$ are solved correctly, and then our original puzzle is solved correctly with probability $h$. But in the case of multiple rounds, the single-puzzle adversary is required to provide the multiple-puzzle adversary with messages from the challenger for all the puzzles, including the original one. The answers of the multiple-puzzles adversary for each puzzle can therefore depend on messages for all the puzzles, but we can run the protocol for the original puzzle with the challenger only once --- i.e., we cannot provide the challenger's answers for the original puzzles more than once. Therefore, we cannot condition on the adversary solving $p_{i+1}, \dots, p_n$ correctly, meaning that if we can run the adversary only once our success probability will not be $h$ as needed.

\section{Transaction Figures}
\label{sec:transaction_figs}
This section contains figures showing transactions can be made with each flavor of quantum money and the kind of communication they require (quantum or classical). Fat arrows $\quantumcom$ indicate quantum communication and thin arrows $\singleclassiccom$ indicate classical communication. A two-sided arrow indicates two-sided communication.

% Regular money figs
\begin{figure}
    \centering
    \subfloat[Payer has money in an account and a credit card]{
        \includegraphics[scale=\minifigscale]{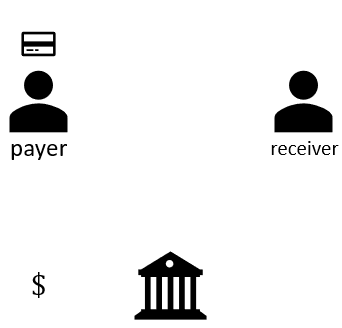}
    }
    \figspace
    \subfloat[Payer sends credit card details to receiver]{
        \includegraphics[scale=\minifigscale]{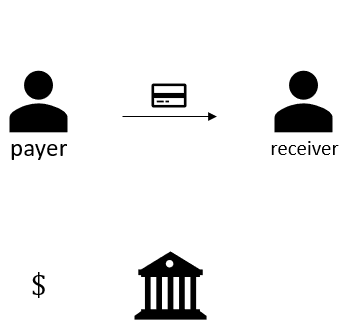}
    }
    \\
    \subfloat[Receiver receives credit card details]{
        \includegraphics[scale=\minifigscale]{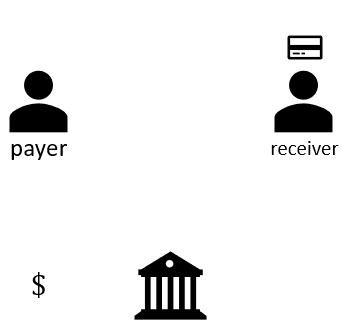}
    }
    \figspace
    \subfloat[Receiver sends credit card details to bank requesting to transfer money from payer]{
        \includegraphics[scale=\minifigscale]{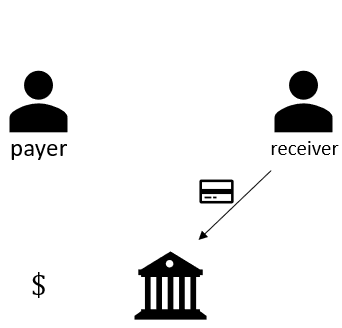}
    }
    \\
    \subfloat[Bank verifies credit card details]{
        \includegraphics[scale=\minifigscale]{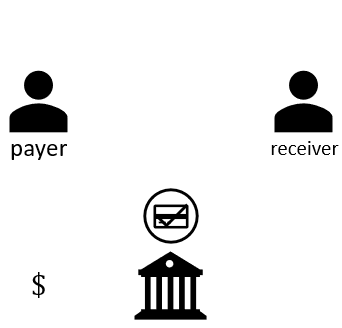}
    }
    \figspace
    \subfloat[Bank performs the transaction, deducting from payer's account and crediting receiver's account]{
        \includegraphics[scale=\minifigscale]{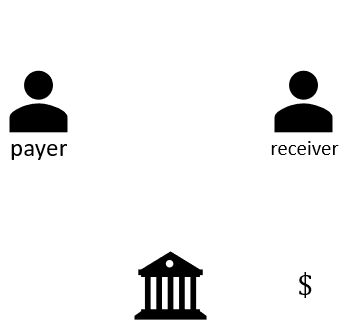}
    }
    \caption{Regular (non-quantum) money direct transaction. The communication required is one-way, though usually two-way communication is used to send confirmations for actions.}
    \label{fig:regular_transaction_direct}
    % Made using standard PowerPoint icons
\end{figure}

\begin{figure}
    \centering
    \subfloat[Payer has money in an account]{
        \includegraphics[scale=\minifigscale]{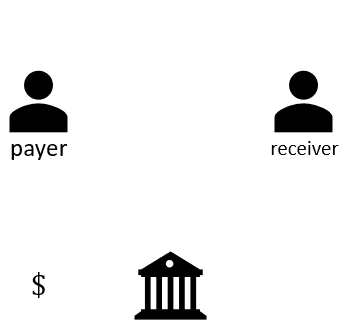}
    }
    \\
    \subfloat[Payer contacts bank requesting to perform a transaction and providing authentication]{
        \includegraphics[scale=\minifigscale]{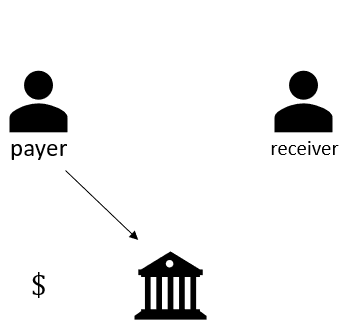}
    }
    \figspace
    \subfloat[Bank verifies payer's authentication information]{
        \includegraphics[scale=\minifigscale]{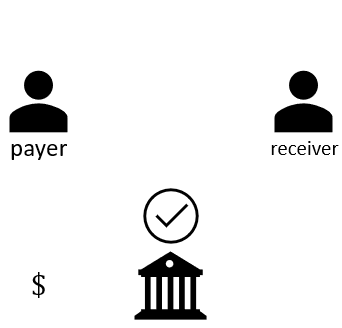}
    }
    \\
    \subfloat[Bank performs the transaction, deducting money from payer's account and crediting receiver's account]{
        \includegraphics[scale=\minifigscale]{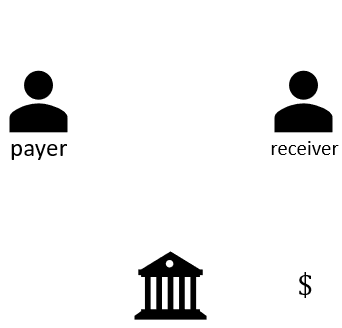}
    }
    \figspace
    \subfloat[Bank notifies receiver that the transaction took place]{
        \includegraphics[scale=\minifigscale]{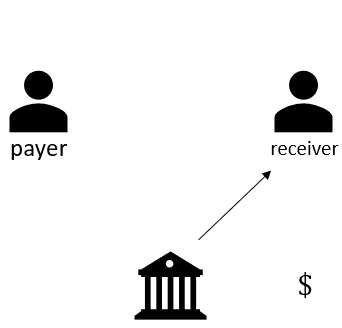}
    }
    \caption{Regular (non-quantum) money bank transaction. Though it is possible for communication between payer and bank to be one-way, usually the authentication process is an interactive protocol, and the payer receives confirmation from the bank, making communication two-way.}
    \label{fig:regular_transaction_bank}
    % Made using standard PowerPoint icons
\end{figure}

% Standard quantum money figs
\begin{figure}
    \centering
    \subfloat[Payer has a banknote]{
        \includegraphics[scale=\minifigscale]{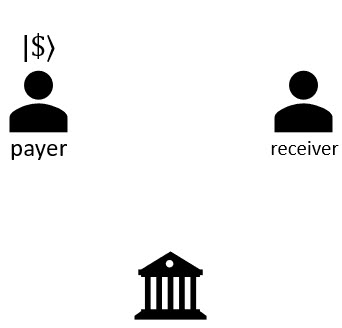}
    }
    \\
    \subfloat[Payer sends banknote to receiver]{
        \includegraphics[scale=\minifigscale]{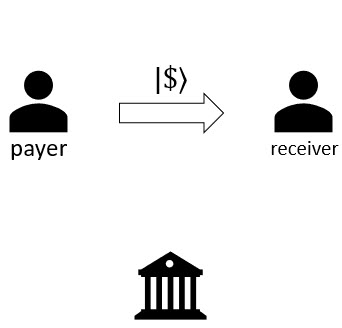}
    }
    \figspace
    \subfloat[Receiver receives banknote]{
        \includegraphics[scale=\minifigscale]{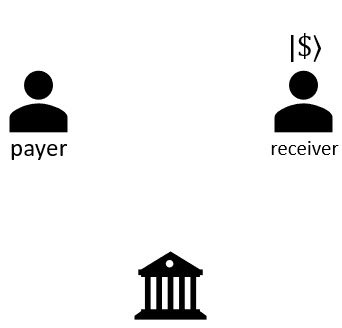}
    }
    \\
    \subfloat[Receiver sends banknote to bank]{
        \includegraphics[scale=\minifigscale]{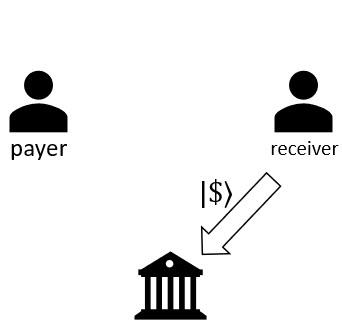}
    }
    \figspace
    \subfloat[Bank verifies banknote]{
        \includegraphics[scale=\minifigscale]{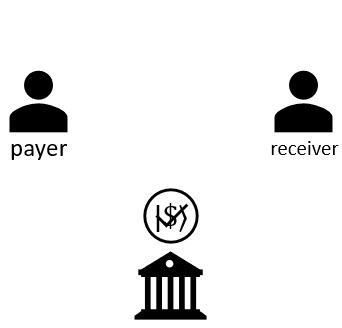}
    }
    \\
    \subfloat[Bank mints a new banknote for receiver]{
        \includegraphics[scale=\minifigscale]{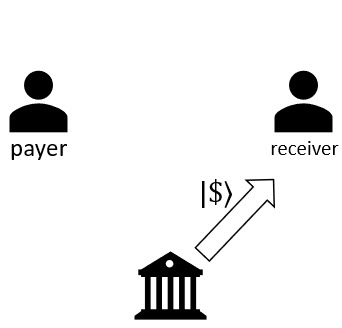}
    }
    \figspace
    \subfloat[Receiver was paid successfully]{
        \includegraphics[scale=\minifigscale]{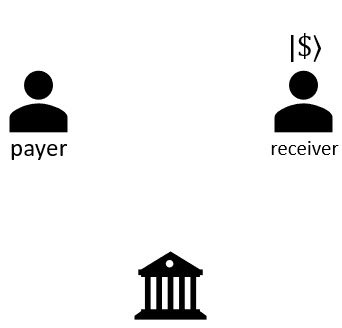}
    }
    \caption{Private standard direct transaction.}
    \label{fig:quantum_transaction_direct}
    % Made using standard PowerPoint icons
\end{figure}

\begin{figure}
    \centering
    \subfloat[Payer has a banknote]{
        \includegraphics[scale=\minifigscale]{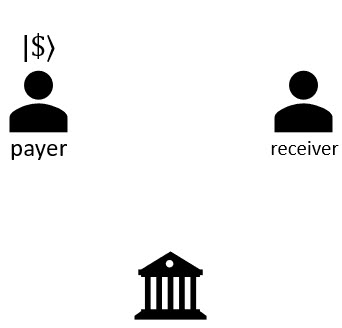}
    }
    \\
    \subfloat[Payer sends banknote to bank]{
        \includegraphics[scale=\minifigscale]{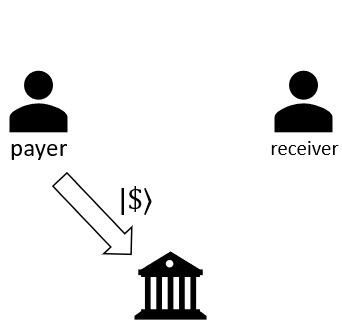}
    }
    \figspace
    \subfloat[Bank verifies banknote]{
        \includegraphics[scale=\minifigscale]{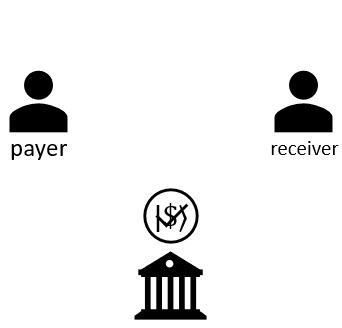}
    }
    \\
    \subfloat[Bank mints a new banknote for receiver]{
        \includegraphics[scale=\minifigscale]{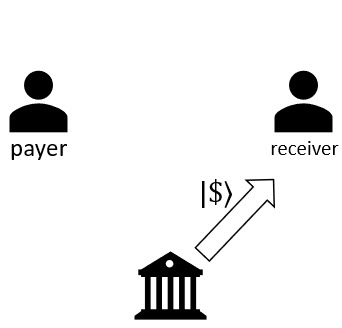}
    }
    \figspace
    \subfloat[Receiver was paid successfully]{
        \includegraphics[scale=\minifigscale]{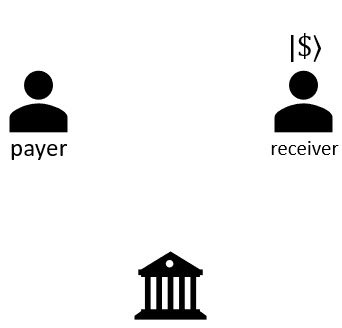}
    }
    \caption{Standard bank transaction.}
    \label{fig:quantum_transaction_bank}
    % Made using standard PowerPoint icons
\end{figure}

% Classically verifiable money figs
\begin{figure}
    \centering
    \subfloat[Payer has a banknote]{
        \includegraphics[scale=\minifigscale]{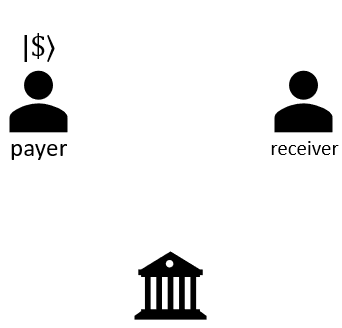}
    }
    \\
    \subfloat[Receiver initiates verification protocol with the bank and acts as a relay between payer and bank (destroying payer's banknote)]{
        \includegraphics[scale=\minifigscale]{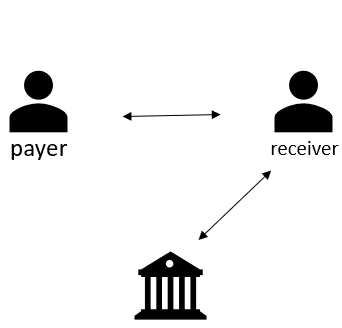}
    }
    \figspace
    \subfloat[Bank verifies protocol was executed correctly]{
        \includegraphics[scale=\minifigscale]{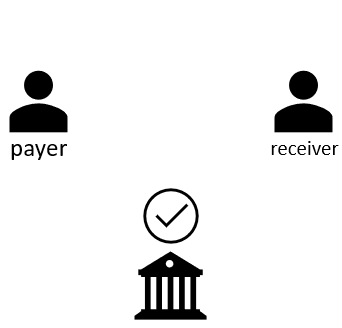}
    }
    \\
    \subfloat[Bank mints a new banknote for receiver]{
        \includegraphics[scale=\minifigscale]{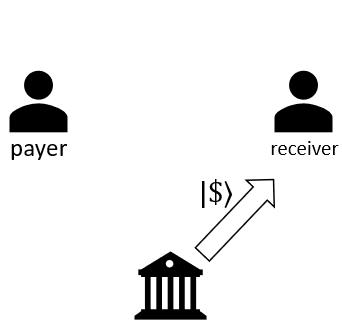}
    }
    \figspace
    \subfloat[Receiver was paid successfully]{
        \includegraphics[scale=\minifigscale]{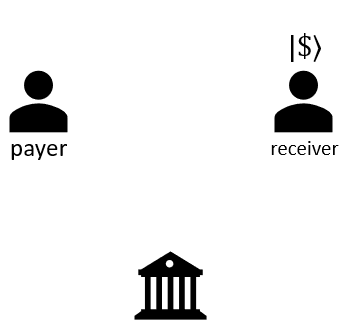}
    }
    \caption{Private classically verifiable direct transaction. In step (b), receiver acts as a relay between payer and bank and thus can be sure the banknote was valid.}
    \label{fig:classically_verifiable_transaction_direct}
    % Made using standard PowerPoint icons
\end{figure}

\begin{figure}
    \centering
    \subfloat[Payer has a banknote]{
        \includegraphics[scale=\minifigscale]{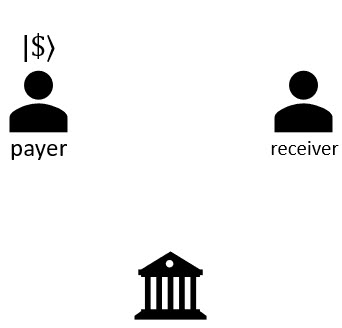}
    }
    \\
    \subfloat[Payer runs verification protocol with bank (losing the banknote)]{
        \includegraphics[scale=\minifigscale]{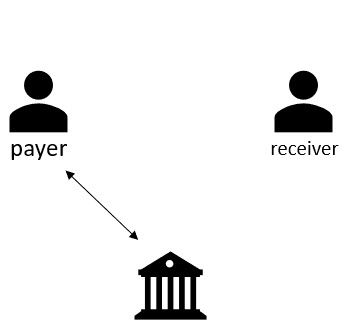}
    }
    \figspace
    \subfloat[Bank verifies protocol was executed successfully]{
        \includegraphics[scale=\minifigscale]{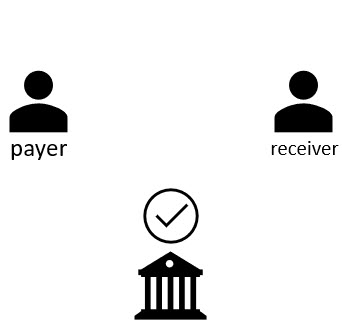}
    }
    \\
    \subfloat[Bank mints a new banknote for receiver]{
        \includegraphics[scale=\minifigscale]{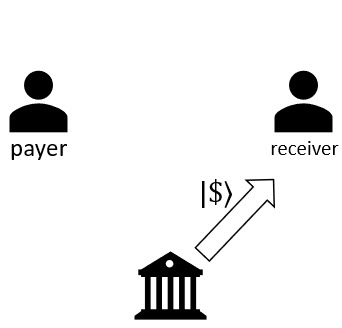}
    }
    \figspace
    \subfloat[Receiver was paid successfully]{
        \includegraphics[scale=\minifigscale]{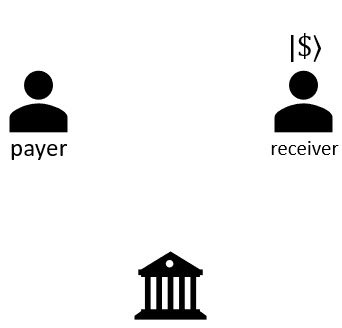}
    }
    \caption{Classically verifiable bank transaction.}
    \label{fig:classically_verifiable_transaction_bank}
    % Made using standard PowerPoint icons
\end{figure}

% Classically verifiable money with multiple verifications figs
\begin{figure}
    \centering
    \subfloat[Payer has a banknote]{
        \includegraphics[scale=\minifigscale]{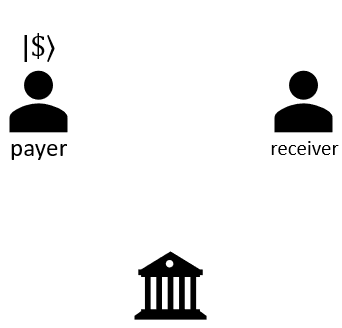}
    }
    \\
    \subfloat[Payer sends banknote to receiver]{
        \includegraphics[scale=\minifigscale]{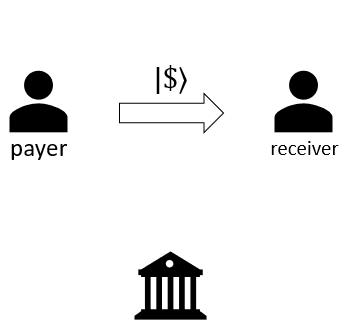}
    }
    \figspace
    \subfloat[Receiver receives banknote]{
        \includegraphics[scale=\minifigscale]{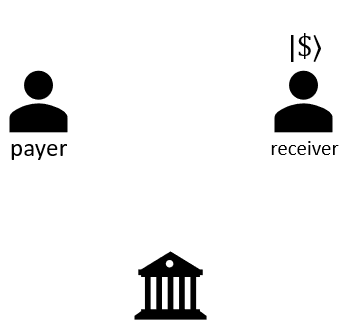}
    }
    \\
    \subfloat[Receiver runs verification protocol with bank (without losing the banknote)]{
        \includegraphics[scale=\minifigscale]{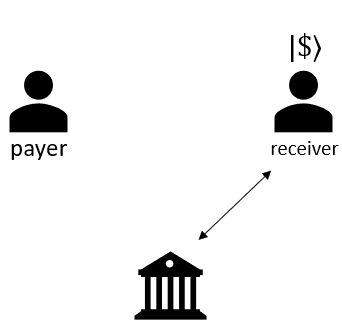}
    }
    \figspace
    \subfloat[Bank verifies protocol was executed successfully]{
        \includegraphics[scale=\minifigscale]{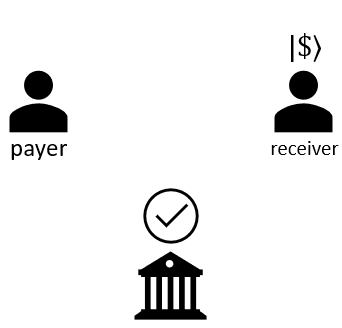}
    }
    \caption{Private classically verifiable direct transaction for a scheme which allows multiple verifications for the same banknote. Note that after a finite number of verifications the banknote is destroyed and communication with the bank is required, like in \cref{fig:classically_verifiable_transaction_direct}.}
    \label{fig:classically_verifiable_MV_transaction_direct}
    % Made using standard PowerPoint icons
\end{figure}

\begin{figure}
    \centering
    \subfloat[Payer has a banknote]{
        \includegraphics[scale=\minifigscale]{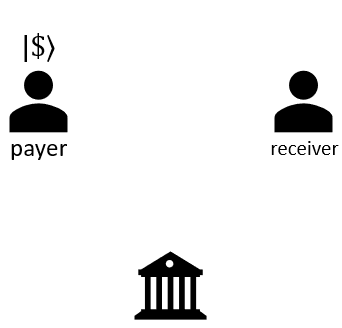}
    }
    \\
    \subfloat[Payer sends banknote to bank]{
        \includegraphics[scale=\minifigscale]{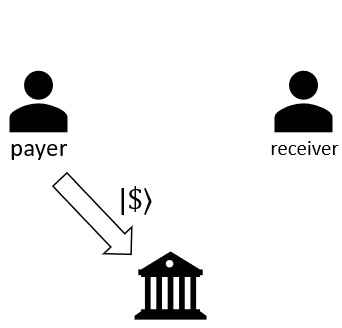}
    }
    \figspace
    \subfloat[Bank verifies banknote]{
        \includegraphics[scale=\minifigscale]{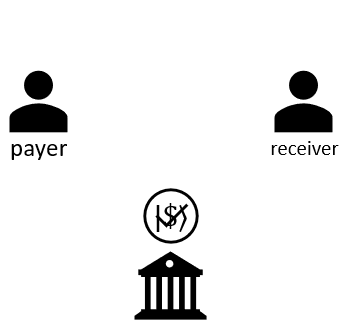}
    }
    \\
    \subfloat[Bank mints a new banknote for receiver]{
        \includegraphics[scale=\minifigscale]{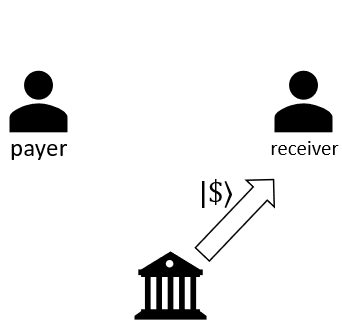}
    }
    \figspace
    \subfloat[Receiver was paid successfully]{
        \includegraphics[scale=\minifigscale]{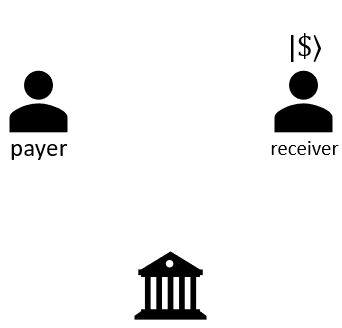}
    }
    \caption{Private classically verifiable bank transaction for a scheme which allows multiple verifications for the same banknote. For some such schemes it could be possible to have a classical verification that ensures the banknote was destroyed, in which case a transaction through the bank could be executed like in \cref{fig:classical_minting_transaction_bank}.}
    \label{fig:classically_verifiable_MV_transaction_bank}
    % Made using standard PowerPoint icons
\end{figure}

% Classical minting figs
\begin{figure}
    \centering
    \subfloat[Payer has a banknote]{
        \includegraphics[scale=\minifigscale]{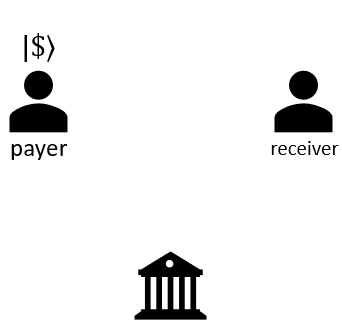}
    }
    \\
    \subfloat[Payer sends banknote to receiver]{
        \includegraphics[scale=\minifigscale]{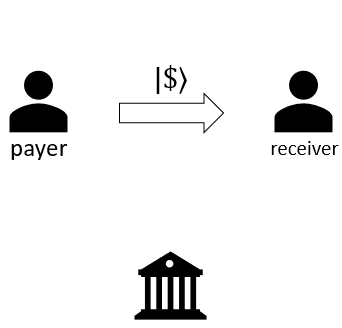}
    }
    \figspace
    \subfloat[Receiver receives banknote]{
        \includegraphics[scale=\minifigscale]{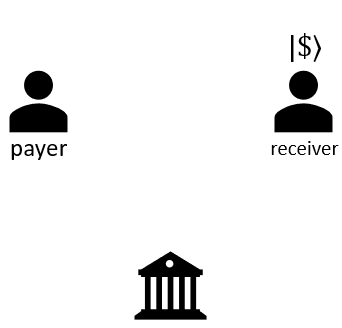}
    }
    \\
    \subfloat[Receiver sends banknote to bank]{
        \includegraphics[scale=\minifigscale]{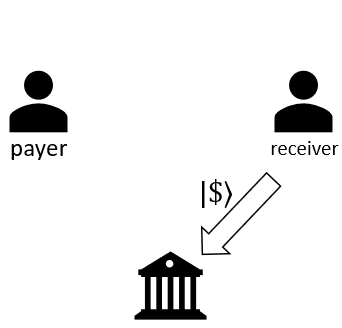}
    }
    \figspace
    \subfloat[Bank verifies banknote]{
        \includegraphics[scale=\minifigscale]{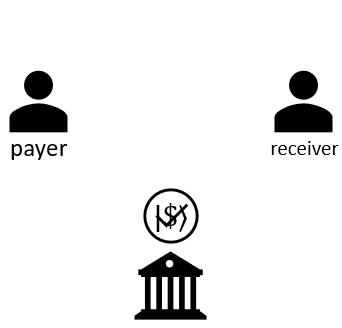}
    }
    \\
    \subfloat[Bank mints a new banknote for receiver]{
        \includegraphics[scale=\minifigscale]{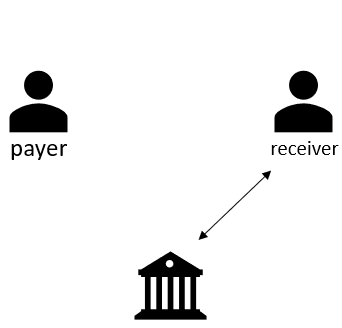}
    }
    \figspace
    \subfloat[Receiver was paid successfully]{
        \includegraphics[scale=\minifigscale]{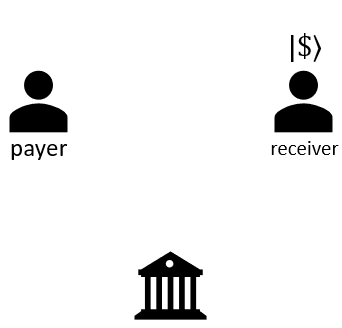}
    }
    \caption{Private classical minting direct transaction.}
    \label{fig:classical_minting_transaction_direct}
    % Made using standard PowerPoint icons
\end{figure}

\begin{figure}
    \centering
    \subfloat[Payer has a banknote]{
        \includegraphics[scale=\minifigscale]{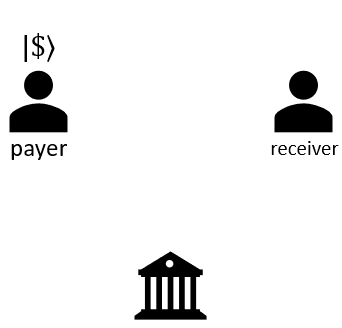}
    }
    \\
    \subfloat[Payer sends banknote to bank]{
        \includegraphics[scale=\minifigscale]{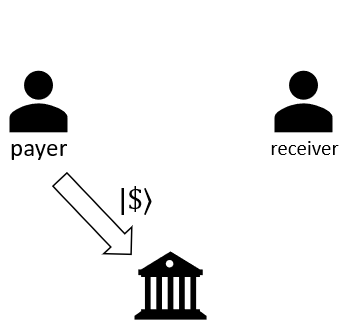}
    }
    \figspace
    \subfloat[Bank verifies banknote]{
        \includegraphics[scale=\minifigscale]{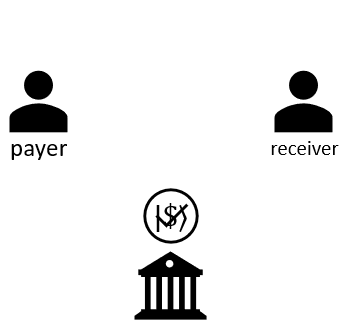}
    }
    \\
    \subfloat[Bank mints a new banknote for receiver]{
        \includegraphics[scale=\minifigscale]{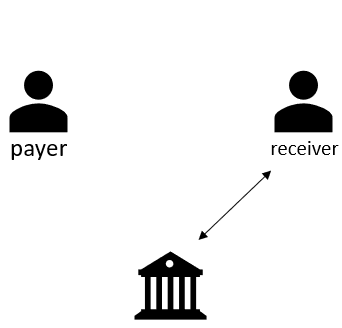}
    }
    \figspace
    \subfloat[Receiver was paid successfully]{
        \includegraphics[scale=\minifigscale]{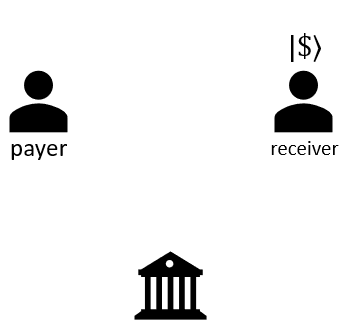}
    }
    \caption{Classical minting bank transaction.}
    \label{fig:classical_minting_transaction_bank}
    % Made using standard PowerPoint icons
\end{figure}

% Semi-quantum money figs
\begin{figure}
    \centering
    \subfloat[Payer has a banknote]{
        \includegraphics[scale=\minifigscale]{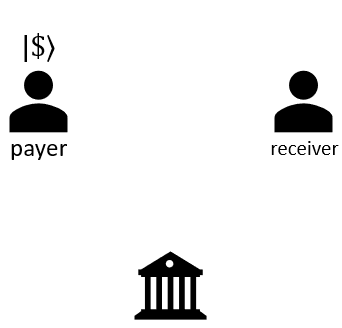}
    }
    \\
    \subfloat[Receiver initiates verification protocol with the bank and acts as a relay between payer and bank (destroying payer's banknote)]{
        \includegraphics[scale=\minifigscale]{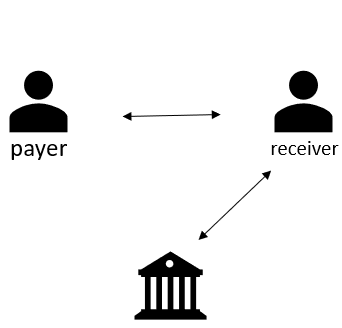}
    }
    \figspace
    \subfloat[Bank verifies protocol was executed correctly]{
        \includegraphics[scale=\minifigscale]{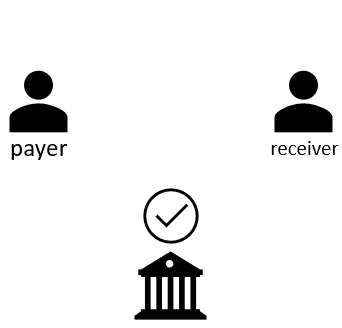}
    }
    \\
    \subfloat[Bank mints a new banknote for receiver]{
        \includegraphics[scale=\minifigscale]{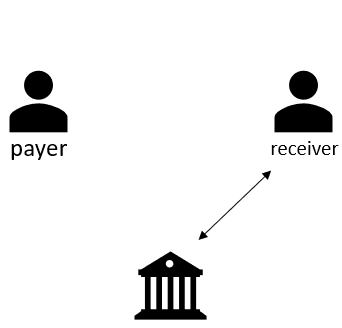}
    }
    \figspace
    \subfloat[Receiver was paid successfully]{
        \includegraphics[scale=\minifigscale]{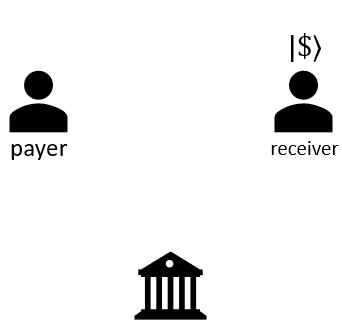}
    }
    \caption{Private semi-quantum direct transaction. In step (b), receiver acts as a relay between payer and bank and thus can be sure the banknote was valid.}
    \label{fig:semi_quantum_transaction_direct}
    % Made using standard PowerPoint icons
\end{figure}

\begin{figure}
    \centering
    \subfloat[Payer has a banknote]{
        \includegraphics[scale=\minifigscale]{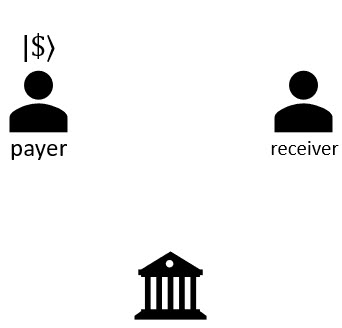}
    }
    \\
    \subfloat[Payer runs verification protocol with bank (losing the banknote)]{
        \includegraphics[scale=\minifigscale]{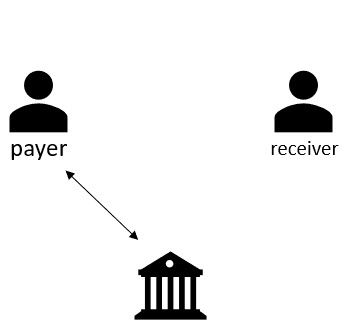}
    }
    \figspace
    \subfloat[Bank verifies protocol was executed successfully]{
        \includegraphics[scale=\minifigscale]{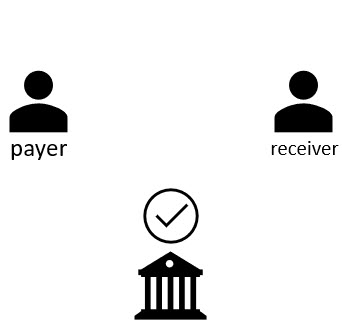}
    }
    \\
    \subfloat[Bank mints a new banknote for receiver]{
        \includegraphics[scale=\minifigscale]{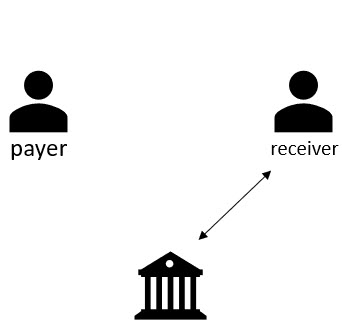}
    }
    \figspace
    \subfloat[Receiver was paid successfully]{
        \includegraphics[scale=\minifigscale]{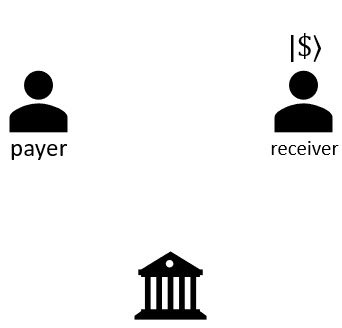}
    }
    \caption{Semi-quantum bank transaction.}
    \label{fig:semi_quantum_transaction_bank}
    % Made using standard PowerPoint icons
\end{figure}

% Public quantum money figs
\begin{figure}
    \centering
    \subfloat[Payer has a banknote]{
        \includegraphics[scale=\minifigscale]{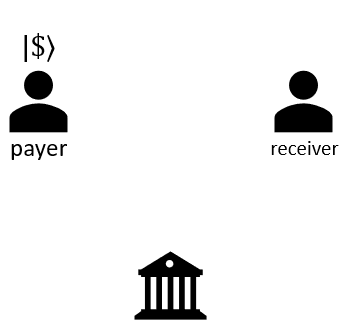}
    }
    \figspace
    \subfloat[Payer sends banknote to receiver]{
        \includegraphics[scale=\minifigscale]{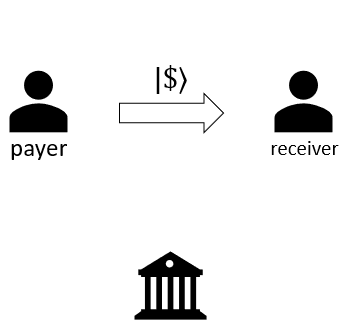}
    }
    \\
    \subfloat[Receiver receives banknote]{
        \includegraphics[scale=\minifigscale]{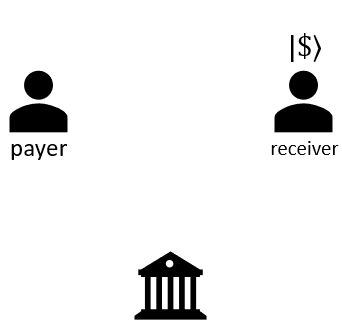}
    }
    \figspace
    \subfloat[Receiver validates banknote --- receiver was paid successfully]{
        \includegraphics[scale=\minifigscale]{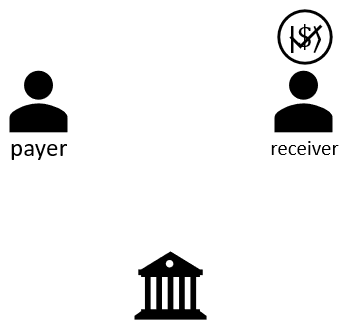}
    }
    \caption{Public quantum money direct transaction.}
    \label{fig:public_transaction_direct}
    % Made using standard PowerPoint icons
\end{figure}

\begin{figure}
    \centering
    \subfloat[Payer has a banknote]{
        \includegraphics[scale=\minifigscale]{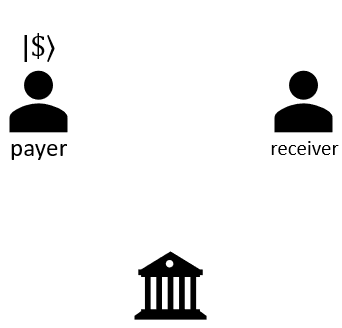}
    }
    \figspace
    \subfloat[Payer runs transaction protocol with receiver]{
        \includegraphics[scale=\minifigscale]{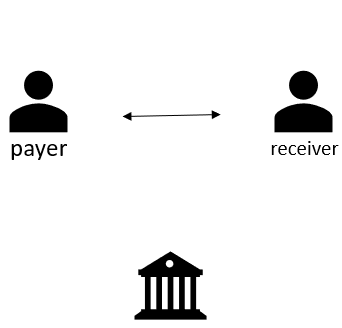}
    }
    \\
    \subfloat[Receiver receives banknote]{
        \includegraphics[scale=\minifigscale]{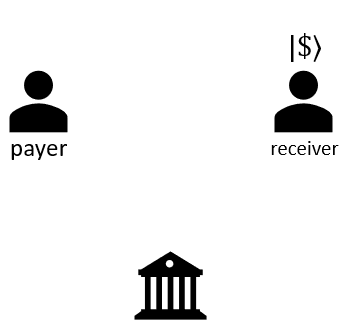}
    }
    \figspace
    \subfloat[Receiver validates banknote --- receiver was paid successfully]{
        \includegraphics[scale=\minifigscale]{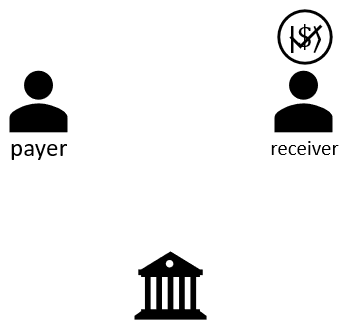}
    }
    \caption{One-shot signature direct transaction.}
    \label{fig:one_shot_transaction_direct}
    % Made using standard PowerPoint icons
\end{figure}

\fi

\ifnum\draft=1
\input{future_work_appendices.tex}
\fi

\ifnum\masterthesis=1
    \mybib
    \includepdf[pages={1,2}]{Hebrew_Abstract.pdf}
\fi
 
\ifnum\cryptology=1{
    \pagebreak
    \footnotesize
    \mybib}
\fi

\end{document}